\documentclass[aps,prx,twocolumn,reprint, superscriptaddress]{revtex4-2}

\usepackage{graphicx}
\usepackage{dcolumn}
\usepackage{bm}
\usepackage{amsmath, amssymb, mathtools, braket, mathrsfs}
\usepackage{color, tcolorbox}
\usepackage{comment,footnote}
\allowdisplaybreaks 

\usepackage{algorithmic}
\usepackage{algorithm}

\usepackage{tikz}
\usetikzlibrary {arrows.meta}
\usetikzlibrary {bending}

\usepackage{thmtools}
\usepackage{amsthm}
\usepackage{thm-restate}

\theoremstyle{definition}
\newtheorem{dfn}{Definition}
\newtheorem{con}{Condition}

\theoremstyle{plain}
\newtheorem{thm}{Theorem}
\newtheorem{lem}{Lemma}

\usepackage{hyperref}
\usepackage{xcolor}
\hypersetup{
	colorlinks=true,
	citecolor=blue,
	linkcolor=blue,
	urlcolor=blue,
}

\newcommand{\mO}{\mathcal{O}}

\newcommand{\tr}{\text{Tr}}
\newcommand{\mcA}{\mathcal{A}}
\newcommand{\mcB}{\mathcal{B}}
\newcommand{\mcC}{\mathcal{C}}
\newcommand{\mcS}{\mathcal{S}}

\newcommand{\mcL}{\mathcal{L}}
\newcommand{\mcP}{\mathcal{P}}
\newcommand{\mcQ}{\mathcal{Q}}

\newcommand{\mc}{\mathcal}
\newcommand{\mf}{\mathfrak}
\newcommand{\mfg}{\mathfrak{g}}
\newcommand{\grad}{\Gamma}

\newcommand{\mcGc}{\mathcal{G}}
\newcommand{\mcG}{\mathcal{G}_{\rm Lie}}
\newcommand{\mcPn}{\mathcal{P}_n}

\newcommand{\LB}{L_{\rm fin}}
\newcommand{\Lbulk}{L_{\rm ini}}

\newcommand{\bt}{\bm{\theta}}

\newcommand{\ex}{\mathcal{X}_{\rm exp}}
\newcommand{\ef}{\mathcal{F}_{\rm eff}}

\newcommand{\x}{x}
\newcommand{\y}{y}

\newcommand{\conn}{ {\stackrel{\mfg}{\longleftrightarrow}} }

\usepackage{titlesec}
\titleformat*{\section}{\raggedright\large\bfseries\sffamily}
\titleformat*{\subsection}{\raggedright\bfseries\sffamily}
\titleformat*{\subsubsection}{\raggedright\bfseries\sffamily}

\begin{document}
	
	\title{Trade-off between Gradient Measurement Efficiency and Expressivity \\ in Deep Quantum Neural Networks}
	
	\author{Koki Chinzei}
	\thanks{chinzei.koki@fujitsu.com}
	\affiliation{Quantum Laboratory, Fujitsu Research, Fujitsu Limited, 4-1-1 Kawasaki, Kanagawa 211-8588, Japan}
	
	\author{Shinichiro Yamano}
	\thanks{yamano@qi.t.u-tokyo.ac.jp \\ These authors contributed equally to this work.}
	\affiliation{Quantum Laboratory, Fujitsu Research, Fujitsu Limited, 4-1-1 Kawasaki, Kanagawa 211-8588, Japan}
	\affiliation{Department of Applied Physics, Graduate School of Engineering, The University of Tokyo, 7-3-1 Hongo Bunkyo-ku, Tokyo 113-8656, Japan}
	
	\author{Quoc Hoan Tran}
	\affiliation{Quantum Laboratory, Fujitsu Research, Fujitsu Limited, 4-1-1 Kawasaki, Kanagawa 211-8588, Japan}
	
	\author{Yasuhiro Endo}
	\affiliation{Quantum Laboratory, Fujitsu Research, Fujitsu Limited, 4-1-1 Kawasaki, Kanagawa 211-8588, Japan}
	
	\author{Hirotaka Oshima}
	\affiliation{Quantum Laboratory, Fujitsu Research, Fujitsu Limited, 4-1-1 Kawasaki, Kanagawa 211-8588, Japan}

	\date{\today}
	
	\begin{abstract}
		
		Quantum neural networks (QNNs) require an efficient training algorithm to achieve practical quantum advantages.
		A promising approach is gradient-based optimization, where gradients are estimated by quantum measurements.
		However, QNNs currently lack general quantum algorithms for efficiently measuring gradients, which limits their scalability.
		To elucidate the fundamental limits and potentials of efficient gradient estimation, we rigorously prove a trade-off between gradient measurement efficiency (the mean number of simultaneously measurable gradient components) and expressivity in deep QNNs. 
		This trade-off indicates that more expressive QNNs require higher measurement costs per parameter for gradient estimation, while reducing QNN expressivity to suit a given task can increase gradient measurement efficiency. 
		We further propose a general QNN ansatz called the stabilizer-logical product ansatz (SLPA), which achieves the trade-off upper bound by exploiting the symmetric structure of the quantum circuit.
		Numerical experiments show that the SLPA drastically reduces the sample complexity needed for training while maintaining accuracy and trainability compared to well-designed circuits based on the parameter-shift method.
		
	\end{abstract}

	\maketitle

	
	\begin{figure*}[t]
		\centering
		\includegraphics[width=\linewidth]{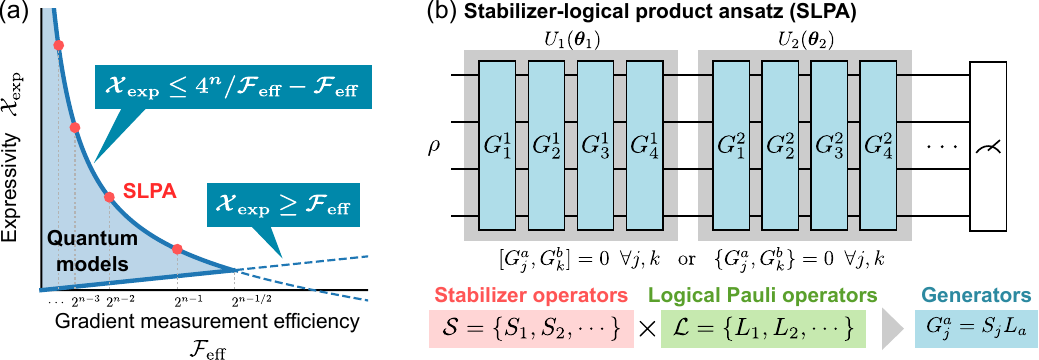}
		\caption{
			{\bf Overview of two main results.}
			(a) A trade-off relation between gradient measurement efficiency and expressivity.
			Any quantum model can only exist in the blue region.
			Gradient measurement efficiency, $\ef$, is defined as the number of simultaneously measurable components in the gradient, and expressivity $\ex$ is the dimension of the dynamical Lie algebra of the parameterized quantum circuit.
			The red circles denote the SLPA, where $2^s$ gradient components can be measured simultaneously (i.e., $\ef=2^s$, $s$ is an integer), reaching the upper bound of the trade-off inequality.
			(b) The circuit structure of SLPA.
			The generators of SLPA are constructed by taking the products of stabilizers and logical Pauli operators.
		}
		\label{fig: overview}
	\end{figure*}

	\noindent
	Deep learning is a breakthrough technology that has a significant impact on a variety of fields~\cite{Hinton2006-lq}.
	In deep learning, neural networks are trained to approximate unknown target functions and represent input-output relationships in various tasks, such as image recognition~\cite{Krizhevsky2012-ya}, natural language processing~\cite{Vaswani2017-ya}, predictions of protein structure~\cite{Jumper2021-ch}, and quantum state approximation~\cite{Carleo2017-nw}.
	Behind the great success of deep learning lies an efficient training algorithm equipped with backpropagation~\cite{Rumelhart1986-qt, Baydin2017-qv}.
	Backpropagation allows us to efficiently evaluate the gradient of an objective function at approximately the same computational cost as a single model evaluation, optimizing the neural network based on the evaluated gradient.
	This technique is essential for training large-scale neural networks with enormous parameters, including large language models~\cite{Vaswani2017-ya}.

	In analogy to deep learning, variational quantum algorithms (VQAs) have emerged as a promising technology for solving classically intractable problems, e.g., in quantum chemistry and physics~\cite{Peruzzo2014-oj}, optimization~\cite{Farhi2014-as}, and machine learning~\cite{Farhi2018-nt, Liu2018-bd, Mitarai2018-ap, Benedetti2019-ph, Schuld2020-vz}.
	In VQAs, a parameterized quantum circuit, called a quantum neural network (QNN) in the context of quantum machine learning (QML)~\cite{Wiebe2014-un, Schuld2015-su, Biamonte2017-fz, Dunjko2018-jv}, is optimized by minimizing an objective function with quantum and classical computers~\cite{Cerezo2021-un}.
	As in classical neural networks, gradient-based optimization algorithms are effective for training QNNs, where gradients are estimated by quantum measurements.
	However, unlike classical models, efficient gradient estimation is challenging in QNNs because quantum states collapse upon measurement~\cite{Gilyen2019-qb}.
	In fact, general QNNs lack an efficient gradient measurement algorithm that achieves the same computational cost scaling as classical backpropagation when only one copy of the quantum data is accessible at a time~\cite{Abbas2023-hy}.
	Instead, the gradient is usually estimated using the parameter-shift method~\cite{Mitarai2018-ap, Schuld2019-rr}, where each gradient component is measured independently.
	This method is easy to implement but leads to a high measurement cost that is proportional to the number of parameters, which prevents QNNs from scaling up in near-term quantum devices without sufficient computational resources.

	Despite the lack of general algorithms to achieve backpropagation scaling, a commuting block circuit (CBC) has recently been proposed as a well-structured QNN for efficient gradient estimation~\cite{Bowles2023-vf, Coyle2024-ff}.
	The CBC consists of $B$ blocks containing multiple variational rotation gates, and the generators of rotation gates in two blocks are either all commutative or all anti-commutative.
	This specific structure allows us to estimate the gradient using only $2B-1$ types of quantum measurements, which is independent of the number of rotation gates in each block, potentially achieving the backpropagation scaling.
	Therefore, the CBC can be trained more efficiently than conventional QNNs based on the parameter-shift method.

	Inspired by these developments, several key questions about gradient measurement in QNNs arise. 
	What factor fundamentally limits gradient measurement efficiency in QNNs? 
	More specifically, what is the theoretically most efficient QNN model for measuring gradients? 
	The specific structure of CBC suggests that there is a trade-off for its high gradient measurement efficiency, but this remains unclear. 
	Answering these questions is crucial not only for the theoretical understanding of efficient training in QNNs but also for the realization of practical QML.

	\newpage
	This work answers these questions for general QNNs.
	We first formulate gradient measurement efficiency in terms of the simultaneous measurability of gradient components, proving a general trade-off between gradient measurement efficiency and expressivity in a wide class of deep QNNs [Fig.~\ref{fig: overview} (a)].
	This trade-off implies that more expressive QNNs require higher measurement costs per parameter for gradient estimation, whereas the gradient measurement efficiency can be increased by reducing the QNN expressivity to suit a given task.
	Based on this trade-off, we further propose a general ansatz of CBC, the stabilizer-logical product ansatz (SLPA), as an optimal model for gradient estimation.
	This ansatz is inspired by the stabilizer code in quantum error correction~\cite{Gottesman1997-uo}, exploiting the symmetric structure of the circuit to enhance gradient measurement efficiency [Fig.~\ref{fig: overview} (b)].
	Remarkably, the SLPA can reach the upper bound of the trade-off inequality, i.e., it enables gradient estimation with theoretically the fewest types of quantum measurements for a given expressivity.
	Owing to its symmetric structure, the SLPA can be applied to various problems involving symmetry, which are common in quantum chemistry, physics, and machine learning.
	As a demonstration, we consider the task of learning an unknown symmetric function.
	Numerical experiments show that the SLPA can drastically reduce the number of data samples needed for training without sacrificing accuracy and trainability compared to several QNNs based on the parameter-shift method.
	These results reveal the theoretical limits and possibilities of efficient training in parameterized quantum circuits, thus paving the way for realizing practical quantum advantages in various research fields related to VQAs, including QML.

	\section*{Results}
	\vspace{-0.42cm}
	\subsection*{Model, efficiency, and expressivity} \label{sec: problem setting}
	\vspace{-0.42cm}
	
	\noindent
	We consider the following QNN on an $n$-qubit system:
	\begin{align}
		U(\bt) = \prod_{j=1}^L \exp(i\theta_j G_j), \label{eq: model}
	\end{align}
	where $G_j$ is a Pauli operator in $\mcPn=\{I, X, Y, Z\}^{\otimes n}$, $\bt = (\theta_1,\theta_2,\cdots)$ are variational rotation angles, and $L$ is the number of rotation gates.
	Let $\mcGc=\{G_j\}_{j=1}^L$ be the set of generators. 
	We also consider a cost function defined as the expectation value of a Pauli observable $O$:
	\begin{align}
		C(\bt) = \tr\left[ \rho U^\dag(\bt) O U(\bt) \right],
	\end{align}
	where $\rho$ is an input quantum state.

	To reveal the theoretical limit of efficient gradient estimation, let us first define the simultaneous measurability of two different components in the gradient.
	The derivative of the cost function by $\theta_j$ is written as
	\begin{align}
		&\partial_j C(\bt) = \tr \left[ \rho \grad_j(\bt) \right],  \label{eq: Cost gradient}
	\end{align}
	where we have introduced a {\it gradient operator} $\grad_j$: 
	\begin{align}
		&\grad_j(\bt) = \partial_j \left[ U^\dag(\bt)OU(\bt) \right]. \label{eq: gradj}
	\end{align}
	In other words, $\partial_j C$ is identical to the expectation value of the gradient operator $\grad_j$ for the input state $\rho$.
	In quantum mechanics, two observables $A$ and $B$ can be simultaneously measured for any input state $\rho$ if and only if $[A,B]=0$.
	Therefore, we define the simultaneous measurability of two gradient components $\partial_j C$ and $\partial_k C$ as $[\grad_j(\bt),\grad_k(\bt)]=0$ for all $\bt$.

	Based on this definition, we partition the gradient operators $\{\grad_j\}_{j=1}^L$ into $M_L$ simultaneously measurable sets.
	In this partitioning, all operators in the same set are simultaneously measurable in the sense that they are mutually commuting.
	This partitioning enables us to estimate all gradient components using $M_L$ types of quantum measurements in principle.
	Thereby, we define gradient measurement efficiency for finite-depth QNNs as 
	\begin{align}
		\ef^{(L)} = \frac{L}{\text{min}(M_L)},
	\end{align}
	where $\text{min}(M_L)$ is the minimum number of sets among all possible partitions.
	We also define the gradient measurement efficiency in the deep circuit limit:
	\begin{align}
		\ef = \lim_{L\to\infty} \ef^{(L)}.
	\end{align}
	In this definition, $\ef$ indicates the mean number of simultaneously measurable components in the gradient.
	Therefore, the larger $\ef$ is, the more efficiently the gradient can be measured.
	This work mainly focuses on the deep circuit limit, but our results also have implications for efficient gradient measurement in finite-depth QNNs.

	Another key factor in this work is expressivity, which is crucial in understanding quantum circuits and realizing universal computation.
	In particular, whether a given class of quantum circuits can express all unitaries, i.e., whether it is universal, has been investigated in various contexts, such as fault-tolerant quantum computation~\cite{Nielsen2010-pf, DiVincenzo1995-qe, Lloyd1995-uk, Barenco1995-nx}, quantum dynamics~\cite{Lloyd2018-ca, Marvian2022-ir, Morales2020-tv}, random quantum circuits~\cite{Gross2007-ts, Dankert2005-dr, Dankert2009-rj, Ji2018-gk}, and VQAs~\cite{Sim2019-am, Larocca2022-so, Zheng2023-ze}, with and without symmetry constraints.
	Here, to quantify expressivity, we employ the dynamical Lie algebra (DLA)~\cite{Albertini2001-xk, Zeier2011-bq, D-Alessandro2007-kk}.
	To this end, we consider the Lie closure $i \mcG= \braket{i\mcGc}_\text{Lie}$, which is defined as the set of Pauli operators obtained by repeatedly taking the nested commutator between the circuit generators in $i\mcGc=\{iG_j\}_{j=1}^L$.
	The DLA is defined from $\mcG$ as
	\begin{align}
		\mfg = \text{span}\left( \mcG \right),
	\end{align}
	which is the subspace of $\mf{su}(2^n)$.
	The DLA characterizes which unitaries the QNN can express in the overparameterized regime of $L\gtrsim \text{dim}(\mfg)$.
	That is, $U(\bt) \in e^{\mfg}$ for all $\bt$~\cite{Larocca2023-ll}.
	Therefore, we define the QNN expressivity in the deep circuit limit as the dimension of the DLA:
	\begin{align}
		\ex = \text{dim}(\mfg).
	\end{align}
	For example, a hardware-efficient ansatz consisting of local Pauli rotations $\mcGc=\{X_j, Y_j, Z_jZ_{j+1}\}_{j=1}^n$ has the DLA of $\ex=4^n-1$~\cite{Larocca2022-so}, indicating that this ansatz is universal and can express all unitaries in the deep circuit limit.
	Note that $\ex\leq 4^n-1$ as $\mfg$ is a subset of $\mf{su}(2^n)$.
	
	\subsection*{Efficiency-expressivity trade-off} \label{sec: main theorem}
	\vspace{-0.42cm}
	
	\noindent
	Remarkably, the upper bound of gradient measurement efficiency depends on expressivity.
	Here, we present the main theorem on the relationship between gradient measurement efficiency and expressivity (the proof is provided in Methods and Supplementary Sections~\ref{secap: App proof}--\ref{secap: proof eqs}):
	\begin{thm}[Informal] \label{thm: main_informal}
		In deep QNNs, gradient measurement efficiency and expressivity obey the following inequalities:
		\begin{align}
			&\ex \leq \frac{4^n}{\ef} - \ef, \label{eq: main_inequality} 
		\end{align} 
		and
		\begin{align}
			&\ex \geq \ef. \label{eq: main_inequality2}
		\end{align} 
	\end{thm}

	The inequality~\eqref{eq: main_inequality} represents a trade-off relation between gradient measurement efficiency and expressivity as shown in Fig.~\ref{fig: overview} (a).
	This trade-off indicates that more expressive QNNs require higher measurement costs per parameter for gradient estimation.
	For example, when a QNN has maximum expressivity $\ex=4^n-1$ (e.g., a hardware-efficient ansatz with $\mcGc=\{X_j, Y_j, Z_jZ_{j+1}\}_{j=1}^n$), the gradient measurement efficiency must be $\ef=1$, which implies that two or more gradient components cannot be measured simultaneously in this model.

	The trade-off inequality~\eqref{eq: main_inequality} also suggests that we can increase gradient measurement efficiency by reducing the expressivity to fit a given problem, i.e., by encoding prior knowledge of the problem into the QNN as an inductive bias.
	In the context of QML, such a problem-tailored model is considered crucial to achieving quantum advantages~\cite{Kubler2021-xk}.
	For instance, an equivariant QNN, where the symmetry of a problem is encoded into the QNN structure, is one of the promising problem-tailored quantum models and exhibits high trainability and generalization by reducing the parameter space to search~\cite{Bronstein2021-is, Verdon2019-pi, Zheng2023-ze, Larocca2022-mj, Meyer2023-vx, Skolik2023-lq, Ragone2022-va, Nguyen2024-nt, Schatzki2024-il, Sauvage2024-vd, Chinzei2024-nm, Chinzei2024-tl}.
	Later, we will introduce a general QNN ansatz called SLPA that leverages prior knowledge about symmetry to reach the upper bound of the trade-off inequality~\eqref{eq: main_inequality}.

	On the other hand, the inequality~\eqref{eq: main_inequality2} shows that gradient measurement efficiency is bounded by expressivity.
	This bound is achieved, for example, in the commuting generator circuit~\cite{Bowles2023-vf}, where all generators, $G_1, \cdots, G_L$, are mutually commuting.
	This circuit allows us to measure all gradient components simultaneously, implying $\ef=L$.
	Also, since all generators are commutative, the expressivity is $\ex=L$ by definition.
	Consequently, the commuting generator circuit achieves the bound of the inequality $\ex\geq \ef$.
	As a result of this inequality, a deeply overparameterized QNN with $L\gg \ex$ parameters requires $L/\ef \gg \ex/\ef \geq 1$ types of quantum measurements to estimate the gradient, leading to high measurement costs for training.
	Furthermore, combining inequalities~\eqref{eq: main_inequality} and \eqref{eq: main_inequality2}, we have $\ef\leq 2^{n-1/2}$, which is the upper bound of gradient measurement efficiency in QNNs.

	We remark on the relationship between gradient measurement efficiency and trainability.
	In deep QNNs, it is known that the variance of the cost function in the parameter space decays inversely proportional to $\text{dim}(\mfg)$: $\text{Var}_{\bt}[C(\bt)] \sim 1/\text{dim}(\mfg)=1/\ex$~\cite{Larocca2022-so, Ragone2024-hl, Fontana2024-ky}.
	Thus, exponentially high expressivity results in an exponentially flat landscape of the cost function (so-called the barren plateaus~\cite{McClean2018-qf}), which makes efficient training of QNNs impossible.
	In other words, expressivity has a trade-off with trainability.
	By contrast, high gradient measurement efficiency and high trainability are compatible.
	From the trade-off inequality~\eqref{eq: main_inequality} and $\text{Var}_{\bt}[C(\bt)] \sim 1/\ex$, we have $\ef \lesssim 4^n \text{Var}_{\bt}[C(\bt)]$.
	This inequality means that we can increase $\ef$ and $\text{Var}_{\bt}[C(\bt)]$ simultaneously, indicating the compatibility of high gradient measurement efficiency and high trainability.

	\subsection*{Stabilizer-logical product ansatz} \label{sec: Most efficient model}
	\vspace{-0.42cm}
	
	\noindent
	While Theorem~\ref{thm: main_informal} reveals the general trade-off relation between gradient measurement efficiency and expressivity, designing optimal quantum models that reach the upper bound of the trade-off inequality is a separate issue for realizing practical QML.
	Here, we propose a general ansatz of CBC called SLPA, which is constructed from stabilizers and logical Pauli operators to reach the upper bound of the trade-off (see Methods for the details of CBC).
	This ansatz is also insightful in understanding the trade-off relation from a symmetry perspective.

	The SLPA is inspired by the stabilizer code in quantum error correction~\cite{Gottesman1997-uo}.
	The stabilizer code is characterized by a stabilizer group $\mcS=\{S_1,S_2,\cdots\}\subset \mcPn$, which satisfies
	\begin{align}
		[S_j, S_k] &= 0 \quad \forall j,k, \label{eq: S and S}
	\end{align}
	where $S_j\in\mcS$ is a Pauli operator.
	Given that $s$ is the number of independent stabilizers in $\mcS$, the order of $\mcS$ is given by $|\mcS|=2^s$, where an arbitrary element of $\mcS$ is written as a product of the $s$ independent stabilizers.
	For this stabilizer group, we consider logical Pauli operators $\mcL=\{L_1,L_2,\cdots\}\subset \mcPn$, which commute with the stabilizers:
	\begin{align}
		[S_j, L_a] &= 0 \quad \forall j,a, \label{eq: S and L}
	\end{align}
	where $L_a\in\mcL$ is a Pauli operator.
	In quantum error correction, the stabilizer code encodes quantum information into $n-s$ logical qubits in the code space (the Hilbert subspace in which the eigenvalues of the stabilizers are all $+1$), detecting errors that flip the eigenvalues of the stabilizers.
	The logical operators correspond to operations on the encoded logical qubits.
	This stabilizer code establishes the foundation of fault-tolerant quantum computation based on promising quantum error-correcting codes, including the surface code~\cite{Kitaev2003-vs, Horsman2012-bu} and more general quantum low-density parity-check codes~\cite{Breuckmann2021-ug}.

	The SLPA leverages the algebraic structure of the stabilizer code to improve gradient measurement efficiency without focusing on fault tolerance.
	Suppose that we are given stabilizers $\mcS$ and logical Pauli operators $\mcL$.
	The SLPA comprises block unitaries containing multiple variational rotation gates, as shown in Fig.~\ref{fig: overview} (b).
	To construct the SLPA, we take the product of the stabilizer $S_j$ and the logical Pauli operator $L_a$, defining the $j$th generator of the $a$th block as
	\begin{align}
		G_j^a = S_jL_a.
		\label{eq: CBC construction}
	\end{align}
	Using these generators, we construct the SLPA
	\begin{align}
		U_\text{SLPA}(\bt)=\prod_a U_a(\bt_a) \label{eq: SLPA unitary}
	\end{align}
	with block unitaries
	\begin{align}
		U_a(\bt_a) = \prod_j \exp(i\theta_j^a G_j^a), \label{eq: SLPA unitary2}
	\end{align}
	where each block can contain at most $|\mcS|=2^s$ variational rotation gates.
	The multi-Pauli rotation gates in the circuit can be implemented using, for example, a single-Pauli rotation and Clifford gates~\cite{Gottesman1997-uo}.

	This SLPA enables efficient gradient estimation due to the following commutation relations between the generators $G_j^a$.
	First, the generators within the same block are all commutative:
	\begin{align}
		[G^a_j, G_k^a] = 0 \,\,\, \forall j,k. 
	\end{align}
	Second, the generators in any two distinct blocks are either all commutative or all anti-commutative:
	\begin{align}
		&[G^a_j, G_k^b] = 0 \,\,\, \forall j,k, \quad\text{or}\quad \{G^a_j, G_k^b\} = 0 \,\,\, \forall j,k. 
	\end{align}
	These commutation relations satisfy the requirements of CBC, ensuring that all gradient components of each block can be measured using only two types of quantum circuits.
	Specifically, the gradient can be measured efficiently using the linear combination of unitaries with an ancilla qubit (Supplementary Section~\ref{secap: CBC measure}).
	Conversely, any CBC can be formulated using a stabilizer group and logical Pauli operators, implying that the SLPA is a general ansatz of CBC (Supplementary Section~\ref{secap: stabilizer formalism}).
	This stabilizer formalism also allows us to understand a necessary condition for the backpropagation scaling of CBC in general (Supplementary Section~\ref{secap: backprop SLPA}).

	It is noteworthy that the SLPA has symmetry $\mcS$: 
	\begin{align}
		[U_\text{SLPA},\mcS]=0,
	\end{align}
	which is derived from $[G_j^a,\mcS]=[S_jL_a, \mcS]=0$.
	As shown in the following numerical experiment section, leveraging this property allows us to solve problems with symmetry $\mcS$ efficiently and accurately.
	Furthermore, when the stabilizers $\mcS$ commute with the observable $O$, we can measure all gradient components within a block simultaneously (i.e., only one type of quantum circuit is sufficient to measure all gradient components of each block) because the generators of the block are either all commutative or all anti-commutative with $O$: $[G_j^a,O]=0$ $\forall j$ or $\{G_j^a,O\}=0$ $\forall j$ (see also Methods).

	Let us revisit the efficiency-expressivity trade-off from the perspective of SLPA.
	We first consider the gradient measurement efficiency of SLPA.
	When the stabilizers $\mcS$ commute with the observable $O$, all gradient components of each block can be measured simultaneously, as discussed above.
	Therefore, given that each block can contain at most $|\mcS|=2^s$ rotation gates, the gradient measurement efficiency, which is the mean number of simultaneously measurable gradient components, obeys
	\begin{align}
		\ef \leq 2^s.
		\label{eq: CBC eff}
	\end{align}
	The equality holds if all blocks of the SLPA have $2^s$ rotation gates.

	The stabilizers limit expressivity instead of providing high gradient measurement efficiency.
	Since all the generators of the SLPA commute with the stabilizers (i.e., $[\mcGc,\mcS]=0$), the DLA generated by them also commute with the stabilizers (i.e., $[\mfg,\mcS]=0$).
	Then, the dimension of the subspace in $\mathfrak{su}(2^n)$ stabilized by $\mcS$ (the centralizer of $\mcS$) is $4^n/|\mcS|=4^n/2^s$.
	Furthermore, given that $[\mfg,\mcS]=[O,\mcS]=0$, we can assume that the stabilizers $\mcS$ are not included in the generators $\mcGc$ and thus the DLA $\mfg$ because rotation gates generated by stabilizers, $e^{i\theta S_j}$, never affect the result of the cost function.
	This consideration provides the upper bound of the DLA dimension, namely the expressivity:
	\begin{align}
		\ex \leq \frac{4^n}{2^s} - 2^s,
		\label{eq: CBC exp}
	\end{align}
	where the equality holds if the DLA covers the entire subspace stabilized by $\mcS$ except for $\mcS$ (i.e., the ansatz is universal under the symmetry constraint).
	Specifically, given $k=n-s$ logical qubits for the stabilizer group $\mcS$, if $\mcL$ contains all $4^k -1$ logical Pauli operators except for the identity, the generators $\mcGc=\{S_jL_a\}_{j,a}$ are sufficient to reach the upper bound of the above inequality as $\ex=(4^k-1)2^s =4^n/2^s - 2^s$.
	For example, in the toric code on the torus with $k=2$~\cite{Kitaev2003-vs}, two logical $X$ and $Z$ operators, which correspond to the two types of $X$ and $Z$ chains that wrap around the torus in the two directions, and their products (logical Pauli strings) are sufficient for universal ansatz.
	We note that $4^k -1$ logical Pauli operators can be systematically obtained for any stabilizer group by considering the standard form of the check matrix~\cite{Nielsen2010-pf}.

	Combining Eqs.~\eqref{eq: CBC eff} and \eqref{eq: CBC exp}, we obtain the trade-off inequality~\eqref{eq: main_inequality} for the SLPA, $\ex\leq 4^n/\ef-\ef$.
	Then, the upper bound of the inequality is attained if the numbers of rotation gates in all blocks are $|\mcS|=2^s$ and the DLA covers the entire subspace under the symmetry constraint $\mcS$, indicating that the SLPA can be optimal for gradient estimation.
	This result also describes the trade-off relation between gradient measurement efficiency and expressivity from the perspective of symmetry; we can measure several gradient components simultaneously by exploiting the symmetric structure of the quantum circuit, but this instead constrains expressivity.

	\subsection*{Numerical demonstration: learning a symmetric function} \label{sec: demo}
	\vspace{-0.42cm}

	\begin{figure*}[t]
		\centering
		\includegraphics[width=\linewidth]{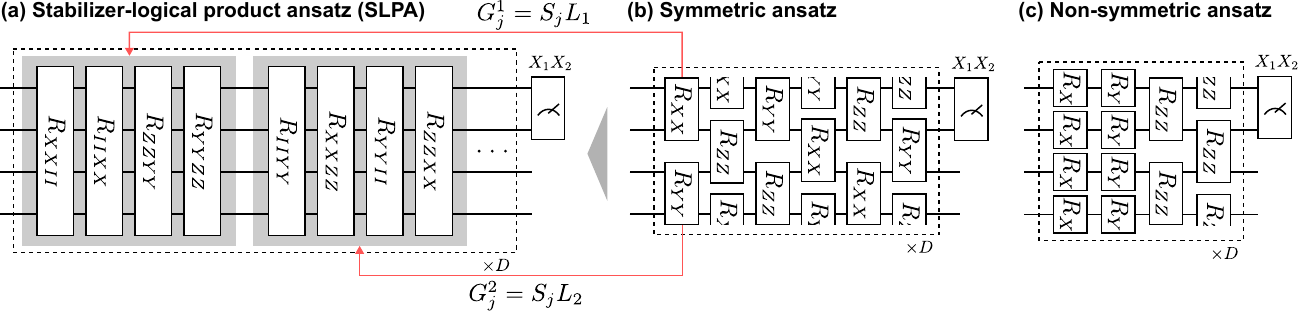}
		\caption{
			{\bf Quantum circuits used in numerical experiments.}
			The SLPA is constructed from the symmetric ansatz by taking the products of the stabilizers and the generators of the symmetric ansatz.
		}
		\label{fig: demo_circ}
	\end{figure*}

	\begin{figure*}[t]
		\centering
		\includegraphics[width=\linewidth]{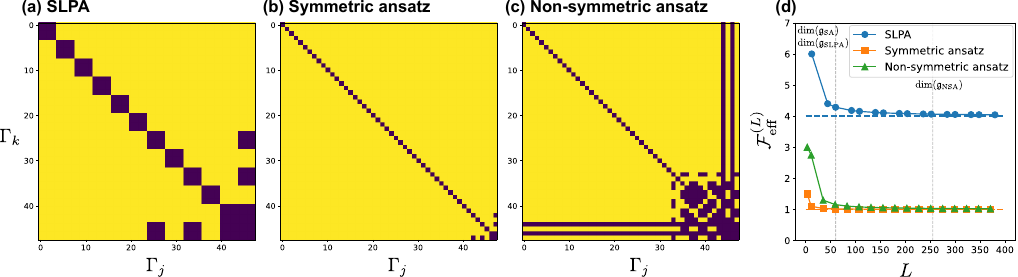}
		\caption{
			{\bf Gradient measurement efficiency.}
			(a)--(c) Commutators between two gradient operators $\grad_j(\bt)$ and $\grad_k(\bt)$ in the SLPA and symmetric and non-symmetric ansatzes for $n=4$ qubits and $L=48$ parameters.
			The black and yellow regions represent $[\grad_j(\bt),\grad_k(\bt)]=0$ and $[\grad_j(\bt),\grad_k(\bt)]\neq 0$ for random $\bt$, respectively.
			(d) Changes in gradient measurement efficiency when the number of parameters $L$ is varied for $n=4$.
			Their values are computed by minimizing the number of simultaneously measurable sets of $\grad_j(\bt)$'s for random $\bt$.
			The blue circles, orange squares, and green triangles are the results of SLPA and symmetric and non-symmetric ansatzes, approaching four and one in the limit of $L\to\infty$, respectively.
			The dashed gray lines represent the DLA dimension of each model, $\text{dim}(\mfg)$.
		}
		\label{fig: commutator}
	\end{figure*}

	\noindent
	Due to its symmetric structure, the SLPA is particularly effective in problems involving symmetry.
	Such problems are common in quantum chemistry, physics, and machine learning, which are the main targets of quantum computing.
	As a demonstration, we consider the task of learning an unknown symmetric function $f:\rho \mapsto y$, where $\rho$ is an input quantum state, $y$ is a real scalar, and $f(\rho)$ is assumed to be linear with respect to $\rho$.
	Here, suppose we know in advance that $f(\rho)=f(S_j\rho S_j^\dag)$ holds for $\forall S_j\in\mcS$, where $\mcS$ is a stabilizer group.
	This type of learning task is standard in the context of geometric deep learning~\cite{Bronstein2021-is} and has broad applications, including molecular dynamics~\cite{Batzner2022-xy}, electronic structure calculations~\cite{Gong2023-fw}, and computer vision~\cite{Cohen2016-mv}.
	To learn the unknown function, we use a quantum model $h_{\bt}(\rho)=\tr[U(\bt)\rho U^\dag(\bt) O]$ to approximate $f(\rho)$.
	For the symmetric function, an equivariant QNN is effective in achieving high accuracy, trainability, and generalization~\cite{Bronstein2021-is, Verdon2019-pi, Zheng2023-ze, Larocca2022-mj, Meyer2023-vx, Skolik2023-lq, Ragone2022-va, Nguyen2024-nt, Schatzki2024-il, Sauvage2024-vd, Chinzei2024-nm, Chinzei2024-tl}.
	The equivariant QNN consists of an $\mcS$-symmetric circuit $U(\bt)$ and an $\mcS$-symmetric observable $O$ (i.e., $[U(\bt),\mcS]=[O,\mcS]=0$), satisfying the same invariance as the target function $h_{\bt}(\rho)=h_{\bt}(S_j\rho S_j^\dag)$.
	The SLPA can be viewed as an equivariant QNN due to its symmetry, allowing us to improve accuracy, trainability, and generalization as well as gradient measurement efficiency.

	For concreteness, we consider a stabilizer group
	\begin{align}
		\mcS=\left\{I,\prod_{j=1}^nX_j,\prod_{j=1}^nY_j,\prod_{j=1}^nZ_j \right\}
	\end{align}
	with even $n$, where the number of independent operators in $\mcS$ is $s=2$.
	Suppose that the target symmetric function is given by $f(\rho)=\tr[U_\text{tag}\rho U_\text{tag}^\dag O]$ with an $\mcS$-symmetric random unitary $U_\text{tag}$ and the $\mcS$-symmetric observable $O$ (i.e., $[U_\text{tag},\mcS]=[O,\mcS]=0$).
	Here, we set $O=X_1X_2$.
	To learn this function, we use $N_t$ training data $\{\ket{x_i},y_i\}_{i=1}^{N_t}$, where $\ket{x_i}=\bigotimes_{j=1}^n \ket{s_{i}^j}$ and $y_i=f(\ket{x_i})$ are the product state of single qubit Haar-random states and its label, respectively.
	The model is optimized by minimizing the mean squared error loss function.
	We also prepare $N_t$ test data for validation, which are sampled independently of training data. 
	Below, we set $N_t=50$.

	We employ three types of variational quantum circuits,  SLPA and local symmetric and non-symmetric ansatzes, to learn the function $f(\rho)$ via $h_{\bt}(\rho)=\tr[U(\bt)\rho U^\dag(\bt) O]$ (see Fig.~\ref{fig: demo_circ} and Methods for detailed descriptions).
	First, the generators of the local symmetric ansatz (SA) are given by $\mcGc_\text{SA}=\{X_j X_{j+1},Y_j Y_{j+1},Z_j Z_{j+1}\}_{j=1}^n$, which commutes with the stabilizers $\mcS$.
	The DLA generated by $\mcGc_\text{SA}$ covers the entire subspace stabilized by $\mcS$, except for $\mcS$, indicating that $\ex=4^n/4 - 4$ (Supplementary Section~\ref{secap: DLA of our ansatz}).
	Then, although the upper bound of the gradient measurement efficiency is $\ef=4$ by Theorem~\ref{thm: main_informal}, this symmetric ansatz cannot reach it but instead exhibits $\ef=1$, as shown later.
	We can construct the SLPA from this symmetric ansatz via Eqs.~\eqref{eq: CBC construction}--\eqref{eq: SLPA unitary2} by regarding the generator of the symmetric ansatz as the logical Pauli operator $L_a$.
	In other words, we construct a block of the SLPA from each rotation gate of the symmetric ansatz by taking the product of the generator $L_a$ and the stabilizers $\mcS$, where the generators are given by $\mcGc_\text{SLPA}=\mcGc_\text{SA}\times\mcS$.
	The DLA dimension of SLPA is the same as the symmetric ansatz, $\ex=4^n/4 - 4$.
	Meanwhile, given that the generators of each block are either all commutative or all anti-commutative with the observable $O$ (which is derived by $[O,\mcS]=0$), we can measure all four gradient components of each block simultaneously.
	Therefore, in contrast to the symmetric ansatz, this SLPA can reach the upper bound of the gradient measurement efficiency $\ef=4$.
	For comparison, we also consider the local non-symmetric ansatz (NSA), where the generators $\mcGc_\text{NSA}=\{X_j,Y_j,Z_jZ_{j+1}\}_{j=1}^n$ do not commute with the stabilizers $\mcS$ and lead to the maximum expressivity in the full Hilbert space $\ex=4^n-1$~\cite{Larocca2022-so}.
	Hence, the gradient measurement efficiency of this model must be $\ef=1$ by Theorem~\ref{thm: main_informal}.
	In gradient estimation, we use the parameter-shift method for the symmetric and non-symmetric ansatzes and the linear combination of unitaries for the SLPA, where $1000$ measurement shots are used per circuit.
	The numerical experiments in this work use Qulacs, an open-source quantum circuit simulator~\cite{Suzuki2021-uh}.

	\begin{figure}[t]
		\centering
		\includegraphics[width=\linewidth]{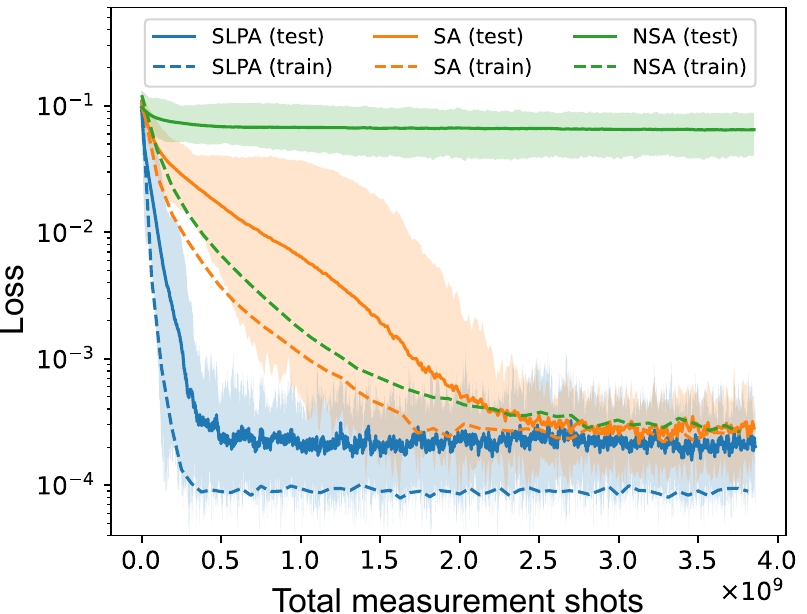}
		\caption{
			{\bf Changes in losses during training.}
			The horizontal axis is the cumulative number of measurement shots.
			The blue, orange, and green solid (dashed) lines represent the test (training) losses for the SLPA, symmetric ansatz (SA), and non-symmetric ansatz (NSA), respectively.
			The shaded areas are the maximum and minimum of the test losses for 20 sets of random initial parameters.
			The numbers of qubits and parameters are $n=4$ and $L=96$.
		}
		\label{fig: main result}
	\end{figure}

	First of all, we investigate the gradient measurement efficiency of these three models.
	In Figs.~\ref{fig: commutator} (a)--(c), we numerically compute the commutation relations between all pairs of the gradient operators for random $\bt$, where the black and yellow regions indicate $[\grad_j(\bt),\grad_k(\bt)]=0$ and $[\grad_j(\bt),\grad_k(\bt)]\neq0$, respectively.
	While most pairs of gradient operators are not commutative in the symmetric and non-symmetric ansatzes, we observe a $4\times4$ block structure in the SLPA.
	This means that the SLPA allows us to measure four gradient components simultaneously, implying the high efficiency of the SLPA.
	Figure~\ref{fig: commutator} (d) shows the gradient measurement efficiency $\ef^{(L)}$ for finite-depth circuits with varying the number of parameters $L$.
	In the symmetric and non-symmetric ansatzes, the efficiency decreases as the circuit becomes deeper and converges to $\ef=1$ after overparameterized.
	The SLPA shows a similar behavior, but the efficiency is larger than those of the other two models for the entire region of $L$, eventually approaching the theoretical upper bound of $\ef=4$ in the deep circuit limit.
	These results demonstrate the validity of our trade-off relation and the high gradient measurement efficiency of the SLPA in both shallow and deep circuits.
	In Supplementary Section~\ref{secap: disentangled circuit}, we also investigate another type of quantum circuit, which has low expressivity despite the absence of stabilizer-type symmetry, supporting the validity of the trade-off.

	The high gradient measurement efficiency of SLPA can reduce the number of data samples needed for training.
	In Fig.~\ref{fig: main result}, the SLPA shows significantly faster convergence of loss functions than the other models in terms of the cumulative number of measurement shots, i.e., training samples.
	This fast convergence stems from the high gradient measurement efficiency.
	As discussed above, the SLPA has $\ef=4$ in the deep circuit limit, which is four times greater than the other two models with $\ef=1$.
	Furthermore, the parameter-shift method used in the symmetric and non-symmetric ansatzes requires twice the number of circuits for estimating a gradient component than the linear combination of unitaries used in the SLPA.
	Thus, the number of measurement shots per epoch for the SLPA is one-eighth that of the other models in total, leading to a drastic reduction in the sample complexity for training.

	Besides high training efficiency, the SLPA can learn the target function with high accuracy (i.e., low test loss).
	As shown in Fig.~\ref{fig: main result}, the test loss of SLPA after training is comparable to that of the symmetric ansatz, implying high generalization performance of SLPA for this problem.
	This is due to encoding the symmetry of the target function into the SLPA as prior knowledge, similar to the symmetric ansatz.
	In the overparameterized SLPA and symmetric ansatz, the same DLA results in similar generalization performance.
	In contrast, the non-symmetric ansatz sufficiently reduces the training loss but does not reduce the test loss, indicating its low generalization performance.
	These results show the importance of symmetry encoding in this problem.

	We highlight the trainability of SLPA.
	A potential concern is that the multi-Pauli rotations within the SLPA, which can be global, could lead to a barren plateau~\cite{McClean2018-qf} even in shallow circuits.
	However, this concern can be ignored; the global operators do not directly induce a barren plateau in the SLPA. 
	Rather, our model is less prone to the barren plateau phenomenon than the symmetric ansatz when the circuit depth is $L/n\sim \mO(\log(n))$, where the SLPA exhibits a larger variance in the cost function compared to the symmetric ansatz.
	We discuss the high trainability of SLPA in more detail in Supplementary Section~\ref{secap: barren plateau}.

	Finally, while this section has investigated the symmetric function learning task, the SLPA can be applied to other types of problems associated with symmetry.
	In Supplementary Section~\ref{secap: QPR}, we tackle a quantum phase recognition task as an example of learning problems without an explicit symmetric target function, demonstrating the high training efficiency and accuracy of SLPA.
	In this task, the data distribution (not the target function) is invariant under the action of symmetry.
	This demonstration suggests the broad applicability of SLPA beyond the symmetric function learning.

	\subsection*{Beyond stabilizer-logical product ansatz} \label{sec: remark on SLPA}
	\vspace{-0.42cm}
	
	\noindent
	Here, we discuss some further potentials of the SLPA beyond the results presented in this work.

	One intriguing potential is to extend the SLPA to encompass more general symmetries beyond the stabilizer group. 
	These general symmetries, which include both Abelian and non-Abelian groups, can lead to more complex and exotic structures within the Hilbert space, introducing a challenge for extending our theory to accommodate such broader situations.
	To this end, there are several possible scenarios in which the SLPA can exploit other types of symmetries.
	First, when a group characterizing symmetry contains a stabilizer group as its subgroup, we can exploit the stabilizer group to construct an SLPA. 
	While this SLPA cannot fully incorporate symmetry into the circuit, it allows for efficient gradient estimation without imposing unnecessary symmetry constraints.
	Second, eliminating assumptions from our theory could break through the limitations of SLPA.
	For instance, we could broaden the types of symmetries applicable in the SLPA by utilizing general generators $G_j$ and general observable $O$ (though this work is limited to Pauli operators) or by performing classical post-processing of the outputs from multiple SLPAs~\cite{Huang2023-ie, Coyle2024-ff}.
	Understanding the efficiency-expressivity trade-off in such situations and elucidating the fundamental limit of efficient gradient estimation for QNNs with general symmetries remains an important open problem in realizing efficient QML.

	Another potential is applying the SLPA to problems that lack symmetry.
	Regardless of the high gradient measurement efficiency of SLPA, the symmetry constraint may result in low accuracy in some problems without symmetry.
	To address this challenge, we can eliminate the symmetry constraint and reinforce the capability of SLPA by combining it with another parameterized quantum circuit.
	For instance, we consider $U(\bt)=U_{\rm SLPA}(\bt_2)U_{\rm SB}(\bt_1)$, where $U_{\rm SLPA}$ and $U_{\rm SB}$ are an SLPA and a non-symmetric parameterized circuit breaking the symmetry of the SLPA, respectively.
	In this circuit, the gradient of $\bt_2$ can be measured efficiently due to the commutation relations of the SLPA, whereas the gradient of $\bt_1$ must be measured with the parameter-shift method in general.
	If $U_{\rm SB}$ is sufficiently shallow compared to $U_{\rm SLPA}$, the gradient measurement efficiency of $U$ is still high, although additional cost is required for the parameter-shift method for the gradient of $\bt_1$.
	Meanwhile, $U_{\rm SB}$ can break the symmetry constraint of the SLPA, reinforcing the capability of the circuit.
	Therefore, combining the SLPA with another symmetry-breaking circuit can eliminate the symmetry constraint while maintaining high gradient measurement efficiency.
	We will leave further investigation of the capabilities of this model for future work.

	\section*{Discussion} \label{sec: conclusions}
	\vspace{-0.42cm}
	
	\noindent
	This work has proven the general trade-off relation between gradient measurement efficiency and expressivity in deep QNNs.
	Furthermore, based on this trade-off relation, we have proposed a general ansatz of CBC called the SLPA, which can reach the upper bound of the trade-off inequality by leveraging the symmetric structure of the quantum circuit.
	These results provide a guiding principle for designing efficient QNNs, fully unleashing the potential of QNNs.

	While the SLPA allows for the most efficient gradient estimation, the theoretical limit of gradient measurement efficiency motivates further investigation of other approaches for training quantum circuits.
	A promising one is the use of gradient-free optimization algorithms, such as Powell's method~\cite{Powell1964-uz} and simultaneous perturbation stochastic approximation~\cite{Spall1992-rp}.
	Such algorithms use only a few quantum measurements to update the circuit parameters, potentially speeding up the training process.
	However, it remains unclear how effective the gradient-free algorithms are for large-scale problems.
	Thorough verification and refinement of these algorithms could overcome the challenge of high computational costs in VQAs.
	Another promising direction is the coherent manipulation of multiple copies of input data.
	Since our theory implicitly assumes that only one copy of input data is available at a time, the existence of more efficient algorithms surpassing our trade-off inequality is not prohibited in multi-copy settings.
	In fact, there is an efficient algorithm for measuring the gradient in multi-copy settings, where $\mO(\text{polylog}(L))$ copies of input data are coherently manipulated to be measured by shadow tomography~\cite{Abbas2023-hy, Aaronson2020-uc}.
	However, this algorithm is hard to implement in near-term quantum devices due to the requirements of many qubits and long execution times.
	Exploring more efficient algorithms in multi-copy settings is an important open issue.

	\section*{Methods}
	\vspace{-0.42cm}

	\subsection*{Proof sketch of Theorem~\ref{thm: main_informal}}
	\vspace{-0.42cm}
	
	\noindent
	To prove Theorem~\ref{thm: main_informal}, we introduce a graph representation of DLA called DLA graph to clarify the relationship between gradient measurement efficiency and expressivity.
	The DLA graph consists of nodes and edges.
	Each node $P$ corresponds to the Pauli basis of the DLA $\mcG$ (i.e., $P\in\mcG$), and two nodes $P\in\mcG$ and $Q\in\mcG$ are connected by an edge if they anti-commute, $\{P,Q\}=0$.
	Remarkably, the commutation relations between the Pauli bases of the DLA are closely related to the simultaneous measurability of gradient components.
	Therefore, the DLA graph visualizes the commutation relations of the DLA, allowing us to clearly understand the relationship between the gradient measurement efficiency and the DLA structure.

	The proof of Theorem~\ref{thm: main_informal} requires understanding the gradient measurement efficiency $\ef$ and the expressivity $\ex$ in terms of the DLA graph.
	By definition, the expressivity $\ex$ corresponds to the total number of nodes in the DLA graph.
	On the other hand, the relationship between the gradient measurement efficiency $\ef$ and the DLA graph is nontrivial.
	The first step to obtaining the gradient measurement efficiency is to consider whether given two gradient operators $\grad_j$ and $\grad_k$ are simultaneously measurable (i.e., whether $[\grad_j, \grad_k]=0$).
	The gradient operator $\grad_j$ is defined from the generator of the quantum circuit $G_j\in\mcGc$ [see Eq.~\eqref{eq: gradj}], corresponding to a node of the DLA graph.
	In Supplementary Section~\ref{secap: lemmas1--5}, we will present Lemmas~\ref{thm_ap: not connected}--\ref{thm_ap: Gj=Gk} to show that the simultaneous measurability of $\grad_j$ and $\grad_k$ is determined by some structural relations between the nodes $G_j, G_k$, and the observable $O$ in the DLA graph.
	Hence, how many gradient components can be simultaneously measured, namely $\ef$, is also determined by the structure of the DLA graph.

	Based on Lemmas~\ref{thm_ap: not connected}--\ref{thm_ap: Gj=Gk}, we decompose the DLA graph into several subgraphs, where nodes belonging to different subgraphs cannot be measured simultaneously.
	Therefore, the number of nodes in each subgraph bounds the number of simultaneously measurable gradient components, i.e., the gradient measurement efficiency.
	Using this idea, we can map the problem on $\ef$ and $\ex$ to the problem on the DLA graph, or specifically on the relation between the total number of nodes and the size of the subgraphs in the DLA graph.
	Finally, by considering constraints on the number of nodes derived from the decomposition to the subgraphs, we prove the inequalities between $\ef$ and $\ex$.
	The detailed proof is provided in Supplementary Sections~\ref{secap: App proof}--\ref{secap: proof eqs}.

	\subsection*{Commuting block circuit}
	\vspace{-0.42cm}
	
	\noindent
	The CBC is a parameterized quantum circuit allowing for efficient gradient estimation~\cite{Bowles2023-vf}.
	It consists of $B$ block unitaries:
	\begin{align}
		U(\bt) = \prod_{a=1}^{B} U_a(\bt_a), \label{eq: CBC unitary}
	\end{align}
	where each block contains multiple variational rotation gates as
	\begin{align}
		U_a(\bt_a) = \prod_{j} \exp\left(i\theta^a_j G^a_j\right). \label{eq: CBC unitary2}
	\end{align}
	Here, $G_j^a\in\mcPn$ is the generator of the $j$th rotation gate in the $a$th block, and $\bt_a=(\theta^a_1, \theta^a_2, \cdots)$ is the variational rotation angles.
	The generators of CBC must satisfy the following two conditions.
	First, generators within the same block are commutative:
	\begin{align}
		[G^a_j, G^a_k] = 0 \quad \forall j,k.
		\label{eq: CBC sameBlock}
	\end{align}
	Second, generators in any two distinct blocks are either all commutative or all anti-commutative:
	\begin{align}
		[G^a_j,G^b_k] = 0  \quad \forall j,k  \quad \text{or} \quad \{G^a_j,G^b_k \} = 0  \quad \forall j,k.
		\label{eq: CBC difBlock}
	\end{align}
	We also consider a cost function $C=\text{tr}[U\rho U^\dag O]$ with a Pauli observable $O$.
	
	This specific structure of CBC allows us to measure the gradient components of the cost function with only two different quantum circuits for each block.
	To measure the gradient components for the $a$th block, we divide the generators of the rotation gates $\mcGc_a=\{G^a_j\}_j$ into the commuting and anti-commuting parts with the observable $O$: $\mcGc_a = \mcGc_a^{\rm com}\sqcup \mcGc_a^{\rm ant}$, where $[\mcGc_a^{\rm com}, O]=\{\mcGc_a^{\rm ant},O\}=0$.
	Then, the gradient components in $\mcGc_a^{\rm com}$ ($\mcGc_a^{\rm ant}$) can be measured simultaneously using the linear combination of unitaries with an ancilla qubit (see Supplementary Section~\ref{secap: CBC measure} for details).
	Therefore, the full gradient of $U(\bt)$ can be estimated with only $2B$ types of quantum measurements, which is independent of the number of rotation gates in each block (to be precise, $2B-1$ quantum measurements are sufficient because the commuting part of the final block does not contribute to the gradient).
	This allows us to measure the gradient more efficiently than conventional variational models based on the parameter-shift method, where the measurement cost is proportional to the number of parameters.

	While the basic framework for CBC has been proposed, there are still some challenges to be addressed.
	First, it remains unclear how the specific structure of the commuting block circuit affects the QNN expressivity.
	Second, a general method to construct the CBC or find the generators $G^a_j$ that satisfy the commutation relations of Eqs.~\eqref{eq: CBC sameBlock} and \eqref{eq: CBC difBlock} has not yet been established.
	The SLPA addresses these challenges by providing a general construction method with stabilizers and logical Pauli operators and uncovering its expressivity based on the stabilizer formalism.

	\subsection*{Quantum circuits used in numerical experiments}
	\vspace{-0.42cm}
	
	\noindent
	The unitary of the local symmetric ansatz is given by
	\begin{align}
		U_\text{SA}(\bt)
		= \prod_{d=1}^D \prod_{k=1}^{3n} \exp(i\theta_k^d L_k). \label{eq: U_SA}
	\end{align}
	Here, a local Pauli operator $L_k$ is defined as
	\begin{align}
		L_k = 
		\begin{dcases}
			P_{2k-1}^{\mu_k}P_{2k}^{\mu_k} & \text{for $k=an + b$}, \\
			P_{2k}^{\nu_k}P_{2k+1}^{\nu_k} & \text{for $k=\left(a+\frac{1}{2}\right)n + b$}
		\end{dcases}
	\end{align}
	with 
	\begin{align}
		P_{j}^{\mu}=
		\begin{cases}
			X_j & \text{for $\mu=3\ell+1$} \\
			Y_j & \text{for $\mu=3\ell+2$} \\
			Z_j & \text{for $\mu=3\ell$} 
		\end{cases},
	\end{align}
	where $a\in\{0,1,2\}, b\in\{1,\cdots,n/2\}, \mu_{k}=a+b$, $\nu_{k}=a+b+n/2$, $\ell\in\mathbb{Z}$, and $P^\mu_j = P^\mu_{j+n}$.
	The SLPA is constructed by taking the products of the stabilizers $\mcS$ and the generators of $U_\text{SA}(\bt)$ as
	\begin{align}
		U_\text{SLPA}(\bt)
		= \prod_{d=1}^D \prod_{k=1}^{3n} U_{k}(\bt_k^d),   \label{eq: U_SL}
	\end{align}
	where $U_{k}(\bt_k^d)$ is the block unitary defined as 
	\begin{align}
		U_{k}(\bt_k^d) = \prod_{j=1}^4 \exp(i\theta_{kj}^d S_j L_k).  
	\end{align}
	Finally, the unitary of the non-symmetric ansatz is given by
	\begin{align}
		U_\text{NSA}(\bt) = \prod_{d=1}^D V_\text{ent}(\bt^d_3) V_\text{rot}(\bt^d_1,\bt^d_2)
	\end{align}
	with
	\begin{align}
		&V_\text{rot}(\bt^d_1,\bt^d_2) = \prod_{j=1}^n \exp(i\theta_{2j}^d Y_j)\exp(i\theta_{1j}^d X_j), \\
		&V_\text{ent}(\bt^d_3) = \prod_{j=1}^n \exp(i\theta_{3j}^d Z_jZ_{j+1}).
	\end{align}
	We optimize these quantum circuits by minimizing the mean squared error loss function with the Adam optimizer~\cite{Kingma2014-db}.
	The hyper-parameter values used in this work are initial learning rate $=10^{-3}$, $\beta_1=0.9$, $\beta_2=0.999$, and $\epsilon=10^{-8}$.
	We also adopt the stochastic gradient descent~\cite{Robbins1951-ql}, where only one training data is used to estimate the gradient at each iteration.

	\section*{Acknowledgements}
	\vspace{-0.42cm}
	\noindent
	Fruitful discussions with Riki Toshio, Yuichi Kamata, Shintaro Sato, Snehal Raj, and Brian Coyle are gratefully acknowledged.
	S.Y. was supported by FoPM, WINGS Program, the University of Tokyo.

	\newpage
	
	\noindent 
	\textbf{\raggedright\Large\bfseries\sffamily Supplementary Information}{\Large\par}

	\vspace{0cm}
	\tableofcontents

	\section{Problem formulation} \label{secap: problem formulation}
	
	In this section, we describe the problem formulation with some additional assumptions on our model.
	We also derive a more specific form of the gradient operator for later use. 
	
	\subsection{Model, efficiency, and expressivity} 
	This work considers the following QNN on an $n$-qubit system: 
	\begin{align}
		U(\bt) = \prod_{j=1}^L \exp(i\theta_j G_j),
	\end{align}
	where $G_j$ is a Pauli operator in $\mcPn=\{I, X, Y, Z\}^{\otimes n}$, $\bt = (\theta_1,\theta_2,\cdots)$ are variational rotation angles, and $L$ is the number of rotation gates.
	Let $\mcGc=\{G_j\}_{j=1}^L$ be the set of generators. 
	We consider a cost function $C(\bt) = \tr\left[ \rho U^\dag(\bt) O U(\bt) \right]$ with a Pauli operator $O\in\mcPn$, where $\rho$ is an input quantum state.
	Here, even circuits with non-variational Clifford gates (e.g., CZ and CNOT gates) can be transformed into this form by swapping the Clifford gates and the Pauli rotation gates.
	Hence, our model includes a wide class of parameterized quantum circuits.

	The gradient of the cost function can be written as the expectation value of the gradient operator $\grad_j$:
	\begin{align}
		&\partial_j C(\bt) = \tr \left[ \rho \grad_j(\bt) \right],  \label{apeq: Cost gradient}
	\end{align}
	where $\grad_j$ is defined as 
	\begin{align}
		&\grad_j(\bt) = \partial_j \left[ U^\dag(\bt) O U(\bt) \right]. \label{apeq: gradj}
	\end{align}
	See Sec.~\ref{secap: gradient operator} for a more specific form of the gradient operator.
	In quantum mechanics, two observables $A$ and $B$ can be simultaneously measured if and only if $[A,B]=0$.
	Therefore, we define the simultaneous measurability of $\partial_j C(\bt)$ and $\partial_k C(\bt)$ by $[\grad_j(\bt),\grad_k(\bt)]=0$ for all $\bt$.
	Note that this definition assumes that the commutation relation holds for any $\bt$.
	Thus, even if $\partial_j C(\bt)$ and $\partial_k C(\bt)$ cannot be simultaneously measured in the sense of this definition, they may be simultaneously measurable for some $\bt$.

	Based on this definition, we partition the gradient operators $\{\grad_j\}_{j=1}^L$ into $M_L$ simultaneously measurable sets.
	This partitioning enables us to estimate all gradient components with $M_L$ types of quantum measurements in principle.
	Thus, we define gradient measurement efficiency in the deep circuit limit as
	\begin{align}
		\ef = \lim_{L\to\infty} \frac{L}{\text{min}(M_L)}, \label{apeq: Feff}
	\end{align}
	where $\text{min}(M_L)$ is the minimum number of sets among all possible partitions. 
	In this definition, $\ef$ indicates the mean number of simultaneously measurable components in the gradient.

	The QNN expressivity is defined using the dynamical Lie algebra (DLA)~\cite{Albertini2001-xk, Zeier2011-bq, D-Alessandro2007-kk}.
	To this end, we consider the following Lie closure:
	\begin{align}
		i\mcG = \braket{i\mcGc}_\text{Lie},
	\end{align}
	where $\braket{i\mcGc}_\text{Lie}$ is defined as the set of Pauli operators obtained by repeatedly taking the nested commutator between the circuit generators in $i\mcGc=\{iG_j\}_{j=1}^L$, i.e., $[\cdots[[iG_j,iG_k],iG_\ell],\cdots]\in i\mcG$ for $G_j,G_k,G_\ell,\cdots\in\mcGc$.
	Letting $\mcPn=\{I,X,Y,Z\}^{\otimes n}$ be the set of $n$-qubit Pauli operators, $\mcGc\subseteq \mcG \subseteq \mcPn$ holds by definition, where we ignore the coefficients of Pauli operators in $\mcG$.
	The DLA is defined by $\mcG$ as
	\begin{align}
		\mfg = \text{span}(\mcG),
	\end{align}
	which is the subspace of $\mf{su}(2^n)$ spanned by the Pauli operators in $\mcG$.
	As the DLA characterizes the types of unitaries that the QNN can express in the overparameterized regime, we define the expressivity of the QNN in the deep circuit limit as
	\begin{align}
		\ex = \text{dim}(\mfg),
	\end{align}
	which is the dimension of DLA.

	\subsection{Assumptions}
	
	Our model imposes several reasonable assumptions, which will be used in the proof of Theorem~\ref{thm: main_informal}.
	
	To remove redundant and irrelevant quantum gates, we assume the following conditions on the quantum circuit $U(\bt)$:
	\begin{con}[No redundant quantum gates] \label{cond: 1}
		For any $j<k$, if $G_j=G_k$, there exists $\ell$ ($j<\ell<k$) such that $\{G_j,G_\ell\}=0$.
	\end{con}
	\begin{con}[No irrelevant quantum gates] \label{cond: 2}
		For any $G_j\in \mcGc$ commuting with $O$ (i.e., $[G_j,O]=0$), there exist $Q_1,\cdots,Q_m \in \mcGc$ such that $\{G_j,Q_1\}=\{Q_1,Q_2\}=\cdots=\{Q_{m-1},Q_m\}=\{Q_m,O\}=0$.
	\end{con}
	The first condition removes the redundancy of quantum gates.
	If there is no $\ell$ satisfying the condition, the two rotation gates, $e^{i\theta_jG_j}$ and $e^{i\theta_kG_k}$, can be merged into one rotation gate, $e^{i(\theta_j+\theta_k)G_j}$, by swapping the positions of gates.
	This means that one of the two rotation gates is redundant.
	The second condition removes irrelevant quantum gates.
	If there is a quantum gate that violates the condition, the circuit is decomposed into two mutually commuting unitaries as $U=U_1U_2$, where $U_2$ commutes with the observable $O$ (i.e., $[U_1,U_2]=[U_2,O]=0$).
	Then, the cost function is written as $C=\tr[\rho U^\dag O U]=\tr[\rho U_1^\dag O U_1]$, which implies that $U_2$ is irrelevant in this QNN. 
	Therefore, we impose these two conditions on the quantum circuit to remove the redundant and irrelevant quantum gates without loss of generality.

	We assume another condition on the quantum circuit in relation to expressivity:
	\begin{con} \label{cond: 3}
		Let $U_\text{fin}$ be the final part of the circuit with depth $\LB$:
		\begin{align}
			U_\text{fin}=\prod_{j=L-\LB+1}^L \exp(i\theta_j G_j).
		\end{align}
		Then, there exists a constant $\LB$ (independent of $L$) such that $U_\text{fin}$ can express $e^{i\phi_1Q_1}e^{i\phi_2Q_2}$ for any $Q_1, Q_2\in\mcG$ and any $\phi_1,\phi_2\in\mathbb{R}$.
	\end{con}
	This condition is easily satisfied in the limit of $L\to\infty$ when all generators in $\mcGc$ appear uniformly in the circuit.
	Note that the contribution from $U_\text{fin}$ with depth $\LB=\mO(1)$ to the gradient measurement efficiency $\ef$ vanishes in the deep circuit limit $L\to\infty$.
	Thus, for evaluating $\ef$, it suffices to consider only the contribution from the initial part of the circuit with depth $\Lbulk= L-\LB = \mO(L)$.

	\subsection{Gradient operator} \label{secap: gradient operator}
	
	For later use, we derive the following form of the gradient operator:
	\begin{align}
		\grad_j(\bt) = -i [\tilde{G}_j(\bt), \tilde{O}(\bt)], \label{apeq: gradient operator}
	\end{align}
	where we have defined 
	\begin{align}
		&\tilde{O}(\bt) = U^\dag(\bt)OU(\bt), \\
		&\tilde{G}_j(\bt) = U_{j+}^\dag(\bt)G_jU_{j+}(\bt),
	\end{align}
	with the unitary circuit before the $j$th gate $U_{j+}(\bt)=\prod_{k=1}^{j-1}\exp(i\theta_kG_k)$.
	This form is used to prove the main theorem in the following sections.

	To derive this form, we first simplify the gradient operator $\grad_j(\bt)=\partial_j [U^\dag(\bt)OU(\bt)]$ to
	\begin{align}
		\grad_j
		&=  \frac{\partial U^\dag}{\partial \theta_j} O U  + U^\dag O \frac{\partial U}{\partial \theta_j}. \label{ap: djC}
	\end{align}
	In this equation, $\partial U/\partial \theta_j$ is written as
	\begin{align}
		\frac{\partial U}{\partial \theta_j} 
		&= i U_{j-} G_j U_{j+}, \label{ap: djU}
	\end{align}
	where $U_{j+}=\prod_{k=1}^{j-1} e^{i\theta_k G_k}$ and $U_{j-}=\prod_{k=j}^L e^{i\theta_k G_k}$ are the unitaries before and after the $j$th rotation gate.
	By inserting $U_{j+}U_{j+}^\dag = I$ into Eq.~\eqref{ap: djU}, we have
	\begin{align}
		\frac{\partial U}{\partial \theta_j} = i U_{j-} (U_{j+}U_{j+}^\dag) G_j U_{j+} = iU \tilde{G}_j, \label{ap: djU2}
	\end{align}
	with $\tilde{G}_j=U_{j+}^\dag G_j U_{j+}$ and $U_{j-} U_{j+}=U$.
	By taking the Hermitian conjugate of Eq.~\eqref{ap: djU2}, we also have
	\begin{align}
		\frac{\partial U^\dag}{\partial \theta_j} = -i\tilde{G}_j U^\dag. \label{ap: djU3}
	\end{align}
	Thereby, Eq.~\eqref{ap: djC} is reduced to
	\begin{align}
		\grad_j 
		&= -i \tilde{G}_j U^\dag O U + i U^\dag O U \tilde{G}_j = -i [\tilde{G}_j, \tilde{O}]
	\end{align}
	with $\tilde{O}=U^\dag O U$.

	\section{Proof of main theorem} \label{secap: App proof}
	
	The gradient measurement efficiency and expressivity defined above obey the following theorem:
	\begin{thm}[The formal version of Theorem~\ref{thm: main_informal}] \label{thm: main}
		In deep QNNs satisfying Conditions~\ref{cond: 1}--\ref{cond: 3}, gradient measurement efficiency and expressivity obey the following inequalities:
		\begin{align}
			&\ex \leq \frac{4^n}{\ef} - \ef, \label{apeq: main_inequality} 
		\end{align} 
		and
		\begin{align}
			&\ex \geq \ef, \label{apeq: main_inequality2}
		\end{align} 
		where $n$ is the number of qubits.
	\end{thm}   
	In this section,  we prove this Theorem by introducing a new theoretical concept, the DLA graph.
	The proof sketch is provided in Methods.
	Several Lemmas and inequalities used here will be proven in Secs.~\ref{secap: lemmas1--5} and \ref{secap: proof eqs}.

	\subsection{Notation}
	
	For convenience, we will use an abbreviation for operator sets $\mcA=\{A_1,A_2,\cdots\}$ and $\mcB=\{B_1,B_2,\cdots\}$ as
	\begin{align*}
		&[\mcA,\mcB]=0 \xLeftrightarrow{\text{def}} [A_j,B_k]=0 \quad \text{$\forall A_j\in\mcA$, $\forall B_k\in\mcB$}, \\
		&\{\mcA,\mcB\}=0 \xLeftrightarrow{\text{def}} \{A_j,B_k\}=0 \quad \text{$\forall A_j\in\mcA$, $\forall B_k\in\mcB$}.
	\end{align*}
	Similarly, we will use an abbreviation for an operator set $\mcA=\{A_1,A_2,\cdots\}$ and an operator $B_k$ as
	\begin{align*}
		&[\mcA,B_k]=0 \xLeftrightarrow{\text{def}} [A_j,B_k]=0 \quad \text{$\forall A_j\in\mcA$}, \\
		&\{\mcA,B_k\}=0 \xLeftrightarrow{\text{def}} \{A_j,B_k\}=0 \quad \text{$\forall A_j\in\mcA$}.
	\end{align*}
	In what follows, we ignore the coefficients of Pauli operators as we only consider the commutation and anti-commutation relations for the proof.

	\subsection{DLA graph and simultaneous measurability} \label{secap: DLA and simultaneous measurability}

	The proof of Theorem~\ref{thm: main} begins by considering the relationship between the simultaneous measurability of gradient components and the DLA structure.
	Here, we introduce the DLA graph to visualize the DLA.
	In this graph representation, each node $P$ corresponds to the Pauli basis of the DLA $\mcG$, and two nodes $P\in\mcG$ and $Q\in\mcG$ are connected by an edge if $\{P,Q\}=0$.
	For example, the following figure is the DLA graph with $\mcG=\{IX,XI,ZY,YZ,ZZ,YY\}$ for a two-qubit system:
	\begin{equation*}
		\begin{tikzpicture} [scale=1.5]
			\node[name=A][draw,circle,fill=white] at (+0.5, +0.866) {$XI$};
			\node[name=B][draw,circle,fill=white] at (-0.5, +0.866) {$IX$};
			\node[name=C][draw,circle,fill=white] at (+1,0) {$YZ$};
			\node[name=D][draw,circle,fill=white] at (-1,0) {$ZY$};
			\node[name=E][draw,circle,fill=white] at (+0.5, -0.866) {$YY$};
			\node[name=F][draw,circle,fill=white] at (-0.5, -0.866) {$ZZ$};
			\draw (A) -- (C);
			\draw (A) -- (D);
			\draw (A) -- (E);
			\draw (A) -- (F);
			\draw (B) -- (C);
			\draw (B) -- (D);
			\draw (B) -- (E);
			\draw (B) -- (F);
			\draw (C) -- (E);
			\draw (C) -- (F);
			\draw (D) -- (E);
			\draw (D) -- (F);
		\end{tikzpicture}.
	\end{equation*}

	Here, we define the connectivity on the DLA graph:
	\begin{dfn}[DLA-connectivity] \label{def: DLA-connectivity}
		We say that $P,Q\in\mc{P}_n$ ($P\neq Q$) are $\mfg$-connected and denote $P\conn Q$ when $\{P,Q\}=0$ or there exist $R_1,R_2,\cdots,R_{d-1}\in\mcG$ such that $\{P,R_1\}=\{R_1,R_2\}=\cdots=\{R_{d-1},Q\}=0$.
	\end{dfn}    
	In terms of the DLA graph, $P,Q\in\mc{P}_n$ are $\mfg$-connected if there exists a path connecting $P$ and $Q$ on the DLA graph.
	This connectivity satisfies a transitive relation that $P\conn Q$ and $Q\conn R$ lead to $P\conn R$ for $P,R\in \mc{P}_n$ and $Q\in \mcG$.
	Note that the $\mfg$-connectivity are defined for $P,Q\in\mc{P}_n=\{I,X,Y,Z\}^{\otimes n}$ not only for $P,Q\in\mcG$.
	We also define the separability of the DLA graph as follows:
	\begin{dfn}[Separability] \label{def: DLA-separability}
		We say that a DLA subgraph $\mcG{}_1 \subset \mcG$ is separated when $[\mcG{}_1, \mcG{}_2]=0$, where $\mcG{}_2=\mcG \smallsetminus \mcG{}_1$ is the complement of $\mcG{}_1$ in $\mcG$.
		We also say that the DLA graph is separable if it consists of two or more separated subgraphs.
	\end{dfn}  
	From a graph perspective, this separability means that a subgraph $\mcG{}_1$ is not connected to the rest of the DLA graph $\mcG{}_2$ by edges, i.e., $\forall P\in\mcG{}_1$ and $\forall Q\in\mcG{}_2$ are not $\mfg$-connected.
	Conversely, if the DLA graph is not separable, any two nodes are $\mfg$-connected.

	\begin{figure*}[t]
		\centering
		\includegraphics[width=\linewidth]{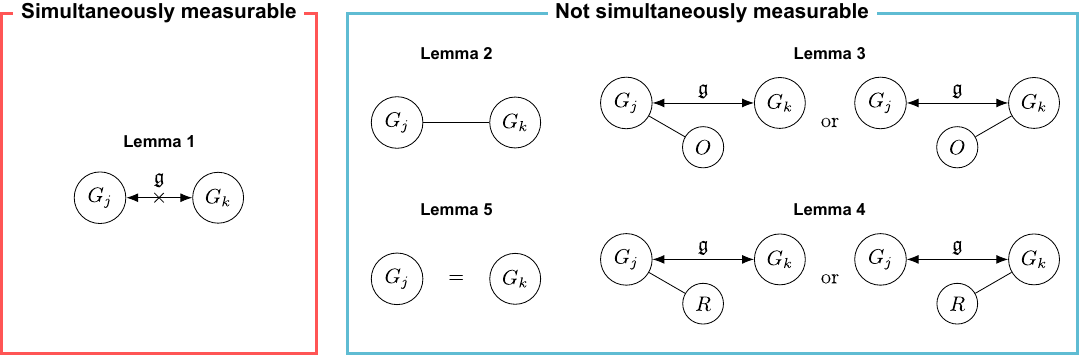}
		\caption{
			{\bf Summary of Lemmas~\ref{thm: not connected}--\ref{thm: Gj=Gk}.}
			Whether two gradient operators $\grad_j$ and $\grad_k$ are simultaneously measurable is determined by the structural relations between the corresponding nodes $G_j, G_k\in\mcGc$ and the observable $O$ in the DLA graph.
		}
		\label{fig: lemmas1--5}
	\end{figure*}

	The graph representations in this work are summarized below:
	\allowdisplaybreaks
	\def\arraystretch{3.5}
	\begin{align*}
		\raise 1.5ex \hbox{$[P,Q]=0$}
		&\quad\quad\quad
		\begin{tikzpicture} 
			\node[name=A][draw,circle,fill=white] at (-1,0) {$P$};
			\node[name=B][draw,circle,fill=white] at (+1,0) {$Q$};
		\end{tikzpicture} 
		\\ 
		\raise 1.5ex \hbox{$\{P,Q\}=0$}
		&\quad\quad\quad
		\begin{tikzpicture} 
			\node[name=A][draw,circle,fill=white] at (-1,0) {$P$};
			\node[name=B][draw,circle,fill=white] at (+1,0) {$Q$};
			\draw (A) -- (B);
		\end{tikzpicture} 
		\\ 
		\raise 1.5ex \hbox{$[P,Q]=0$ or $\{P,Q\}=0$}
		&\quad\quad\quad
		\begin{tikzpicture} 
			\node[name=A][draw,circle,fill=white] at (-1,0) {$P$};
			\node[name=B][draw,circle,fill=white] at (+1,0) {$Q$};
			\draw[dashed] (A) -- (B);
		\end{tikzpicture} 
		\\ 
		\raise 1.5ex \hbox{$P\conn Q$}
		&\quad\quad\quad
		\begin{tikzpicture} 
			\node[name=A][draw,circle,fill=white] at (-1,0) {$P$};
			\node[name=B][draw,circle,fill=white] at (+1,0) {$Q$};
			\node[name=conn] at (0,0.3) {$\mfg$};
			\draw[{Latex[length=2mm]}-{Latex[length=2mm]}] (A) -- (B); 
		\end{tikzpicture} 
		\\ 
		\raise 1.5ex \hbox{$P\conn Q$ and $[P,Q]=0$}
		&\quad\quad\quad
		\begin{tikzpicture} 
			\node[name=A][draw,circle,fill=white] at (-1,0) {$P$};
			\node[name=B][draw,circle,fill=white] at (+1,0) {$Q$};
			\node[name=conn] at (0,0.3) {$\mfg$};
			\draw[dotted, {Latex[length=2mm]}-{Latex[length=2mm]}] (A) -- (B); 
		\end{tikzpicture} 
		\\ 
		\raise 1.5ex \hbox{Not $P\conn Q$}
		&\quad\quad\quad
		\begin{tikzpicture} 
			\node[name=A][draw,circle,fill=white] at (-1,0) {$P$};
			\node[name=B][draw,circle,fill=white] at (+1,0) {$Q$};
			\node[name=conn] at (0,0.3) {$\mfg$};
			\node[name=conn][align=right] at (0,0) {$\times$};
			\draw[{Latex[length=2mm]}-{Latex[length=2mm]}] (A) -- (B); 
		\end{tikzpicture} 
		\\ 
	\end{align*}
	\allowdisplaybreaks[0]

	In this proof, we map the problem on $\ef$ and $\ex$ to the problem on the DLA graph.
	To this end, we need to interpret $\ef$ and $\ex$ in terms of the DLA graph.
	By definition, the expressivity $\ex$ corresponds to the total number of nodes in the DLA graph.
	On the other hand, identifying the gradient measurement efficiency $\ef$ from the DLA graph is not straightforward.
	The first step in evaluating $\ef$ is to understand when two gradient operators $\grad_j$ and $\grad_k$ (namely $\partial_j C$ and $\partial_k C$) are simultaneously measurable.
	The gradient operator $\grad_j$ is defined from the generator of the circuit $G_j\in\mcGc$ [see Eq.~\eqref{apeq: gradient operator}], which corresponds to a node of the DLA graph $\mcG$.
	Remarkably, whether $\grad_j$ and $\grad_k$ are simultaneously measurable is determined from the structural relations between the corresponding nodes $G_j,G_k$ and the observable $O$ on the DLA graph.
	The following lemmas show the relationship between the simultaneous measurability of gradient components and the DLA structure (see Sec.~\ref{secap: lemmas1--5} for their proofs):
	\begin{lem}\label{thm: not connected}
		If $G_j,G_k\in\mcGc$ are not $\mfg$-connected, $\partial_j C$ and $\partial_k C$ can be simultaneously measured.  
	\end{lem}
	
	\begin{lem}\label{thm: PjPk}
		For $j,k \leq \Lbulk$, if $G_j,G_k\in\mcGc$ anti-commute, $\partial_j C$ and $\partial_k C$ cannot be simultaneously measured.    
	\end{lem}

	\begin{lem}\label{thm: [P,O]}
		Consider $\mfg$-connected $G_j,G_k\in\mcGc$ for $j,k \leq \Lbulk$.
		If $\{G_j,O\}=[G_k,O]=0$ or $[G_j,O]=\{G_k,O\}=0$, $\partial_j C$ and $\partial_k C$ cannot be simultaneously measured.   
	\end{lem}
	
	\begin{lem}\label{thm: PjR&PkR}
		Consider $\mfg$-connected $G_j, G_k\in\mcGc$ for $j,k \leq \Lbulk$.
		If there exists $R\in\mcG$ such that $\{G_j,R\}=[G_k,R]=0$ or $[G_j,R]=\{G_k,R\}=0$, $\partial_j C$ and $\partial_k C$ cannot be simultaneously measured.    
	\end{lem}
	
	\begin{lem}\label{thm: Gj=Gk}
		For $j<k \leq \Lbulk$, if $G_j=G_k$, $\partial_j C$ and $\partial_k C$ cannot be simultaneously measured.    
	\end{lem}

	Lemma~\ref{thm: not connected} (Lemmas~\ref{thm: PjPk}--\ref{thm: Gj=Gk}) gives the sufficient (necessary) condition of the DLA structure for measuring multiple gradient components simultaneously.
	Figure~\ref{fig: lemmas1--5} illustrates the DLA graph representation of these lemmas.
	We note that although Lemmas~\ref{thm: PjPk}--\ref{thm: Gj=Gk} assume $j,k\leq \Lbulk=\mO(L)$, the contribution from the final part of the circuit with constant depth $\LB=\mO(1)$ to the gradient measurement efficiency is negligible in the deep circuit limit.
	Thus, in the following subsection, we only consider the gradient components for $j,k\leq \Lbulk$.

	It is noteworthy that Lemmas~\ref{thm: PjPk}--\ref{thm: PjR&PkR} are closely related to the commutation relations of the CBC.
	First, Lemma~\ref{thm: PjPk} states that measuring two gradient components $\partial_j C$ and $\partial_k C$ simultaneously requires that the corresponding generators $G_j$ and $G_k$ are commutative.
	This requirement is satisfied in the CBC since the generators in the same block are all commutative.
	Second, by Lemma~\ref{thm: [P,O]}, two gradient components $\partial_j C$ and $\partial_k C$ are not simultaneously measurable if the corresponding generators $G_j$ and $G_k$ have different (anti-)commutation relations with the observable $O$.
	This corresponds to the fact that the gradient of CBC must be measured separately for the commuting and anti-commuting parts with $O$, in each block.
	Finally, by Lemma~\ref{thm: PjR&PkR}, two gradient components $\partial_j C$ and $\partial_k C$ are not simultaneously measurable if there exists another generator $R\in\mcGc$ such that $\{G_j,R\}=[G_k,R]=0$ or $[G_j,R]=\{G_k,R\}=0$.
	In the CBC, from the commutation relations between distinct blocks, there does not exist such $R$ for generators $G_j^a$ and $G_k^a$ in the same block, and thus the simultaneous gradient measurement is possible.

	Before moving on to the proof, we revisit Condition~\ref{cond: 2} in terms of the DLA graph.
	This condition states that, for any $G\in\mcGc$, there exist $Q_1,\cdots,Q_m\in\mcGc$ such that $\{G,Q_1\}=\cdots=\{Q_m,O\}=0$. 
	That is, the observable $O$ is $\mfg$-connected to $\forall G\in\mcGc$ and thus $\forall G\in\mcG$.
	Otherwise, the DLA graph can be decomposed into separated subgraphs as $\mcG=\mcG{}_1 \sqcup \mcG{}_2$, where $\mcG{}_1$ ($\mcG{}_2$) is (not) $\mfg$-connected to $O$ (i.e., $[\mcG{}_1,\mcG{}_2]=[\mcG{}_2,O]=0$).
	Then, the unitary circuit is also decomposed as $U=U_1U_2$ with $U_1\in e^{\mfg{}_1}$ and $U_2\in e^{\mfg{}_2}$, where we have defined $\mfg{}_{1}=\text{span}(\mcG{}_1)$ and $\mfg{}_{2}=\text{span}(\mcG{}_2)$.
	Using $[U_1,U_2]=[U_2,O]=0$ derived from $[\mcG{}_1,\mcG{}_2]=[\mcG{}_2,O]=0$, we have $C=\tr[\rho U^\dag OU]=\tr[\rho U^\dag_1OU_1]$, indicating that $U_2$ does not affect the result.
	Therefore, we assume that such irrelevant gates are absent in the circuit.

	\begin{figure*}[t]
		\centering
		\includegraphics[width=0.9\linewidth]{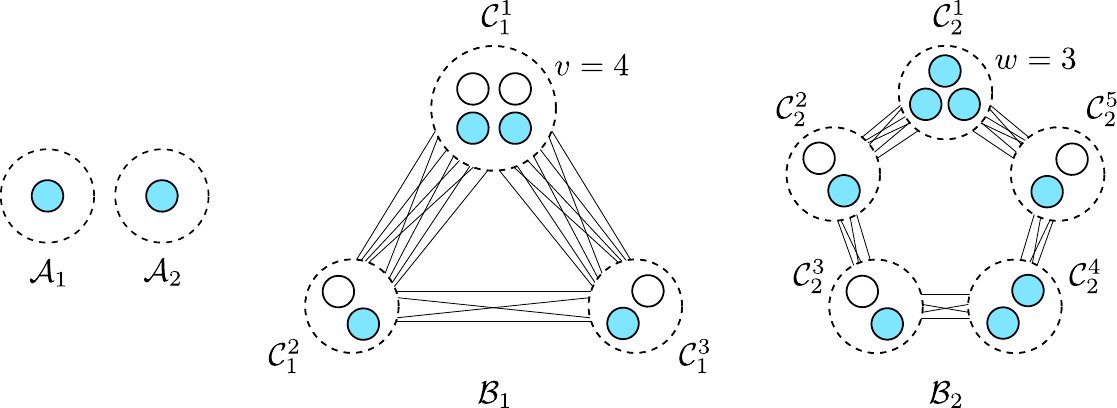}
		\caption{
			{\bf Example of DLA decomposition.}
			The white (blue) circles denote Pauli operators in $\mcG$ that (anti-)commute with the observable $O$.
			This DLA graph is decomposed as $\mcG=\mcA_1\sqcup\mcA_2\sqcup\mcB_1\sqcup\mcB_2$ with $\mcB_1=\bigsqcup_{a=1}^3 \mcC_1^{a}$ and $\mcB_2=\bigsqcup_{a=1}^5 \mcC_2^{a}$, where $p=2, q=2, r_1=3$, and $r_2=5$.
			As defined in Eqs.~\eqref{apeq: def v}--\eqref{apeq: def Cw}, we have $v=4, w=3, \mcC_v=\mcC_1^{1}, \mcC_w=\mcC_2^{1}$, and $\mcS=\{A_1,A_2,F_1F_1,F_1F_2,F_1F_3\}$ with $\mcA_1=\{A_1\}, \mcA_2=\{A_2\}$, and $\mcC_w=\{F_1,F_2,F_3\}$.	
		}
		\label{fig: DLA_decomposition}
	\end{figure*}

	\subsection{Trade-off between $\ef$ and $\ex$} \label{secap: main proof}

	Here, we prove Theorem~\ref{thm: main} using Lemmas~\ref{thm: not connected}--\ref{thm: Gj=Gk}. 
	Let us begin by showing Eq.~\eqref{apeq: main_inequality2}.
	By Lemma~\ref{thm: Gj=Gk}, $\partial_j C$ and $\partial_k C$ are not simultaneously measurable if $G_j=G_k$, which indicates that the maximum number of simultaneously measurable gradient components is bounded by the total number of nodes in the DLA graph, namely the expressivity $\ex$.
	Since the gradient measurement efficiency $\ef$ is defined as the mean number of simultaneously measurable components in the gradient, we obtain $\ex\geq\ef$.

	To prove Eq.~\eqref{apeq: main_inequality}, we identify the gradient measurement efficiency $\ef$ in the DLA graph.
	Lemmas~\ref{thm: not connected}--\ref{thm: Gj=Gk} give the conditions for measuring two gradient components simultaneously.
	Based on these lemmas, we decompose the DLA graph into several subgraphs that consist of (potentially) simultaneously measurable nodes, deriving the upper bound of the gradient measurement efficiency.
	To this end, we first decompose the DLA graph as 
	\begin{align}
		\mcG = \left( \mcA_1 \sqcup \cdots \sqcup \mcA_p \right) \sqcup \left( \mcB_1 \sqcup \cdots \sqcup \mcB_q \right), \label{apeq: dec A and B}
	\end{align}
	where $\mcA_\x$'s and $\mcB_\x$'s are separated subgraphs that have one and multiple nodes, respectively (i.e., $|\mcA_\x|=1$ and $|\mcB_\x|\geq2$, see Fig.~\ref{fig: DLA_decomposition}).
	Here, $p$ and $q$ are the numbers of subgraphs $\mcA_\x$'s and $\mcB_\x$'s.
	Since $\mcA_\x$ and $\mcB_\x$ are separated, they satisfy
	\begin{align}
		&[\mcA_\x,\mcA_\y]=[\mcB_\x,\mcB_\y]=0 \quad \forall \x \neq \y, \\
		&[\mcA_\x,\mcB_\y]=0 \quad \forall \x, \y.
	\end{align}
	According to Lemma~\ref{thm: not connected}, the gradient components for not $\mfg$-connected $G_j, G_k\in\mcGc$ can be measured simultaneously.
	Therefore, when $G_j$ and $G_k$ are nodes in different subgraphs in the decomposition of Eq.~\eqref{apeq: dec A and B}, $\partial_j C$ and $\partial_k C$ are simultaneously measurable.
	We note that the element $A_\x\in\mcA_\x$ anti-commutes with the observable $O$, $\{A_\x,O\}=0$, because all the Pauli operators in $\mcG$ must be $\mfg$-connected to $O$ by Condition~\ref{cond: 2}.
	Also, $A_\x\in\mcA_\x$ can appear in the quantum circuit just once due to Condition~\ref{cond: 1} since it commutes with all the other generators in the circuit.

	To examine the simultaneous measurability in each $\mcB_\x$, we further decompose it into $r_\x$ subgraphs, based on Lemmas~\ref{thm: PjPk} and \ref{thm: PjR&PkR}:
	\begin{align}
		\mcB_\x = \mcC_\x^{1}\sqcup \cdots \sqcup \mcC_\x^{r_\x},
	\end{align}
	where $\mcC_\x^{a}$'s satisfy
	\begin{align}
		[\mcC_\x^{a},\mcC_\x^{b}]=0 \quad \text{or} \quad \{\mcC_\x^{a},\mcC_\x^{b}\}=0  \label{apeq: Cst1}
	\end{align}
	and
	\begin{align}
		&\text{$\forall a\neq b$, $\exists c$, s.t.} \notag \\
		&\left\{
		\begin{aligned}
			&[\mcC_\x^{a},\mcC_\x^{c}]=0 \\
			&\{\mcC_\x^{b},\mcC_\x^{c}\}=0 
		\end{aligned}
		\right.
		\quad\text{or}\quad
		\left\{
		\begin{aligned}
			&\{\mcC_\x^{a},\mcC_\x^{c}\}=0 \\
			&[\mcC_\x^{b},\mcC_\x^{c}]=0.
		\end{aligned}
		\right. \label{apeq: Cst2}
	\end{align}
	Equation~\eqref{apeq: Cst1} means that all the nodes in a subgraph $\mcC_\x^{a}$ share the same (anti-)commutation relations with all the nodes in a subgraph $\mcC_\x^{b}$: $[P,Q]=0$ $\forall P\in\mcC_\x^{a}, \forall Q\in\mcC_\x^{b}$ or $\{P,Q\}=0$ $\forall P\in\mcC_\x^{a}, \forall Q\in\mcC_\x^{b}$ (see Fig.~\ref{fig: DLA_decomposition}).
	Considering $a=b$, this condition naturally leads to the commutation relation within each $\mcC_\x^{a}$: $[P,Q]=0$ $\forall P,Q\in\mcC_\x^{a}$.
	Thus, Lemma~\ref{thm: PjPk} does not exclude the simultaneous measurability for $P,Q\in \mcC_\x^{a}$.
	Also, because of Eq.~\eqref{apeq: Cst2}, it is impossible to merge several $\mcC_\x^{a}$'s into a larger subgraph satisfying Eq.~\eqref{apeq: Cst1}.
	That is, when $P, Q\in\mcB_\x$ are nodes in different $\mcC_\x^{a}$'s, there always exists $R\in\mcB_\x$ such that $[P,R]=\{Q,R\}=0$ or $\{P,R\}=[Q,R]=0$.
	Therefore, for circuit generators $G_j, G_k\in\mcGc$ belonging to different subgraphs $\mcC_\x^{a}$'s, their corresponding gradient components $\partial_j C$ and $\partial_k C$ are not simultaneously measurable by Lemma~\ref{thm: PjR&PkR}.

	Finally, we decompose $\mcC_\x^{a}$ into the commuting and anti-commuting parts with the observable $O$, based on Lemma~\ref{thm: [P,O]}:
	\begin{align}
		\mcC_\x^{a} = \mcC_\x^{a+}\sqcup \mcC_\x^{a-},
	\end{align}
	where $\mcC_\x^{a\pm}$ commute and anti-commute with the observable $O$, respectively:
	\begin{align}
		[\mcC_\x^{a+}, O]=0, \quad \{\mcC_\x^{a-}, O\}=0.
	\end{align}
	By Lemma~\ref{thm: [P,O]}, for $G_j\in \mcC_\x^{a+}$ and $G_k\in \mcC_\x^{a-}$, their corresponding gradient components $\partial_j C$ and $\partial_k C$ are not simultaneously measurable.

	In summary, when $G_j,G_k\in\mcGc$ are nodes in different separated subgraphs $\mcA_\x$ or $\mcB_\x$, $\partial_j C$ and $\partial_k C$ are simultaneously measurable.
	Meanwhile, within each $\mcB_\x$, the necessary condition for simultaneously measuring $\partial_j C$ and $\partial_k C$ for $G_j,G_k\in\mcB_\x$ is that $G_j$ and $G_k$ are nodes in the same $\mcC_\x^{a\pm}$.
	Therefore, the maximum size of $\mcC_\x^{a\pm}$ bounds the maximum number of simultaneously measurable gradient components in each $\mcB_\x$.

	Using this decomposition, we first show the trade-off inequality for $q=0$, where $q$ is the number of $\mcB_\x$.
	This case corresponds to the commuting generator circuit in Ref.~\cite{Bowles2023-vf}, where all the circuit generators $A_\x\in\mcA_\x$ are mutually commuting.
	Then, the number of gates $L$ is finite by Condition~\ref{cond: 1}.
	In this case, the total number of nodes is $p$, and they can all be measured simultaneously by Lemma~\ref{thm: not connected}.
	Therefore, the gradient measurement efficiency and the expressivity are 
	\begin{align}
		\ef=\ex=p. \label{apeq: exef_q0}
	\end{align}
	Meanwhile, $A_1,\cdots,A_p$ and $A_1A_1,\cdots,A_1A_p$ are Pauli operators that differ and commute with each other~\footnote{For $j\neq k$, $A_j\neq A_k$ and $A_1A_j \neq A_1 A_k$ trivially hold. For any $j$ and $k$, $A_j \neq A_1A_k$ also holds because of $\{A_j,O\}=[A_1A_k,O]= 0$.}. 
	Because the maximum number of mutually commuting Pauli operators is $2^n$, we have
	\begin{align}
		2p\leq 2^n. \label{apeq: 2p2n_q0}
	\end{align}
	From Eqs.~\eqref{apeq: exef_q0} and \eqref{apeq: 2p2n_q0}, we obtain $\ex\leq 4^n/\ef - \ef$.

	Next, to prove the case of $q\neq0$, we define
	\begin{align}
		&v= \underset{\mcC_\x^{a}}{\text{max}} \left( \left| \mcC_\x^{a} \right| \right), \label{apeq: def v} \\
		&w= \underset{\mcC_\x^{a\pm}}{\text{max}} \left( \left| \mcC_\x^{a\pm} \right| \right), \\
		&\mcC_v= \underset{\mcC_\x^{a}}{\text{argmax}} \left( \left| \mcC_\x^{a} \right| \right), \\
		&\mcC_w= \underset{\mcC_\x^{a\pm}}{\text{argmax}} \left( \left| \mcC_\x^{a\pm} \right| \right). \label{apeq: def Cw}
	\end{align}
	Then, by Lemmas~\ref{thm: not connected}--\ref{thm: Gj=Gk}, the gradient measurement efficiency $\ef$ and the expressivity $\ex$ (i.e., the total number of nodes) are bounded as
	\begin{align}
		&\ef \leq qw, \label{apeq: Deff_cal1} \\
		&\ex \leq p + v(r_1 + \cdots + r_q). \label{apeq: Dexp_cal0}
	\end{align}
	Note that there are no contributions from $\mcA_\x$ to the right-hand side of Eq.~\eqref{apeq: Deff_cal1}.
	This is because, since the generator $A_\x\in\mcA_\x$ commutes with all the other generators, it cannot appear more than once in the circuit due to Condition~\ref{cond: 1} and thus does not contribute to the gradient measurement efficiency in the deep circuit limit.

	Besides Eq.~\eqref{apeq: Dexp_cal0}, the DLA decomposition uncovers another constraint on the expressivity $\ex$.
	To see that, we define a set of Pauli operators $\mcS$ as follows.
	If $v\geq w+p$, we define 
	\begin{align}
		\mcS = \{E_1E_1, E_1E_2, \cdots, E_1E_v\}
	\end{align}
	with $\mcC_v = \{E_1,\cdots,E_v\}$.
	Otherwise, we define 
	\begin{align}
		\mcS = \{F_1F_1, F_1F_2, \cdots, F_1F_w\} \sqcup \{A_1,\cdots,A_p\}
	\end{align}
	with $\mcC_w = \{F_1,\cdots,F_w\}$ and $\mcA_\x=\{A_\x\}$. 
	From Eq.~\eqref{apeq: Cst1}, all the operators in $\mcS$ commute with themselves and with $\mcG$: $[\mcS,\mcS]=[\mcG,\mcS]=0$.
	Thus, $\mcS$ is the stabilizer (or symmetry) of the quantum circuit $U(\bt)$, limiting the degrees of freedom in the DLA, i.e., the expressivity $\ex$.

	We can derive the inequality~\eqref{apeq: main_inequality} by counting the remaining degrees of freedom within the subspace stabilized by $\mcS$.
	For instance, let us consider the simplest case of $p=0, q=1$, and $[\mcC_v,O]=0$ or $\{\mcC_v,O\}=0$ (i.e., $\mcC_v=\mcC_w$). 
	In this case, $v=w=|\mcS|$ and $[O,\mcS]=0$ hold.
	Then, the gradient measurement efficiency is bounded as $\ef\leq w=|\mcS|$ from Eq.~\eqref{apeq: Deff_cal1}.
	As for the expressivity, on the other hand, the stabilizers $\mcS$ constrain the dimension of the DLA as $4^n/|\mcS|$ (see Lemma~\ref{lem: stab1} in Sec.~\ref{secap: preliminaries}).
	Furthermore, since $[\mcG,\mcS]=[O,\mcS]=0$, the stabilizers $\mcS$ are not included in the generators $\mcGc$ and thus the DLA $\mcG$ because of Condition~\ref{cond: 2}.
	Hence, we have $\ex\leq 4^n/|\mcS|-|\mcS|$.
	These results lead to the inequality $\ex\leq 4^n/\ef - \ef$.
	This trade-off inequality is proven in general cases, where we use the following nontrivial constraints derived from the DLA decomposition, the stabilizers $\mcS$, and Eq.~\eqref{apeq: Dexp_cal0} (see Sec.~\ref{secap: proof eqs} for details):
	\begin{align}
		\ex \leq \frac{4^n v}{4^{q-1}|\mcS|^2} + (3q-4)v + p \label{apeq: ex_inequality}
	\end{align}
	and
	\begin{align}
		\frac{4^n v}{4^{q-1}|\mcS|^2} + (3q-4)v + p \leq \frac{4^n}{qw} - qw. \label{apeq: sub_inequality}
	\end{align}
	Combining Eqs.~\eqref{apeq: Deff_cal1}, \eqref{apeq: ex_inequality}, and \eqref{apeq: sub_inequality}, we finally obtain the trade-off inequality $\ex \leq 4^n/\ef - \ef$, as required.

	\begin{figure*}[t]
		\centering
		\includegraphics[width=\linewidth]{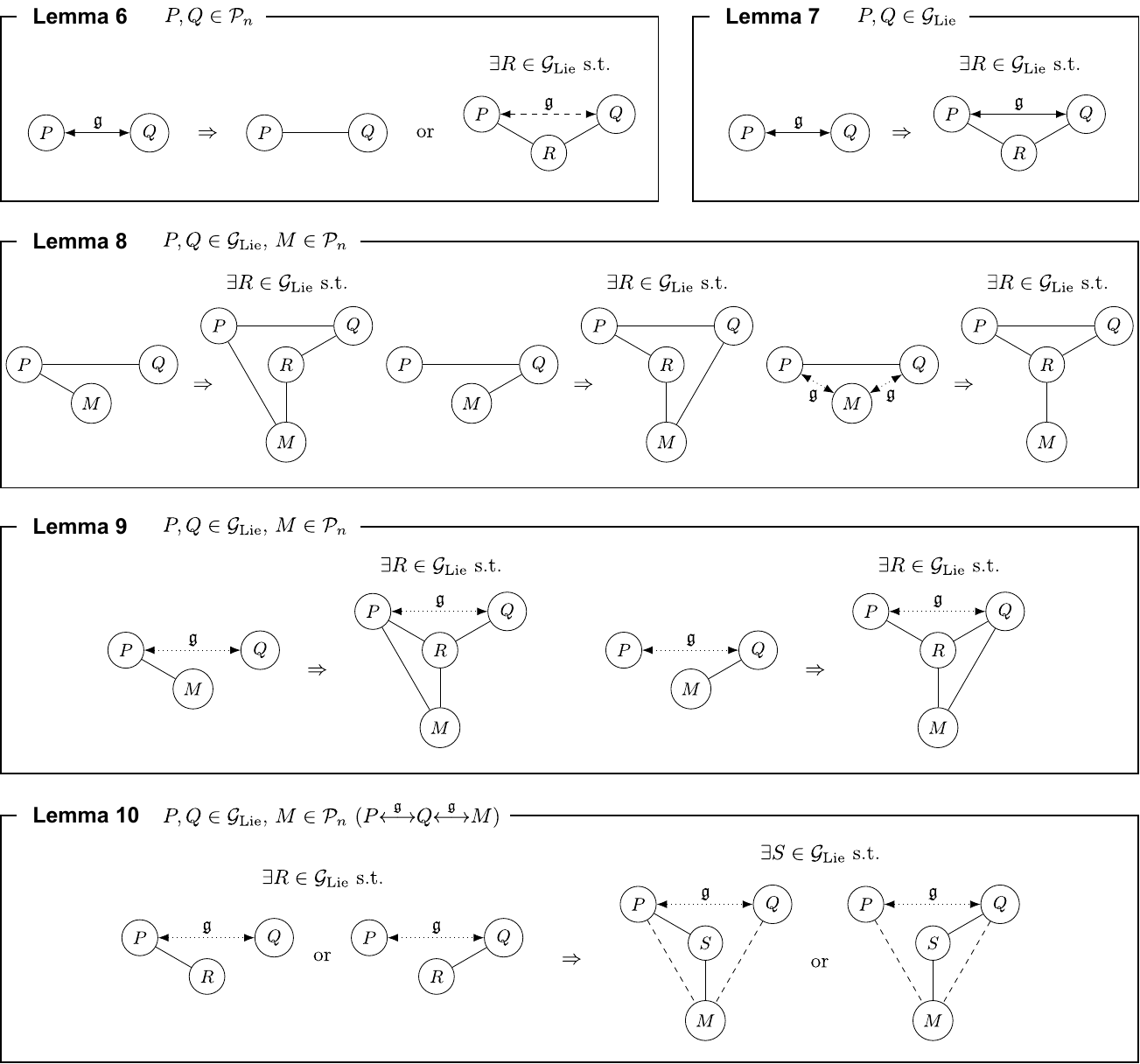}
		\caption{{\bf Summary of Lemmas~\ref{lem: shortest distance}--\ref{lem: PQSM}.}
		}
		\label{fig: lemma10--14}
	\end{figure*}

	\section{Lemmas for main theorem} \label{secap: lemmas1--5}
	
	This section proves Lemmas~\ref{thm_ap: not connected}--\ref{thm_ap: Gj=Gk} on the relationship between the simultaneous measurability of gradient components and the DLA structure, which have been used for the proof of Theorem~\ref{thm: main}.
	
	\subsection{Preliminaries}
	
	For preliminaries, we first show several relevant properties of the DLA graph to prove Lemmas~\ref{thm_ap: not connected}--\ref{thm_ap: Gj=Gk}.
	For the sake of convenience, we additionally define the path and distance on the DLA graph:
	\begin{dfn}[Path and distance]
		For $\mfg$-connected $P,Q\in\mc{P}_n$ ($P\neq Q$), consider $R_1,R_2,\cdots,R_{d-1}\in\mcG$ such that $\{P,R_1\}=\{R_1,R_2\}=\cdots=\{R_{d-1},Q\}=0$.
		Then, we call $P\to R_1\to \cdots \to R_{d-1}\to Q$ a path between $P$ and $Q$ and define its distance as $d$.
	\end{dfn}

	In what follows, we prove several lemmas on the DLA graph, which are summarized in Fig.~\ref{fig: lemma10--14}.
	Below, we ignore the coefficients of Pauli operators because only commutation and anti-commutation relations are relevant for the proof.
	\setcounter{lem}{5}
	\begin{lem} \label{lem: shortest distance}
		If $P, Q\in\mcPn$ ($P\neq Q$) are $\mfg$-connected, the shortest distance between them on the DLA graph is one or two.
	\end{lem}         

	\begin{proof}
		We prove this lemma by contradiction.
		Assume that the shortest distance between $P$ and $Q$ on the DLA graph is $d>2$, and let $R_1.\cdots,R_{d-1}\in\mcG$ be the nodes on the shortest path.
		For convenience, we denote $P$ and $Q$ by $R_0$ and $R_d$, respectively:
		\begin{equation*}
			\begin{tikzpicture}  
				\draw (2,0) -- (2.8,0);
				\draw[dashed] (2.8,0) -- (4.2,0);
				\draw (4.2,0) -- (5,0);
				\node[name=A][draw,circle,fill=white] at (0,0) {$R_0$}; 
				\node[name=B][draw,circle,fill=white] at (2,0) {$R_1$};
				\node[name=D][draw,circle,fill=white] at (5,0) {$R_{d}$};
				\draw (A) -- (B);
			\end{tikzpicture}.
		\end{equation*}
		Since this path is the shortest, these nodes are not connected by an edge except for the neighboring nodes:
		\begin{align}
			&\{R_j, R_k\} = 0 \quad \text{if} \quad |j-k|=1, \label{apeq: RjRk1} \\
			&[R_j, R_k] = 0 \quad \text{if} \quad |j-k|\neq 1. \label{apeq: RjRk2}
		\end{align}
		By definition of the DLA, the following nested commutator $R$ is included in $\mcG$:
		\begin{align}
			R 
			&= [R_{d-1},[\cdots,[R_3,[R_2,R_1]]\cdots]] \notag \\
			&= 2^{d-2} R_{d-1}\cdots R_1 \in \mcG,
		\end{align}
		where we have used Eqs.~\eqref{apeq: RjRk1} and \eqref{apeq: RjRk2} for the second equation.
		Then, $P\to R \to Q$ is a path of distance two because the Pauli operator $R\in\mcG$ anti-commutes with $P=R_0$ and $Q=R_d$.
		This contradicts the assumption that the shortest distance between $P$ and $Q$ is greater than two.
		Therefore, the shortest distance is one or two.
	\end{proof}

	\begin{lem} \label{lem: shortest distance2}
		If $P, Q\in \mcG$ ($P\neq Q$) are $\mfg$-connected, there exists $R\in\mcG$ such that $\{P,R\}=\{Q,R\}=0$.
	\end{lem}   

	\begin{proof}
		By Lemma~\ref{lem: shortest distance}, the shortest distance between $P$ and $Q$ on the DLA graph is one or two.
		If the shortest distance is two, there exists $R\in\mcG$ such that $\{P,R\}=\{Q,R\}=0$ by definition of the distance.
		If the shortest distance is one (i.e., $\{P,Q\}=0$), $R=[P,Q]=2PQ\in\mcG$ satisfies $\{P,R\}=\{Q,R\}=0$, as required.
	\end{proof}

	\begin{lem} \label{lem: graph-anticommute}
		For $P, Q\in \mcG$ and $M\in\mcPn$ satisfying $\{P,Q\}=0$ and $P\conn Q\conn M$, the following statements hold: 
		\begin{align*}
			\text{(i) If $\{P,M\}=[Q,M]=0$, there exists $R\in\mcG$}& \\
			\text{s.t. $[P,R]=\{Q,R\}=\{M,R\}=0$.}& \\
			\text{(ii) If $[P,M]=\{Q,M\}=0$, there exists $R\in\mcG$}& \\
			\text{s.t. $\{P,R\}=[Q,R]=\{M,R\}=0$.}& \\
			\text{(iii) If $[P,M]=[Q,M]=0$, there exists $R\in\mcG$}& \\
			\text{s.t. $\{P,R\}=\{Q,R\}=\{M,R\}=0$.}&
		\end{align*}
	\end{lem}        

	\begin{proof}
		We prove the three cases separately.
		
		\begin{enumerate}
			\item[(i)] If $\{P,M\}=[Q,M]=0$, $R=P\in\mcG$ satisfies $[P,R]=\{Q,R\}=\{M,R\}=0$:
			
			\item[(ii)] If $[P,M]=\{Q,M\}=0$, $R=Q\in\mcG$ satisfies $\{P,R\}=[Q,R]=\{M,R\}=0$:
			
			\item[(iii)] If $[P,M]=[Q,M]=0$, the shortest distance between $P$ and $M$ is two by Lemma~\ref{lem: shortest distance}, where we have used $P\neq M$ derived from $\{P,Q\}=[M,Q]=0$.
			Thus, there exists $S\in\mcG$ such that $\{P,S\}=\{M,S\}=0$:
			\begin{equation*}
				\begin{tikzpicture}
					\node[name=A][draw,circle,fill=white] at (-1.3,0.75) {$P$};
					\node[name=B][draw,circle,fill=white] at (+1.3,0.75) {$Q$};
					\node[name=C][draw,circle,fill=white] at (0,0) {$S$};
					\node[name=D][draw,circle,fill=white] at (0,-1.5) {$M$};
					\draw (A) -- (B);
					\draw (A) -- (C);
					\draw[dashed] (B) -- (C);
					\draw (C) -- (D);
				\end{tikzpicture}.
			\end{equation*}
			If $\{Q,S\}=0$, $R=S\in\mcG$ satisfies $\{P,R\}=\{Q,R\}=\{M,R\}=0$.
			If $[Q,S]=0$, on the other hand, $R=[P,S]=2PS\in\mcG$ satisfies $\{P,R\}=\{Q,R\}=\{M,R\}=0$.
			Therefore, in both cases, there exists $R\in\mcG$ such that $\{P,R\}=\{Q,R\}=\{M,R\}=0$.
			
		\end{enumerate}
	\end{proof}

	\begin{lem} \label{lem: graph-POQO}
		For $P, Q\in \mcG$ and $M\in\mcPn$ satisfying $[P,Q]=0$ and $P\conn Q\conn M$, if $\{P,M\}=[Q,M]=0$ or $[P,M]=\{Q,M\}=0$, there exists $R\in\mcG$ such that $\{P,R\}=\{Q,R\}=\{M,R\}=0$.
	\end{lem}   

	\begin{proof}
		Consider the case of $\{P,M\}=[Q,M]=0$ (the other case is similarly provable).
		Given that $P\neq Q$ holds from $\{P,M\}=[Q,M]=0$, there exists $S\in\mcG$ such that $\{P,S\}=\{Q,S\}=0$ by Lemma~\ref{lem: shortest distance2}, where $M\neq S$ because of $[Q,M]=\{Q,S\}=0$:
		\begin{equation*}
			\begin{tikzpicture}[every node/.style={fill=white}]
				\node[name=A][draw,circle] at (0,0) {$S$};
				\node[name=B][draw,circle] at (+1.3,0.75) {$Q$};
				\node[name=C][draw,circle] at (-1.3,0.75) {$P$};
				\node[name=D][draw,circle] at (0,-1.5) {$M$};
				\draw (A) -- (B);
				\draw (A) -- (C);
				\draw[dashed] (A) -- (D);
				\draw (C) -- (D);
			\end{tikzpicture}.
		\end{equation*}
		Lemma~\ref{lem: shortest distance} states that the shortest distance between $M$ and $S$ is one or two.
		If the shortest distance is one, $R=S$ satisfies $\{P,R\}=\{Q,R\}=\{M,R\}=0$.
		If the shortest distance is two (i.e., $[M,S]=0$), $R=[S,P]=2SP\in \mcG$ satisfies $\{P,R\}=\{Q,R\}=\{M,R\}=0$.
		Therefore, there always exists $R\in\mcG$ such that $\{P,R\}=\{Q,R\}=\{M,R\}=0$.
	\end{proof}

	\begin{lem} \label{lem: PQSM}
		For $P, Q\in \mcG$ and $M\in\mcPn$ satisfying $[P,Q]=0$ and $P\conn Q\conn M$, if there exists $R\in\mcG$ such that $\{P,R\}=[Q,R]=0$ or $[P,R]=\{Q,R\}=0$, then there exists $S\in\mcG$ such that $\{P,S\}=[Q,S]=\{M,S\}=0$ or $[P,S]=\{Q,S\}=\{M,S\}=0$.
	\end{lem}      

	\begin{proof}
		
		Consider the case of $\{P,R\}=[Q,R]=0$ (the other case is similarly provable).
		We prove the lemma in two cases, (i) $R=M$ and (ii) $R\neq M$.
		\begin{enumerate}
			\item[(i)] $R=M$ \\
			We have $P\neq Q$ from $\{P,M\}=[Q,M]=0$.
			Thus, there exists $T\in\mcG$ such that $\{P,T\}=\{Q,T\}=0$ by Lemma~\ref{lem: shortest distance2}, where $M\neq T$ because of $[M,Q]=\{T,Q\}=0$:
			\begin{equation*}
				\begin{tikzpicture}[every node/.style={fill=white}]   
					\node[name=A][draw,circle] at (-1.3,0.75) {$P$};
					\node[name=B][draw,circle] at (+1.3,0.75) {$Q$};
					\node[name=C][draw,circle] at (0,0) {$M$};
					\node[name=D][draw,circle] at (0,1.5) {$T$};
					\draw (A) -- (C); 
					\draw[dashed] (C) -- (D); 
					\draw (A) -- (D);
					\draw (B) -- (D);
				\end{tikzpicture}.
			\end{equation*}
			If $[T,M]=0$, $S=[M,[P,T]]=4MPT\in\mcG$ satisfies $[P,S]=\{Q,S\}=\{M,S\}=0$ (note that $R=M$ leads to $M\in\mcG$).
			If $\{T,M\}=0$, $S=[M,T]=2MT\in\mcG$ satisfies $[P,S]=\{Q,S\}=\{M,S\}=0$.
			Therefore, the lemma is proven for $R=M$.
			
			\item[(ii)] $R\neq M$ \\
			By Lemma~\ref{lem: shortest distance}, the shortest distance between $R$ and $M$ is one or two.
			If the shortest distance is one, $S=R\in\mcG$ satisfies $\{P,S\}=[Q,S]=\{M,S\}=0$, i.e., the lemma holds:
			\begin{equation*}
				\begin{tikzpicture}[every node/.style={fill=white}]   
					\node[name=A][draw,circle] at (-1.3,0.75) {$P$};
					\node[name=B][draw,circle] at (+1.3,0.75) {$Q$};
					\node[name=C][draw,circle] at (0,0) {$R$};
					\node[name=D][draw,circle] at (0,-1.5) {$M$};
					\draw (A) -- (C); 
					\draw (C) -- (D); 
					\draw[dashed] (A) -- (D);
					\draw[dashed] (B) -- (D);
				\end{tikzpicture}.
			\end{equation*}
			When the shortest distance is two, let $T\in\mcG$ be the node connecting $R$ and $M$.
			Then, there are four patterns regarding the (anti-)commutation relations between $P,Q$ and $T$, namely $[P,T]=0$ or $\{P,T\}=0$ and $[Q,T]=0$ or $\{Q,T\}=0$, as follows:
			\begin{equation*}
				\begin{tikzpicture}[every node/.style={fill=white}]   
					\node[name=a][draw=white] at (0,1.5) {(a)};
					\node[name=A][draw,circle] at (-1.3,0.75) {$P$};
					\node[name=B][draw,circle] at (+1.3,0.75) {$Q$};
					\node[name=C][draw,circle] at (0,0) {$R$};
					\node[name=D][draw,circle] at (0,-1.5) {$T$};
					\node[name=E][draw,circle] at (0,-3) {$M$};
					\draw (A) -- (C); 
					\draw (C) -- (D); 
					\draw (D) -- (E);
					\draw[dashed] (A) -- (E);
					\draw[dashed] (B) -- (E);
				\end{tikzpicture}
				\quad
				\begin{tikzpicture}[every node/.style={fill=white}]   
					\node[name=a][draw=white] at (0,1.5) {(b)};
					\node[name=A][draw,circle] at (-1.3,0.75) {$P$};
					\node[name=B][draw,circle] at (+1.3,0.75) {$Q$};
					\node[name=C][draw,circle] at (0,0) {$R$};
					\node[name=D][draw,circle] at (0,-1.5) {$T$};
					\node[name=E][draw,circle] at (0,-3) {$M$};
					\draw (A) -- (C); 
					\draw (C) -- (D); 
					\draw (D) -- (E);
					\draw (A) -- (D);
					\draw[dashed] (A) -- (E);
					\draw[dashed] (B) -- (E);
				\end{tikzpicture}
			\end{equation*}
			\begin{equation*}
				\begin{tikzpicture}[every node/.style={fill=white}]   
					\node[name=a][draw=white] at (0,1.5) {(c)};
					\node[name=A][draw,circle] at (-1.3,0.75) {$P$};
					\node[name=B][draw,circle] at (+1.3,0.75) {$Q$};
					\node[name=C][draw,circle] at (0,0) {$R$};
					\node[name=D][draw,circle] at (0,-1.5) {$T$};
					\node[name=E][draw,circle] at (0,-3) {$M$};
					\draw (A) -- (C); 
					\draw (C) -- (D); 
					\draw (D) -- (E);
					\draw (B) -- (D);
					\draw[dashed] (A) -- (E);
					\draw[dashed] (B) -- (E);
				\end{tikzpicture}
				\quad
				\begin{tikzpicture}[every node/.style={fill=white}]   
					\node[name=a][draw=white] at (0,1.5) {(d)};
					\node[name=A][draw,circle] at (-1.3,0.75) {$P$};
					\node[name=B][draw,circle] at (+1.3,0.75) {$Q$};
					\node[name=C][draw,circle] at (0,0) {$R$};
					\node[name=D][draw,circle] at (0,-1.5) {$T$};
					\node[name=E][draw,circle] at (0,-3) {$M$};
					\draw (A) -- (C); 
					\draw (C) -- (D); 
					\draw (D) -- (E);
					\draw (A) -- (D);
					\draw (B) -- (D);
					\draw[dashed] (A) -- (E);
					\draw[dashed] (B) -- (E);
				\end{tikzpicture}.
				\label{apeq: 4patterns}
			\end{equation*} 
			We can concretely construct $S\in\mcG$ satisfying the lemma for these four patterns as follows: 
			\begin{enumerate}
				\item[(a)] $S=[R,T]=2RT\in\mcG$ satisfies $\{P,S\}=[Q,S]=\{M,S\}=0$.
				\item[(b)] $S=T\in\mcG$ satisfies $\{P,S\}=[Q,S]=\{M,S\}=0$.
				\item[(c)] $S=T\in\mcG$ satisfies $[P,S]=\{Q,S\}=\{M,S\}=0$.
				\item[(d)] $S=[R,T]=2RT\in\mcG$ satisfies $[P,S]=\{Q,S\}=\{M,S\}=0$.
			\end{enumerate}
			Therefore, the lemma is proven for $R\neq M$.
		\end{enumerate}
	\end{proof}

	\subsection{Proof of Lemmas~\ref{thm_ap: not connected}--\ref{thm_ap: Gj=Gk}}

	We are ready to prove Lemmas~\ref{thm_ap: not connected}--\ref{thm_ap: Gj=Gk} on the relationship between the simultaneous measurability of gradient components and the DLA structure.
	In what follows, we use the form of $\grad_j(\bt)=-i[\tilde{G}_j(\bt),\tilde{O}(\bt)]$ as the gradient operator [see Eq.~\eqref{apeq: gradient operator}].
	
	As discussed in Sec.~\ref{secap: problem formulation}, two distinct gradient components $\partial_j C(\bt)$ and $\partial_k C(\bt)$ can be measured simultaneously if $[\grad_j(\bt), \grad_k(\bt)]=0$ for all $\bt$.
	In the DLA graph, the gradient operator $\grad_j(\bt)$ corresponds to a node of the graph $G_j$.
	Below, we show that whether $[\grad_j(\bt), \grad_k(\bt)]=0$ is determined by the structural relationship between $G_j, G_k$, and $O$ on the DLA graph.

	\setcounter{lem}{0}
	\begin{lem} \label{thm_ap: not connected}
		If $G_j,G_k\in\mcGc$ are not $\mfg$-connected, $\partial_j C$ and $\partial_k C$ can be simultaneously measured:
		\begin{equation*}
			\begin{tikzpicture}
				\node[name=conn] at (0,0.3) {$\mfg$};
				\node[name=conn] at (0,0) {$\times$};
				\node[name=A][draw,circle,fill=white] at (-1,0) {$G_j$};
				\node[name=B][draw,circle,fill=white] at (+1,0) {$G_k$};
				\draw[{Latex[length=2mm]}-{Latex[length=2mm]}] (A) -- (B); 
			\end{tikzpicture}
			\quad
			\raise 2ex\hbox{$\Rightarrow$}
			\quad
			\raise 2ex\hbox{$[\grad_j,\grad_k]= 0 \quad \forall \bt$.}
		\end{equation*}
	\end{lem}    

	\begin{proof}
		When $G_j$ and $G_k$ are not $\mfg$-connected, the DLA graph $\mcG$ is separable into two subgraphs $\mcG{}_1$ and $\mcG{}_2$ (i.e., $[\mcG{}_1,\mcG{}_2]=0$).
		This is because if the DLA graph is not separable, then any two nodes are $\mfg$-connected.
		Therefore, we can decompose the DLA $\mfg$ into two subalgebras $\mfg_1=\text{span}(\mcG{}_1)$ and $\mfg_2=\text{span}(\mcG{}_2)$ that commute with each other:
		\begin{align}
			\mfg = \mfg_1 \oplus \mfg_2, \quad [\mfg_1,\mfg_2]=0, \label{apeq: g1g2}
		\end{align}
		where $G_j\in\mfg_1$ and $G_k\in\mfg_2$.
		Then, the unitaries of the quantum circuit are decomposed as $U=VW$, $U_{j+}=V_{j+}W_{j+}$, and $U_{k+}=V_{k+}W_{k+}$ with $V,V_{j+},V_{k+}\in e^{\mfg_1}$ and $W,W_{j+},W_{k+}\in e^{\mfg_2}$ ($U_{j+}$ and $U_{k+}$ are the unitary circuits before the $j$th and $k$th gates, respectively).

		Now, we prove $[\grad_j,\grad_k]= 0$ using these decompositions.
		A straightforward calculation shows that the commutator of gradient operators $\grad_j=-i[\tilde{G}_j,\tilde{O}]$ and $\grad_k=-i[\tilde{G}_k,\tilde{O}]$ is written as
		\begin{align}
			&[\grad_j,\grad_k] \notag \\
			&= -[\tilde{G}_j,\tilde{O}\tilde{G}_k\tilde{O}] - [\tilde{O}\tilde{G}_j\tilde{O},\tilde{G}_k] + [\tilde{G}_j,\tilde{G}_k] + \tilde{O}[\tilde{G}_j,\tilde{G}_k]\tilde{O}. \label{apeq: gradj gradk} 
		\end{align}
		Below, we show that each term in the above commutator vanishes.
		Let us first prove $[\tilde{G}_j,\tilde{G}_k]=0$.
		From the decomposition of $U_{j+}=V_{j+}W_{j+}$, we have
		\begin{align}
			\tilde{G}_j 
			&= U_{j+}^\dag G_j U_{j+} \notag \\
			&= W_{j+}^\dag V_{j+}^\dag G_j V_{j+}W_{j+} \notag \\
			&= V_{j+}^\dag G_j V_{j+} \in \mfg_1, \label{apeq: Pj_in_g1}
		\end{align}
		where we have used $[V_{j+},W_{j+}]=[G_j,W_{j+}]=0$.
		In the same manner, we have 
		\begin{align}
			\tilde{G}_k 
			&= W_{k+}^\dag G_k W_{k+} \in \mfg_2. \label{apeq: Pk_in_g2}
		\end{align}
		Thereby, $[\tilde{G}_j,\tilde{G}_k]=0$ is proven from $[\mfg_1,\mfg_2]=0$, indicating that the third and fourth terms vanish in Eq.~\eqref{apeq: gradj gradk}.
		Next, we prove $[\tilde{G}_j,\tilde{O}\tilde{G}_k\tilde{O}]=0$.
		Since $\tilde{O}=U^\dag O U$ and $\tilde{G}_k= W_{k+}^\dag G_k W_{k+}$ [see Eq.~\eqref{apeq: Pk_in_g2}], $\tilde{O}\tilde{G}_k\tilde{O}$ is written as
		\begin{align}
			\tilde{O}\tilde{G}_k\tilde{O}
			&= (W^\dag V^\dag O VW) (W_{k+}^\dag G_k W_{k+}) (W^\dag V^\dag O VW) \notag \\
			&= W^\dag V^\dag O V A V^\dag O VW \notag \\
			&= W^\dag V^\dag O A O VW
		\end{align}
		where we have defined $A=W W_{k+}^\dag G_k W_{k+} W^\dag\in\mfg_2$ and used $[V,A]=0$.
		Here, for any $w = \sum_{i} c_i g_i \in \mfg_2$, $OwO = \sum_{i} (-1)^{\sigma_i} c_i g_i$ is also included in $\mfg_2$, where $g_i$ is the Pauli basis of $\mfg_2$, $c_i$ is an expansion coefficient, and $\sigma_i$ is defined as $Og_i O=(-1)^{\sigma_i}g_i$. 
		Therefore, $A' = O A O$ is included in $\mfg_2$ because of $A\in\mfg_2$, leading to
		\begin{align}
			\tilde{O}\tilde{G}_k\tilde{O}
			&=W^\dag V^\dag A' VW = W^\dag A' W \in \mfg_2, \label{apeq: OPkO_in_g2}
		\end{align}
		where we have used $[A',V]=0$.
		Thus, $[\tilde{G}_j,\tilde{O}\tilde{G}_k\tilde{O}]=0$ holds from Eqs.\eqref{apeq: Pj_in_g1}, \eqref{apeq: OPkO_in_g2}, and $[\mfg_1,\mfg_2]=0$. 
		Similarly, $[\tilde{O}\tilde{G}_j\tilde{O},\tilde{G}_k]=0$ can also be proven.
		Therefore, the first and second terms vanish in Eq.~\eqref{apeq: gradj gradk}.
		These results prove that $[\grad_j, \grad_k]=0$ holds for any $\bt$, i.e., $\partial_j C$ and $\partial_k C$ are simultaneously measurable.
	\end{proof}


	Lemma~\ref{thm_ap: not connected} states that $\partial_j C$ and $\partial_k C$ can be simultaneously measured if $G_j$ and $G_k$ are not $\mfg$-connected.
	Thus, we focus on the case that $G_j, G_k\in\mcGc$ are $\mfg$-connected below.
	To simplify the proof, we divide the quantum circuit into two parts, $U=U_\text{fin} U_\text{ini}$, where we have defined $U_\text{ini} = \prod_{j=1}^{\Lbulk} e^{i\theta_jG_j}$ and $U_\text{fin} = \prod_{j=\Lbulk+1}^{L} e^{i\theta_jG_j}$.
	The depth of the final part, $\LB=L-\Lbulk$, is set as a sufficiently large but constant value such that $U_\text{fin}$ can express $e^{i\phi_1 Q_1}e^{i\phi_2 Q_2}$ for any $Q_1, Q_2\in \mcG$ and $\phi_1,\phi_2\in\mathbb{R}$.
	There exists such $\LB$ due to Condition~\ref{cond: 3}.
	We emphasize that whether the gradient components for $U_\text{fin}$ can be simultaneously measured does not affect the gradient measurement efficiency for the full circuit, $\ef$, because the contribution from the constant depth circuit to $\ef$ is negligible in the deep circuit limit.
	Therefore, we will use $U_\text{fin}$ as a buffer circuit for the proof and investigate whether $\partial_j C$ and $\partial_k C$ are simultaneously measurable for the parameters of $U_\text{ini}$ (i.e., $j,k\leq \Lbulk$).
	Then,  the gradient operators are given by
	\begin{align}
		\begin{split}
			&\grad_j = -i[\tilde{G}_j,\tilde{O}] = -i\left[\left(U^\dag_\text{ini} G_j U_\text{ini}\right), \left(U^\dag_\text{ini}U_\text{fin}^\dag O U_\text{fin}U_\text{ini}\right)\right],  \\
			&\grad_k = -i[\tilde{G}_k,\tilde{O}] = -i\left[\left(U^\dag_\text{ini} G_k U_\text{ini}\right), \left(U^\dag_\text{ini}U_\text{fin}^\dag O U_\text{fin}U_\text{ini}\right)\right]. 
		\end{split} \label{apeq: grad_UiniI}
	\end{align}

	In our definition, $\partial_j C$ and $\partial_k C$ can be measured simultaneously if $[\grad_j(\bt),\grad_k(\bt)]= 0$ for all $\bt$.
	Therefore, when proving that $\partial_j C$ and $\partial_k C$ cannot be measured simultaneously, it suffices to show that $[\grad_j(\bt),\grad_k(\bt)]\neq 0$ for some $\bt$.
	In the following, we find such $\bt$ to prove the lemmas.

	\begin{lem} \label{thm_ap: PjPk}
		For $j,k \leq \Lbulk$, if $G_j,G_k\in\mcGc$ anti-commute, $\partial_j C$ and $\partial_k C$ cannot be simultaneously measured:
		\begin{equation*}
			\begin{tikzpicture}[every node/.style={fill=white}]   
				\node[name=A][draw,circle] at (-1,0) {$G_j$};
				\node[name=B][draw,circle] at (+1,0) {$G_k$};
				\draw (A) -- (B);  
			\end{tikzpicture}
			\quad
			\raise 2ex\hbox{$\Rightarrow$}
			\quad
			\raise 2ex\hbox{$\exists \bt \,\,\, \text{s.t.} \,\,\, [\grad_j,\grad_k]\neq 0$.}
		\end{equation*}
	\end{lem}    

	\begin{proof}

		We prove this lemma in two cases: (i) $\{G_j,O\}=\{G_k,O\}=0$ and (ii) otherwise.
		\begin{enumerate}
			\item[(i)]
			In this case, $[\grad_j(\bt),\grad_k(\bt)]\neq 0$ holds for $\bt=0$.
			In fact, when $\bt=0$ (i.e., $U_\text{ini}=U_\text{fin}=I$), we have $\tilde{G}_j=G_j, \tilde{G}_k=G_k, \tilde{O}=O$, and thus
			\begin{align}
				&\grad_j = -i[\tilde{G}_j, \tilde{O}] = -2i G_jO, \\
				&\grad_k = -i[\tilde{G}_k, \tilde{O}] = -2i G_kO.
			\end{align}
			Therefore, we obtain $[\grad_j,\grad_k]=4[G_j,G_k] \neq 0$ from $\{G_j,O\}=\{G_k,O\}=0$, showing that $\partial_j C$ and $\partial_k C$ cannot be simultaneously measured.
			
			\item[(ii)] 
			By Lemma~\ref{lem: graph-anticommute}, there exists $R\in\mcG$ satisfying one of the following commutation relations:
			\begin{equation*}
				\begin{tikzpicture}
					\node[name=A][draw,circle,fill=white] at (-1.3,0.75) {$G_j$};
					\node[name=B][draw,circle,fill=white] at (+1.3,0.75) {$G_k$};
					\node[name=C][draw,circle,fill=white] at (0,0) {$R$};
					\node[name=D][draw,circle,fill=white] at (0,-1.5) {$O$};
					\draw (A) -- (B);
					\draw (A) -- (D);
					\draw (B) -- (C);
					\draw (C) -- (D);
				\end{tikzpicture}
				\quad
				\begin{tikzpicture}
					\node[name=A][draw,circle,fill=white] at (-1.3,0.75) {$G_j$};
					\node[name=B][draw,circle,fill=white] at (+1.3,0.75) {$G_k$};
					\node[name=C][draw,circle,fill=white] at (0,0) {$R$};
					\node[name=D][draw,circle,fill=white] at (0,-1.5) {$O$};
					\draw (A) -- (B);
					\draw (A) -- (C);
					\draw (B) -- (D);
					\draw (C) -- (D);
				\end{tikzpicture}
			\end{equation*}
			\begin{equation*}
				\begin{tikzpicture}
					\node[name=A][draw,circle,fill=white] at (-1.3,0.75) {$G_j$};
					\node[name=B][draw,circle,fill=white] at (+1.3,0.75) {$G_k$};
					\node[name=C][draw,circle,fill=white] at (0,0) {$R$};
					\node[name=D][draw,circle,fill=white] at (0,-1.5) {$O$};
					\draw (A) -- (B);
					\draw (A) -- (C);
					\draw (B) -- (C);
					\draw (C) -- (D);
				\end{tikzpicture}.
			\end{equation*}
			By setting $\bt$ such that $U_\text{ini}=I$ and $U_\text{fin}=e^{-i\pi R/4}$ in Eq.~\eqref{apeq: grad_UiniI}, we have $\tilde{G}_j=G_j, \tilde{G}_k=G_k, \tilde{O}=e^{i\pi R/4}Oe^{-i\pi R/4}=iRO$, and thus
			\begin{align}
				&\grad_j = -i[\tilde{G}_j, \tilde{O}] = 2G_jRO, \\
				&\grad_k = -i[\tilde{G}_k, \tilde{O}] = 2G_kRO.
			\end{align}
			Therefore, we obtain $[\grad_j,\grad_k]=4[G_k,G_j] \neq 0$ from $\{G_j,iRO\}=\{G_k,iRO\}=0$, showing that $\partial_j C$ and $\partial_k C$ cannot be simultaneously measured.
		\end{enumerate}
		These results prove that $\partial_j C$ and $\partial_k C$ cannot be simultaneously measured for anti-commuting $G_j$ and $G_k$.
	\end{proof}

	\begin{lem} \label{thm_ap: [P,O]}
		Consider $\mfg$-connected $G_j,G_k\in\mcGc$ for $j,k \leq \Lbulk$.
		If $\{G_j,O\}=[G_k,O]=0$ or $[G_j,O]=\{G_k,O\}=0$, $\partial_j C$ and $\partial_k C$ cannot be simultaneously measured:
		\begin{equation*}
			\begin{tikzpicture}
				\node[name=conn] at (0,0.95) {$\mfg$};
				\node[name=A][draw,circle,fill=white] at (-1.3,0.75) {$G_j$};
				\node[name=B][draw,circle,fill=white] at (+1.3,0.75) {$G_k$};
				\node[name=C][draw,circle,fill=white] at (0,0) {$O$};
				\draw (A) -- (C); 
				\draw[{Latex[length=2mm]}-{Latex[length=2mm]}] (A) -- (B); 
			\end{tikzpicture}
			\quad
			\raise 4ex\hbox{or}
			\quad
			\begin{tikzpicture}
				\node[name=conn] at (0,0.95) {$\mfg$};
				\node[name=A][draw,circle,fill=white] at (-1.3,0.75) {$G_j$};
				\node[name=B][draw,circle,fill=white] at (+1.3,0.75) {$G_k$};
				\node[name=C][draw,circle,fill=white] at (0,0) {$O$};
				\draw (B) -- (C);  
				\draw[{Latex[length=2mm]}-{Latex[length=2mm]}] (A) -- (B); 
			\end{tikzpicture}
		\end{equation*}
		\begin{equation*}
			\raise 4ex\hbox{$\Rightarrow \quad \exists \bt \,\,\, \text{s.t.} \,\,\, [\grad_j,\grad_k]\neq 0$.}
		\end{equation*}
	\end{lem}    

	\begin{proof}

		It suffices to consider only the case of $[G_j,G_k]=0$ by Lemma~\ref{thm_ap: PjPk}.
		
		Here, we consider the case of $\{G_j,O\}=[G_k,O]=0$ (the other case is similarly provable).
		Then, there exists $R\in\mcG$ such that $\{G_j,R\}=\{G_k,R\}=\{O,R\}=0$ by Lemma~\ref{lem: graph-POQO}:
		\begin{equation*}
			\begin{tikzpicture}[every node/.style={fill=white}]
				\node[name=A][draw,circle] at (0,0) {$R$};
				\node[name=B][draw,circle] at (+1.3,0.75) {$G_k$};
				\node[name=C][draw,circle] at (-1.3,0.75) {$G_j$};
				\node[name=D][draw,circle] at (0,-1.5) {$O$};
				\draw (A) -- (B);
				\draw (A) -- (C);
				\draw (A) -- (D);
				\draw (C) -- (D);
				\node[name=conn] at (0,0.95) {$\mfg$};
				\draw[dotted, {Latex[length=2mm]}-{Latex[length=2mm]}] (B) -- (C); 
			\end{tikzpicture}.
		\end{equation*}
		By setting $\bt$ such that $U_\text{ini}=I$ and $U_\text{fin}=e^{i\phi R}$, we have $\tilde{G}_j=G_j$, $\tilde{G}_k=G_k$, $\tilde{O}=e^{-i\phi R}Oe^{i\phi R}=\cos(2\phi)O + i\sin(2\phi)OR$, and thus
		\begin{align}
			\grad_j 
			&= -i[\tilde{G}_j, \tilde{O}] \notag \\
			&= -i \cos(2\phi)[G_j,O] + \sin(2\phi)[G_j,OR] \notag \\
			&= - 2i \cos(2\phi) G_j O  \notag \\
			\grad_k 
			&= -i[\tilde{G}_k, \tilde{O}] \notag \\
			&= -i \cos(2\phi)[G_k,O] + \sin(2\phi)[G_k,OR] \notag \\
			&= 2 \sin(2\phi) G_k O R. \notag 
		\end{align}
		Therefore, we have
		\begin{align}
			&[\grad_j,\grad_k] = -4i\cos(2\phi) \sin(2\phi) [G_j O, G_kOR].
		\end{align}
		From the commutation relations between $G_j,G_k,R$ and $O$, we have $[G_j O, G_kOR]\neq 0$, showing that there exists $\bt$ such that $[\grad_j,\grad_k]\neq 0$.
		Therefore, $\partial_j C$ and $\partial_k C$ cannot be simultaneously measured if $\{G_j,O\}=[G_k,O]=0$ or $[G_j,O]=\{G_k,O\}=0$.
	\end{proof}

	\begin{lem} \label{thm_ap: PjR&PkR}
		Consider $\mfg$-connected $G_j, G_k\in\mcGc$ for $j,k \leq \Lbulk$.
		If there exists $R\in\mcG$ such that $\{G_j,R\}=[G_k,R]=0$ or $[G_j,R]=\{G_k,R\}=0$, $\partial_j C$ and $\partial_k C$ cannot be simultaneously measured:
		\begin{equation*}
			\begin{tikzpicture}
				\node[name=conn] at (0,0.95) {$\mfg$};
				\node[name=A][draw,circle,fill=white] at (-1.3,0.75) {$G_j$};
				\node[name=B][draw,circle,fill=white] at (+1.3,0.75) {$G_k$};
				\node[name=C][draw,circle,fill=white] at (0,0) {$R$};
				\draw (A) -- (C); 
				\draw[{Latex[length=2mm]}-{Latex[length=2mm]}] (A) -- (B); 
			\end{tikzpicture}
			\quad
			\raise 4ex\hbox{or}
			\quad
			\begin{tikzpicture}
				\node[name=conn] at (0,0.95) {$\mfg$};
				\node[name=A][draw,circle,fill=white] at (-1.3,0.75) {$G_j$};
				\node[name=B][draw,circle,fill=white] at (+1.3,0.75) {$G_k$};
				\node[name=C][draw,circle,fill=white] at (0,0) {$R$};
				\draw (B) -- (C);  
				\draw[{Latex[length=2mm]}-{Latex[length=2mm]}] (A) -- (B); 
			\end{tikzpicture}
		\end{equation*}
		\begin{equation*}
			\raise 4ex\hbox{$\Rightarrow \quad \exists \bt \,\,\, \text{s.t.} \,\,\, [\grad_j,\grad_k]\neq 0$.}
		\end{equation*}
	\end{lem}    

	\begin{proof}
		
		It suffices to consider only the case of $[G_j,G_k]=0$ and $[G_j,O]=[G_k,O]=0$ or $\{G_j,O\}=\{G_k,O\}=0$ by Lemmas~\ref{thm_ap: PjPk} and \ref{thm_ap: [P,O]}.

		When $\{G_j,R\}=[G_k,R]=0$ or $[G_j,R]=\{G_k,R\}=0$, there necessarily exists $S\in\mcG$ satisfying $\{G_j,S\}=[G_k,S]=\{O,S\}=0$ or $[G_j,S]=\{G_k,S\}=\{O,S\}=0$ by Lemma~\ref{lem: PQSM}:
		\begin{equation*}
			\begin{tikzpicture}[every node/.style={fill=white}]
				\node[name=A][draw,circle] at (-1.3,0.75) {$G_j$};
				\node[name=B][draw,circle] at (+1.3,0.75) {$G_k$};
				\node[name=C][draw,circle] at (0,0) {$S$};
				\node[name=D][draw,circle] at (0,-1.5) {$O$};
				\draw (A) -- (C);
				\draw (C) -- (D);
				\draw[dashed] (A) -- (D);
				\draw[dashed] (B) -- (D);
				\node[name=conn] at (0,0.95) {$\mfg$};
				\draw[dotted, {Latex[length=2mm]}-{Latex[length=2mm]}] (A) -- (B); 
			\end{tikzpicture}
			\quad
			\raise 8ex\hbox{or}
			\quad
			\begin{tikzpicture}[every node/.style={fill=white}]
				\node[name=A][draw,circle] at (-1.3,0.75) {$G_j$};
				\node[name=B][draw,circle] at (+1.3,0.75) {$G_k$};
				\node[name=C][draw,circle] at (0,0) {$S$};
				\node[name=D][draw,circle] at (0,-1.5) {$O$};
				\draw (B) -- (C);
				\draw (C) -- (D);
				\draw[dashed] (A) -- (D);
				\draw[dashed] (B) -- (D);
				\node[name=conn] at (0,0.95) {$\mfg$};
				\draw[dotted, {Latex[length=2mm]}-{Latex[length=2mm]}] (A) -- (B); 
			\end{tikzpicture}.
		\end{equation*}
		Assume $\{G_j,S\}=[G_k,S]=\{O,S\}=0$ (the other case is similarly provable).
		Then, we consider $O'=iSO\in\mcPn$, which satisfies $\{G_j,O'\}=[G_k,O']=0$ or $[G_j,O']=\{G_k,O'\}=0$ ($\because$ $[G_j,O]=[G_k,O]=0$ or $\{G_j,O\}=\{G_k,O\}=0$):
		\begin{equation*}
			\begin{tikzpicture}
				\node[name=conn] at (0,0.95) {$\mfg$};
				\node[name=A][draw,circle,fill=white] at (-1.3,0.75) {$G_j$};
				\node[name=B][draw,circle,fill=white] at (+1.3,0.75) {$G_k$};
				\node[name=C][draw,circle,fill=white] at (0,0) {$O'$};
				\draw (A) -- (C); 
				\draw[dotted, {Latex[length=2mm]}-{Latex[length=2mm]}] (A) -- (B); 
			\end{tikzpicture}
			\quad
			\raise 4ex\hbox{or}
			\quad
			\begin{tikzpicture}
				\node[name=conn] at (0,0.95) {$\mfg$};
				\node[name=A][draw,circle,fill=white] at (-1.3,0.75) {$G_j$};
				\node[name=B][draw,circle,fill=white] at (+1.3,0.75) {$G_k$};
				\node[name=C][draw,circle,fill=white] at (0,0) {$O'$};
				\draw (B) -- (C);  
				\draw[dotted, {Latex[length=2mm]}-{Latex[length=2mm]}] (A) -- (B); 
			\end{tikzpicture}.
		\end{equation*}
		We assume $\{G_j,O'\}=[G_k,O']=0$ (the other case is similarly provable).
		By Lemma~\ref{lem: graph-POQO}, there exists $T\in\mcG$ satisfying $\{G_j,T\}=\{G_k,T\}=\{O',T\}=0$:
		\begin{equation*}
			\begin{tikzpicture}[every node/.style={fill=white}]
				\node[name=A][draw,circle] at (0,0) {$T$};
				\node[name=B][draw,circle] at (+1.3,0.75) {$G_k$};
				\node[name=C][draw,circle] at (-1.3,0.75) {$G_j$};
				\node[name=D][draw,circle] at (0,-1.5) {$O'$};
				\draw (A) -- (B);
				\draw (A) -- (C);
				\draw (A) -- (D);
				\draw (C) -- (D);
				\node[name=conn] at (0,0.95) {$\mfg$};
				\draw[dotted, {Latex[length=2mm]}-{Latex[length=2mm]}] (B) -- (C); 
			\end{tikzpicture}.
		\end{equation*}
		We set $\bt$ such that $U_\text{ini}=I$ and $U_\text{fin}=e^{-i\pi S/4}e^{i\phi T}$, having $\tilde{G}_j=G_j$, $\tilde{G}_k=G_k$, $\tilde{O}=e^{-i\phi T}e^{i\pi S/4}Oe^{-i\pi S/4}e^{i\phi T}=\cos(2\phi)O' + i\sin(2\phi)O'T$, and thus
		\begin{align}
			\grad_j 
			&= -i[\tilde{G}_j,\tilde{O}] \notag \\
			&= -i \cos(2\phi) [G_j,O'] + \sin(2\phi) [G_j,O'T] \\
			&= -2i \cos(2\phi) G_j O' \notag \\
			\grad_k 
			&= -i [\tilde{G}_k,\tilde{O}] \notag \\
			&= -i \cos(2\phi) [G_k,O'] + \sin(2\phi) [G_k,O'T] \\
			&= 2 \sin(2\phi) G_k O'T. \notag 
		\end{align}
		Therefore, the gradient operators satisfy
		\begin{align}
			&[\grad_j,\grad_k] = -4i\cos(2\phi)\sin(2\phi) [G_j O', G_kO'T].
		\end{align}
		From the commutation relations between $G_j,G_k,O'$ and $T$, we have $[G_j O', G_kO'T]\neq 0$, showing that there exists $\bt$ such that $[\grad_j,\grad_k]\neq 0$.
		Therefore, $\partial_j C$ and $\partial_k C$ cannot be simultaneously measured if there exists $R\in\mcG$ such that $\{G_j,R\}=[G_k,R]=0$ or $[G_j,R]=\{G_k,R\}=0$.
	\end{proof}

	\begin{lem}\label{thm_ap: Gj=Gk}
		For $j<k \leq \Lbulk$, if $G_j=G_k$, $\partial_j C$ and $\partial_k C$ cannot be simultaneously measured:
		\begin{equation*}
			\begin{tikzpicture}
				\node[name=conn] at (0,0) {$=$};
				\node[name=A][draw,circle,fill=white] at (-0.75,0) {$G_j$};
				\node[name=B][draw,circle,fill=white] at (+0.75,0) {$G_k$};
			\end{tikzpicture}
			\quad
			\raise 2ex\hbox{$\Rightarrow$}
			\quad
			\raise 2ex\hbox{$\exists \bt \,\,\, \text{s.t.} \,\,\, [\grad_j,\grad_k]\neq 0$.}
		\end{equation*}
	\end{lem}    

	\begin{proof}
		
		By Condition~\ref{cond: 1}, there exists $\ell$ ($j<\ell<k$) such that $\{G_j,G_\ell\}=0$.
		Then we consider $G_k' =iG_\ell G_k = i[G_\ell,G_k]/2\in \mcG$, which satisfies $\{G_j,G_k'\}=0$.
		We prove this lemma in two cases: (i) $\{G_j,O\}=\{G_k',O\}=0$ and (ii) otherwise.
		\begin{enumerate}
			\item[(i)] 
			We set $\bt$ such that $U_{j+}=I$, $U_{k+}=e^{-i\pi G_\ell/4}$, and $U=I$ (i.e., $U_\text{ini}=e^{-i\pi G_\ell/4}$ and $U_\text{fin}=e^{+i\pi G_\ell/4}$).
			Then we have $\tilde{G}_j=G_j, \tilde{G}_k=e^{i\pi G_\ell/4}G_k e^{-i\pi G_\ell/4}=G_k', \tilde{O}=O$, and thus
			\begin{align}
				&\grad_j = -i[\tilde{G}_j, \tilde{O}] = -2i G_jO, \\
				&\grad_k = -i[\tilde{G}_k, \tilde{O}] = -2i G_k'O.
			\end{align}
			Therefore, we obtain $[\grad_j,\grad_k]=4[G_j,G_k'] \neq 0$ from $\{G_j,O\}=\{G_k',O\}=0$, showing that $\partial_j C$ and $\partial_k C$ cannot be simultaneously measured.
			
			\item[(ii)]
			By Lemma~\ref{lem: graph-anticommute}, there exists $R\in\mcG$ satisfying one of the following commutation relations:
			\begin{equation*}
				\begin{tikzpicture}
					\node[name=A][draw,circle,fill=white] at (-1.3,0.75) {$G_j$};
					\node[name=B][draw,circle,fill=white] at (+1.3,0.75) {$G_k'$};
					\node[name=C][draw,circle,fill=white] at (0,0) {$R$};
					\node[name=D][draw,circle,fill=white] at (0,-1.5) {$O$};
					\draw (A) -- (B);
					\draw (A) -- (D);
					\draw (B) -- (C);
					\draw (C) -- (D);
				\end{tikzpicture}
				\quad
				\begin{tikzpicture}
					\node[name=A][draw,circle,fill=white] at (-1.3,0.75) {$G_j$};
					\node[name=B][draw,circle,fill=white] at (+1.3,0.75) {$G_k'$};
					\node[name=C][draw,circle,fill=white] at (0,0) {$R$};
					\node[name=D][draw,circle,fill=white] at (0,-1.5) {$O$};
					\draw (A) -- (B);
					\draw (A) -- (C);
					\draw (B) -- (D);
					\draw (C) -- (D);
				\end{tikzpicture}
			\end{equation*}
			\begin{equation*}
				\begin{tikzpicture}
					\node[name=A][draw,circle,fill=white] at (-1.3,0.75) {$G_j$};
					\node[name=B][draw,circle,fill=white] at (+1.3,0.75) {$G_k'$};
					\node[name=C][draw,circle,fill=white] at (0,0) {$R$};
					\node[name=D][draw,circle,fill=white] at (0,-1.5) {$O$};
					\draw (A) -- (B);
					\draw (A) -- (C);
					\draw (B) -- (C);
					\draw (C) -- (D);
				\end{tikzpicture}.
			\end{equation*}
			By setting $\bt$ such that $U_{j+}=I$, $U_{k+}=e^{-i\pi G_\ell/4}$, and $U=e^{-i\pi R/4}$ (i.e., $U_\text{ini}=e^{-i\pi G_\ell/4}$, and $U_\text{fin}=e^{-i\pi R/4}e^{+i\pi G_\ell/4}$), we have $\tilde{G}_j=G_j, \tilde{G}_k=G_k', \tilde{O}=e^{i\pi R/4}Oe^{-i\pi R/4}=iRO$, and thus
			\begin{align}
				&\grad_j = -i[\tilde{G}_j, \tilde{O}] = 2G_jRO, \\
				&\grad_k = -i[\tilde{G}_k, \tilde{O}] = 2G_k'RO.
			\end{align}
			Therefore, we obtain $[\grad_j,\grad_k]=4[G_j,G_k'] \neq 0$ from $\{G_j,iRO\}=\{G_k',iRO\}=0$, showing that $\partial_j C$ and $\partial_k C$ cannot be simultaneously measured.
		\end{enumerate}
		These results prove that $\partial_j C$ and $\partial_k C$ cannot be simultaneously measured if $G_j=G_k$.
	\end{proof}

	\section{Inequalities for main theorem} \label{secap: proof eqs}

	The main proof of Theorem~\ref{thm: main} leaves the derivations of Eqs.~\eqref{apeq: ex_inequality} and \eqref{apeq: sub_inequality}.
	Here, we derive these equations.

	\subsection{Preliminaries} \label{secap: preliminaries}
	We first show several lemmas for preliminaries.
	
	\setcounter{lem}{10}
	\begin{lem} \label{lem: stab1}
		Consider a set of commuting Pauli operators $\mcS=\{S_1,S_2,\cdots\}$, where $[S_j,S_k]=0$ $\forall S_j,S_k\in\mcS$.
		Define a subset of Pauli operators, $\mc{P}_\mcS \subset \mc{P}_n$, stabilized by $\mcS$ as
		\begin{align}
			\mc{P}_\mcS = \{P\in \mc{P}_n \,|\, [P,S_j]=0 \,\, \forall S_j\in\mcS\}.
		\end{align}
		Then, letting $s$ be the number of independent Pauli operators in $\mcS$, the following equality and inequalities hold:
		\begin{align}
			&|\mc{P}_\mcS| = \frac{4^n}{2^s} \leq \frac{4^n}{|\mcS|}, \quad |\mcS|\leq 2^n.
		\end{align}
	\end{lem}    

	\begin{proof}
		Let $g_1,\cdots,g_s\in\mcS$ be independent Pauli operators in $\mcS$.
		According to Proposition~10.4 in Ref.~\cite{Nielsen2010-pf}, there is $h_j\in \mc{P}_n$ such that $\{g_j,h_j\}=0$ and $[g_k,h_j]=0$ for all $k\neq j$.
		We consider 
		\begin{align}
			\mc{H}({\bm{x}}) 
			&\equiv h_1^{x_1}\cdots h_s^{x_s} \mc{P}_\mcS \notag \\
			&= \{h_1^{x_1}\cdots h_s^{x_s} P \,|\, P\in \mc{P}_\mcS\} \subset \mcPn,
		\end{align}
		where $\bm{x}=(x_1,\cdots,x_s)$ is a binary vector ($x_j=0,1$).
		The set $\mc{H}(\bm{x})$ has the following properties:
		\begin{enumerate}
			\item $\mc{H}(\bm{x}) \cap \mc{H}(\bm{x'})=\varnothing$ for $\bm{x}\neq\bm{x}'$.
			\item $|\mc{H}(\bm{x})| = |\mc{P}_\mcS|$ for any $\bm{x}$
			\item $\bigsqcup_{\bm{x}} \mc{H}(\bm{x}) = \mc{P}_n$
		\end{enumerate}
		
		Let us prove the first property.
		By construction, we have $Pg_j - (-1)^{x_j} g_j P = 0$ for any $P\in \mc{H}(\bm{x})$.
		This means that, for $\bm{x}\neq\bm{x}'$, $\mc{H}(\bm{x})$ and $\mc{H}(\bm{x}')$ have different (anti-)commutation relations with $g_1,\cdots,g_s$,
		leading to $\mc{H}(\bm{x}) \cap \mc{H}(\bm{x'})=\varnothing$.
		The second property also holds because $h_1^{x_1}\cdots h_s^{x_s} P \neq h_1^{x_1}\cdots h_s^{x_s} P'$ for $P\neq P'$.
		The third property is proven by contradiction.
		Let us assume that there exists $P\in \mc{P}_n$ satisfying $P\notin \bigsqcup_{\bm{x}} \mc{H}(\bm{x})$.
		Then, a binary vector $\bm{y}$ is defined such that $Pg_j - (-1)^{y_j} g_j P = 0$.
		Because $P'=h_1^{y_1}\cdots h_s^{y_s} P$ satisfies $[P',g_j]=0$ for any $g_j$, $P'$ belongs to $\mc{P}_\mcS$ by definition.
		Hence, we have $P = h_1^{y_1}\cdots h_s^{y_s} P' \in \mc{H}(\bm{y})$.
		This contradicts the assumption of $P\notin \bigsqcup_{\bm{x}} \mc{H}(\bm{x})$.
		Therefore, we have $\bigsqcup_{\bm{x}} \mc{H}(\bm{x}) = \mc{P}_n$.
		
		These properties lead to $4^n=|\mc{P}_n|=|\bigsqcup_{\bm{x}} \mc{H}(\bm{x})|= \sum_{\bm{x}} |\mc{H}(\bm{x})| = 2^s |\mc{P}_\mcS|$, proving $|\mcP_\mcS|=4^n/2^s$.
		Besides, given that $s$ independent Pauli operators can generate only $2^s$ Pauli operators by multiplying them, we have $|\mcS| \leq 2^s$, which readily shows $|\mcP_\mcS|=4^n/2^s \leq 4^n/|\mcS|$.
		Also, combining  $|\mcP_\mcS|\leq 4^n/|\mcS|$ and $|\mcS| \leq |\mcP_\mcS|$ (the latter is derived from the fact that $\forall S_j\in\mcS$ is included in $\mc{P}_\mcS$), we obtain $|\mcS|\leq 2^n$.

	\end{proof}

	\begin{lem} \label{lem: inequality}
		Let $c$ be a real constant.
		For real variables $a_1,\cdots,a_q \geq c$ and a real constant $A \geq c^q$, the following inequality holds under a constraint $\prod_{j=1}^q a_j \leq A$:
		\begin{align}
			\sum_{j=1}^q a_j \leq \frac{A}{c^{q-1}} + c(q-1).
		\end{align}
	\end{lem}    

	\begin{proof}
		We first prove this lemma for $c=1$ by mathematical induction.
		
		\begin{enumerate}
			\item[(i)] For $q=1$, $a_1 \leq A$ holds trivially by the constraint.
			\item[(ii)] We assume that this lemma holds for $q=m$.
			That is, under the constraint $\prod_{j=1}^m a_j \leq A$, $\sum_{j=1}^m a_j \leq A+m-1$ holds for any $A\geq1$.
			
			Now, we consider the case of $q=m+1$, where the constraint is given by $\prod_{j=1}^{m+1} a_j \leq A$.
			Let us fix $a_{m+1}$ in the range of $1\leq a_{m+1} \leq A$ (if $a_{m+1}>A$, the condition $a_1,\cdots,a_m \geq1$ cannot be satisfied).
			Then, the constraint for $a_1,\cdots,a_m$ is written as $\prod_{j=1}^{m} a_j \leq A/a_{m+1}=\tilde{A}$.
			Under this constraint, we have $\sum_{j=1}^m a_j \leq \tilde{A}+m-1=A/a_{m+1}+m-1$ from the assumption for $q=m$.
			Therefore, we have $\sum_{j=1}^{m+1} a_j= a_{m+1} + (\sum_{j=1}^{m} a_j)\leq a_{m+1}+A/a_{m+1}+m-1$.
			For $1\leq x \leq A$, $f(x)=x+A/x$ has a maximum value $f(1)=f(A)=A+1$.
			Thus, we finally obtain $\sum_{j=1}^{m+1} a_j\leq A+m$, indicating that the lemma holds even for $q=m+1$.
		\end{enumerate}
		
		These discussions prove the lemma for $c=1$ by mathematical induction.
		Then, we rescale $a_j$ by a factor of $c$ as $a_j \to a_j'=ca_j$, where the constraint is also rescaled as $\prod_{j=1}^q a_j' = c^q \prod_{j=1}^q a_j \leq c^q A = A'$.
		Therefore, we obtain
		\begin{align}
			\sum_{j=1}^q a_j' = c\sum_{j=1}^q a_j \leq cA + c(q-1) = \frac{A'}{c^{q-1}} + c(q-1), \notag
		\end{align}
		where we have used the lemma for $c=1$ in the inequality.
		This proves the lemma for any $c>0$.
	\end{proof}

	\subsection{Proof of Eq.~\eqref{apeq: ex_inequality}}
	
	Here, we prove Eq.~\eqref{apeq: ex_inequality}:
	\begin{align}
		\ex \leq \frac{4^n v}{4^{q-1}|\mcS|^2} + (3q-4)v + p. 
	\end{align}
	Let us first show the following lemma: 
	\begin{lem} \label{lem: rj}
		$r_\x\geq 3$.
	\end{lem}
	\begin{proof}
		If $r_\x=1$, all the elements of $\mcB_\x=\mcC_\x^1$ commute with each other by Eq.~\eqref{apeq: Cst1}, which contradicts the fact that $\mcB_\x$ is a connected graph, thus $r_\x \geq 2$.
		If $r_\x=2$, Eq.~\eqref{apeq: Cst1} leads to $\{\mcC_\x^{1},\mcC_\x^{2}\}=0$ because $\mcB_\x$ is connected.
		Then, for $P\in\mcC_\x^{1}$ and $Q\in\mcC_\x^{2}$, $iR = [iP,iQ]=-2PQ$ satisfies $\{P,R\}=\{Q,R\}=0$.
		This indicates that $R$ is included in $\mcB_\x$ but not in both $\mcC_\x^{1}$ and $\mcC_\x^{2}$, thus showing the existence of $\mcC_\x^{3}$.
		Therefore $r_\x\geq 3$.
	\end{proof}

	Now, we prove Eq.~\eqref{apeq: ex_inequality}. 
	Let $C_{\x, a} \in \mcC_\x^{a}$ be a representative element of $\mcC_\x^{a}$ and define $\mc{D}_\x=\{I, C_{\x, 1}, \cdots, C_{\x, r_\x}\}$.
	Also, let $S_{a}\in \mcS$ and $D_{\x, a}\in \mc{D}_\x$ be the $a$th elements of $\mcS$ and $\mc{D}_\x$, respectively.
	Then we consider the following Pauli operator:
	\begin{align}
		M_{\bm{a}} = S_{a_0} D_{1, a_1} \cdots D_{q, a_{q}},
	\end{align}
	where $\bm{a}=(a_0,\cdots,a_q)$ is a vector with $a_0\in\{1,\cdots,|\mcS|\}$ and $a_{\x(\neq0)}\in \{1,\cdots,r_\x+1\}$.
	The operator $M_{\bm{a}}$ has the following properties:
	\begin{enumerate}
		\item[(i)]  $M_{\bm{a}} \in\mc{P}_\mcS$,
		\item[(ii)] $M_{\bm{a}} \neq M_{\bm{b}}$ for $\bm{a}\neq \bm{b}$,
	\end{enumerate}
	where $\mc{P}_\mcS = \{P\in \mc{P}_n \,|\, [P,S_j]=0 \,\, \forall S_j\in\mcS\}$ is a subset of Pauli operators stabilized by $\mcS$ ($|\mc{P}_\mcS| \leq 4^n/|\mcS|$ holds by Lemma~\ref{lem: stab1}).
	The first property is readily proven by noticing $[M_{\bm{a}},\mcS]=0$, which is derived from $[\mcS,\mcS]=[\mc{D}_\x,\mcS]=0$.
	Then, we prove the second property below.
	If $a_\x \neq b_\x$ ($\x\geq1$), according to Eq.~\eqref{apeq: Cst2}, there exists $R\in \mcB_\x$ such that $[D_{\x, a_\x},R]=\{D_{\x, b_\x},R\}=0$ (or $\{D_{\x, a_\x},R\}=[D_{\x, b_\x},R]=0$) and $[\mcS,R]=[\mc{D}_{\y(\neq \x)},R]=0$.
	These commutation relations lead to $[M_{\bm{a}}, R] \neq [M_{\bm{b}}, R]$, showing $M_{\bm{a}} \neq M_{\bm{b}}$.
	If $a_0 \neq b_0$ and $a_\x = b_\x$ for all $\x\geq1$, $M_{\bm{a}} = M_{\bm{b}}$ leads to $S_{a_0}=S_{b_0}$, which contradicts the fact that the elements of $\mcS$ are not duplicated, showing $M_{\bm{a}} \neq M_{\bm{b}}$.
	Therefore, $M_{\bm{a}} \neq M_{\bm{b}}$ for $\bm{a}\neq \bm{b}$.

	Given that $M_{\bm{a}} \in\mc{P}_\mcS$ and $M_{\bm{a}} \neq M_{\bm{b}}$ for $\bm{a}\neq \bm{b}$, the following inequality holds:
	\begin{align}
		|\mcS| (r_1+1) \cdots (r_q+1) \leq |\mc{P}_\mcS|, \label{apeq: constraint}
	\end{align}
	where the left-hand side corresponds to the number of $M_{\bm{a}}$.
	Since $|\mc{P}_\mcS|\leq 4^n/|\mcS|$, we have 
	\begin{align}
		\prod_{\x=1}^q (r_\x+1)\leq \frac{4^n}{|\mcS|^2}.  \label{apeq: constraint2}
	\end{align}
	This constraint gives a new upper bound of the expressivity $\ex\leq p+v \sum_{\x=1}^q r_\x$ in Eq.~\eqref{apeq: Dexp_cal0}.
	By minimizing $\sum_{\x=1}^q r_\x$ in $\ex$ under the constraint $\prod_{\x=1}^q (r_\x+1)\leq 4^n/|\mcS|^2$, we obtain
	\begin{align}
		\ex 
		&\leq p+v \sum_{\x=1}^q r_\x \notag \\
		&\leq v\sum_{\x=1}^q (r_\x+1) - qv + p \notag \\
		&\leq v \left[ \frac{4^n}{4^{q-1}|\mcS|^2} + 4(q-1) \right] - qv + p \notag \\
		&= \frac{4^n v}{4^{q-1}|\mcS|^2} + (3q-4)v + p,
	\end{align}
	where we have used Lemma~\ref{lem: inequality} with $r_\x+1\geq 4$ in the derivation of the third line.
	This is Eq.~\eqref{apeq: ex_inequality}, as required.

	For later use, we derive another inequality.
	Combining Eq.~\eqref{apeq: constraint2} and $r_\x+1\geq 4$, we have
	\begin{align}
		\frac{4^n}{|\mcS|^2} \geq 4^q. \label{apeq: condition_S}
	\end{align}
	This inequality will be used below.

	\subsection{Proof of Eq.~\eqref{apeq: sub_inequality}}
	
	Here, we prove Eq.~\eqref{apeq: sub_inequality}:
	\begin{align}
		\frac{4^n v}{4^{q-1}|\mcS|^2} + (3q-4)v + p \leq \frac{4^n}{qw} - qw. \label{apeq: Dexp_cal1}
	\end{align}
	By considering the difference between the both sides, we have
	\begin{align}
		&\text{(RHS)}-\text{(LHS)} \notag \\
		&= \frac{4^n}{4^{q-1}qw|\mcS|^2}(4^{q-1}|\mcS|^2 - qvw) -qw -(3q-4)v - p. \notag
	\end{align}
	This is written as
	\begin{align}
		\text{(RHS)}-\text{(LHS)} 
		&\geq \frac{4^q|\mcS|^2 - 4qvw}{qw} -qw -(3q-4)v - p. \notag \\
		&= \frac{4^q|\mcS|^2 - 3q^2vw -q^2w^2 - pqw}{qw}  \notag \\
		&\geq \frac{4q^2|\mcS|^2 - 3q^2|\mcS|w -q^2w^2 - pq^2w}{qw} \notag \\
		&= \frac{q(4|\mcS|^2 - 3|\mcS|w - w^2 - pw)}{w} \notag \\
		&= \frac{q\Big[ (4|\mcS|+w)(|\mcS|-w)-pw \Big]}{w} \notag \\
		&\geq \frac{q\Big[ (4|\mcS|+w)p-pw \Big]}{w} \notag \\
		&= \frac{4pq|\mcS|}{w} \notag \\
		&\geq 0. \notag
	\end{align}
	Here, we have used Eq.~\eqref{apeq: condition_S} with $4^{q-1}|\mcS|^2 - qvw \geq 0$ in the second line, $4^q\geq4q^2, q^2\geq q$, and $|\mcS|=\text{max}(v,w+p)\geq v$ in the fourth line, and $|\mcS|=\text{max}(v,w+p)\geq w+p$ in the seventh line.
	This result proves Eq.~\eqref{apeq: sub_inequality}.

	\section{Supplementary information for SLPA}
	\subsection{Gradient measurement circuit} \label{secap: CBC measure}
	
	We briefly review how to measure gradients in the commuting block circuit (CBC)~\cite{Bowles2023-vf}, including the stabilizer-logical product ansatz (SLPA).
	As introduced in the main text, the unitary of CBC is given by $U(\bt)=\prod_{a=1}^B U_a(\bt_a)$ with $U_a(\bt_a)=\prod_{j} \exp(i \theta_j^a G_j^a)$.
	For an input state $\ket{\phi}$ and an Pauli observable $O$, the cost function $C(\bt)=\braket{\phi|U^\dag(\bt)OU(\bt)|\phi}$ is reduced to
	\begin{align}
		C(\bt) = \braket{\phi_a | W_a^\dag O W_a | \phi_a},
	\end{align}
	where we have defined the quantum state at the $a$th block and the quantum circuit after the $a$th block as
	\begin{align}
		&\ket{\phi_a} = U_a(\bt_a)\cdots U_1(\bt_1) \ket{\phi}, \\
		&W_a = U_B(\bt_B)\cdots U_{a+1}(\bt_{a+1}).
	\end{align}
	Using $\partial \ket{\phi_a}/\partial \theta_j^a = iG_j^a \ket{\phi_a}$, we have the derivative of the cost function by the $j$th parameter of the $a$th block:
	\begin{align}
		\frac{\partial C}{\partial \theta_j^a} = \braket{\phi_a | \left(i W_a^\dag O W_a G_j^a - i G_j^a W_a^\dag O W_a \right) | \phi_a}.
	\end{align}
	Then, since all generators in the block share the same commutation relations with other blocks, we can define $\tilde{W}_a$ such that $W_a G_j^a = G_j^a \tilde{W}_a$.
	This $\tilde{W}_a$ is easily obtained by using $e^{i\theta P} G_j^a = G_j^a e^{\pm i\theta P}$ ($P$ is a Pauli operator), where $\pm$ correspond to the cases of $[G_j^a,P]=0$ and $\{G_j^a,P\}=0$, respectively.
	Thereby, we have
	\begin{align}
		\frac{\partial C}{\partial \theta_j^a} = \braket{\phi_a | \left(W_a^\dag (-1)^{g_j^a} iG_j^a O \tilde{W}_a - \tilde{W}_a^\dag i G_j^a O W_a \right) | \phi_a},
	\end{align}
	where $g_j^a=0$ if $[G_j^a,O]=0$ and $g_j^a=1$ if $\{G_j^a,O\}=0$.
	By defining a Pauli operator $O_j^a=i^{g_j^a} G_j^a O$ and a unitary operator $W_a'=(-i)^{g_j^a+1}W_a$, the derivative is written as
	\begin{align}
		&\frac{\partial C}{\partial \theta_j^a} \notag \\
		&= \braket{\phi_a | \left( (W_a')^\dag O_j^a  \tilde{W}_a + \tilde{W}_a^\dag O_j^a (W_a') \right) | \phi_a} \notag\\
		&= \frac{1}{2} \Big[ \braket{\phi_a | \left( \tilde{W}_a^\dag + (W_a')^\dag \right) O_j^a \left( \tilde{W}_a + W_a' \right) | \phi_a} \notag \\
		&\hspace{1cm}- \braket{\phi_a | \left( \tilde{W}_a^\dag - (W_a')^\dag \right) O_j^a \left( \tilde{W}_a - W_a' \right) | \phi_a} \Big] \notag\\
		&= \frac{1}{2} \Big[ \braket{\phi_a | (L_{W_a}^+)^\dag O_j^a L_{W_a}^+ | \phi_a} -\braket{\phi_a | (L_{W_a}^-)^\dag O_j^a L_{W_a}^- | \phi_a} \Big],
	\end{align}
	where $L_{W_a}^\pm=\tilde{W}_a \pm W_a'$ are the linear combinations of unitaries.

	\begin{figure}[t]
		\centering
		\includegraphics[width=\linewidth]{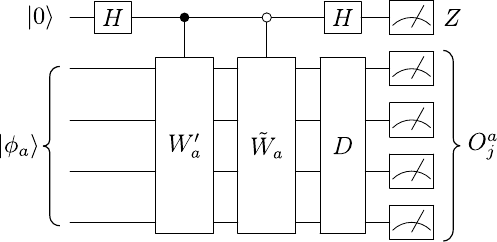}
		\caption{
			{\bf Gradient measurement in commuting block circuits.}
			The gate $D$ is a basis transformation that diagonalizes $O_j^a$'s with the same $g_j^a$ simultaneously. 
		}
		\label{fig: CBC measure}
	\end{figure}

	These $L_{W_a}^\pm$ can be implemented using an ancilla qubit as shown in Fig.~\ref{fig: CBC measure}.
	For the quantum state $\ket{\phi_a}$ and the ancilla qubit $H\ket{0}$ ($H$ is a Hadamard gate), we apply a $W'$ gate controlled on the ancilla qubit being in $\ket{1}$, followed by a $\tilde{W}$ gate controlled on the ancilla qubit being in $\ket{0}$. 
	After these controlled gates, we finally apply a $H$ gate on the ancilla qubit and then have
	\begin{align}
		\ket{\psi_a} = \frac{1}{2} \left( \ket{0}L_{W_a}^+ \ket{\phi_a} + \ket{1}L_{W_a}^- \ket{\phi_a} \right).
	\end{align}
	For this quantum state, we can estimate the derivative by measuring an observable $\tilde{O}_j^a = 2(Z\otimes O_j^a)$ as
	\begin{align}
		&\braket{\psi_a | \tilde{O}_j^a | \psi_a} \notag \\
		&= \frac{1}{2} \left[ \braket{\phi_a | (L_{W_a}^+)^\dag O_j^a L_{W_a}^+ | \phi_a} -  \braket{\phi_a | (L_{W_a}^-)^\dag O_j^a L_{W_a}^- | \phi_a} \right] \notag \\ 
		&= \frac{\partial C}{\partial \theta_j^a}.
	\end{align}
	In this method, we can simultaneously measure multiple derivatives in the same block if their generators share the same commutation relations with the observable $O$ (i.e., the same $g_j^a$).
	This is because the measured observables $\tilde{O}_j^a$ are commutative, $[\tilde{O}_j^a, \tilde{O}_k^a]=0$, when $g_j^a=g_k^a$.
	Therefore, given that the generators of the final block commuting with $O$ do not contribute to the gradient, we can estimate the full gradient with only $2B-1$ types of quantum circuits.

	\subsection{Stabilizer formalism of commuting block circuits} \label{secap: stabilizer formalism}
	
	The SLPA is a general ansatz of CBC, i.e., any CBC can be formulated with stabilizers and logical Pauli operators.
	To show this, we consider a CBC with generators $\mcGc=\{G_j^a\}$ that satisfy the commutation relations of
	\begin{align}
		[G^a_j, G^a_k] = 0 \quad \forall j,k, \label{apeq: CBC sameBlock}
	\end{align}
	and
	\begin{align}
		[G^a_j,G^b_k] = 0  \quad \forall j,k  \quad \text{or} \quad \{G^a_j,G^b_k \} = 0  \quad \forall j,k.
		\label{apeq: CBC difBlock}
	\end{align}
	For these generators, we define stabilizers and logical Pauli operators as 
	\begin{align}
		&\mcS=\{G^a_1G^a_j \}_{j,a}=\{S_{ja}\}_{j,a}, \\
		&\mcL=\{G^a_1\}_{a}=\{L_{a}\}_{a},
	\end{align}
	where we have defined $S_{ja}=G^a_1G^a_j$ and $L_a=G^a_1$.
	These operator sets $\mcS$ and $\mcL$ not only satisfy the requirements of stabilizers and logical Pauli operators but also reproduce the original CBC.
	In other words, the CBC that has generators $\mcGc$ is equivalent to the SLPA constructed from $\mcS$ and $\mcL$.
	In fact, $\mcS$ and $\mcL$ obey the commutation relations of stabilizers and logical Pauli operators as $[S_{ja},S_{kb}]=[G_1^aG_j^a, G_1^bG_k^b] = 0$ and $[S_{ja},L_{b}]=[G_1^aG_j^a, G_1^b] = 0$, where we have used Eqs.~\eqref{apeq: CBC sameBlock} and \eqref{apeq: CBC difBlock}.
	Besides, the generators of CBC can be constructed from $\mcS$ and $\mcL$ by taking their products as $G_j^a =(G_1^a G_j^a)G_1^a = S_{ja}L_a$.
	Therefore, any CBC can be formulated with stabilizers and logical Pauli operators, which implies that the SLPA is a general ansatz of CBC.

	\subsection{Backpropagation scaling in SLPA} \label{secap: backprop SLPA}
	
	The stabilizer formalism of CBC clarifies when it can achieve the backpropagation scaling, which specifies the scalability of learning models with respect to the number of training parameters~\cite{Abbas2023-hy, Bowles2023-vf}.
	The backpropagation scaling is defined as
	\begin{align}
		\frac{\text{Time}(\nabla C)}{\text{Time}(C)} \leq \mO(\log L),
	\end{align}
	where $\text{Time}(C)$ and $\text{Time}(\nabla C)$ are the time complexity of estimating the cost function $C$ and its gradient $\nabla C$ with a certain accuracy.
	Although Ref.~\cite{Bowles2023-vf} proven that the backpropagation scaling is achievable for $B=1$, whether it is possible even for $B\neq 1$ remains unclear.

	In the CBC, the cost function can be estimated with one quantum circuit, whereas $2B-1$ quantum circuits are needed to estimate the gradient.
	Therefore, when ignoring the difference between the single circuit execution times for estimating $C$ and $\nabla C$, we have $\text{Time}(\nabla C)/\text{Time}(C) \sim B$.
	Also, in the stabilizer formalism, the CBC has at most $L=2^s B$ training parameters.
	Therefore, the backpropagation scaling is written as
	\begin{align}
		B \leq \mO(\log B + s\log 2).
	\end{align}
	This indicates that $s$ must increase proportionally to or faster than $B$ to achieve the backpropagation scaling.
	We note that this discussion ignores the circuit execution time, thus requiring a more thorough analysis to understand the precise condition for the backpropagation scaling.

	\section{DLA in numerical experiment}\label{secap: DLA of our ansatz}
	
	We prove the following lemma to understand the DLA in the numerical experiments of the main text:
	\begin{lem}
		Consider a set of generators
		\begin{align}
			\mcGc=\{X_j X_{j+1},Y_j Y_{j+1},Z_j Z_{j+1}\}_{j=1}^n
		\end{align}
		with even $n$. 
		Then, the Lie closure of $\mcGc$ is given by
		\begin{align}
			\mcG = \mcQ_\mcS,
		\end{align}
		where we have defined $\mcQ_\mcS=\mcP_\mcS \smallsetminus \mcS$ with
		\begin{align}
			&\mcS = \{I, \prod_{j=1}^n X_j, \prod_{j=1}^n Y_j, \prod_{j=1}^n Z_j\}, \\
			&\mcP_\mcS = \{P\in \mcPn \, | \, [P,S_j]=0 \,\, \forall S_j\in \mcS\}.
		\end{align}    
		This readily leads to $\text{dim}(\mfg)=|\mcG|=4^{n}/4-4$ by Lemma~\ref{lem: stab1}.
	\end{lem}    

	\begin{proof}
		We prove this lemma by showing $\mcG\subseteq \mcQ_\mcS$ and $\mcQ_\mcS \subseteq\mcG$.

		We first show $\mcG\subseteq \mcQ_\mcS$.
		One can easily verify $\mcG\subseteq \mcP_\mcS$ because all the generators in $\mcGc$ commute with $\mcS$.
		Thus, it suffices to show $\mcG \cap \mcS = \varnothing$.
		We prove this by contradiction.
		Assume $S_j\in\mcS$ is contained in $\mcG$.
		Then, because $S_j\notin \mcGc$, there exist anti-commuting Pauli operators $P,Q\in\mcG$ such that $S_j \propto [P,Q]=2PQ$.
		Then, we have $[S_j,P]=[2PQ,P]\neq 0$, which contradicts the fact that $S_j$ commutes with all the operators in $\mcP_\mcS$ and thus $\mcG$.
		Therefore, $\mcG$ does not contain $\mcS$, proving $\mcG\subseteq \mcQ_\mcS$.

		We then show $\mcQ_\mcS \subseteq\mcG$.
		For convenience, we call the number of $X, Y$, and $Z$ operators in a Pauli string $P$ as the weight of $P$.
		Below, we prove $\mcQ_\mcS \subseteq\mcG$ by mathematical induction with respect to the weight of Pauli strings.
		To this end, let us introduce the subset of $\mcQ_\mcS$ with weight-$w$:
		\begin{align}
			\mcQ_\mcS^w = \{P\in \mcQ_\mcS \, |\, \text{the weight of $P$ is $w$}\}.
		\end{align}
		The sum of $\mcQ_\mcS^w$ gives $\bigsqcup_{w=0}^n \mcQ_\mcS^w = \mcQ_\mcS$.
		For example, we have
		\begin{align}
			&\mcQ_\mcS^0 = \varnothing, \\
			&\mcQ_\mcS^{1} = \varnothing, \\
			&\mcQ_\mcS^2 = \{X_jX_k, Y_jY_k, Z_jZ_k\}_{j,k=1}^n, \\
			& \hspace{1cm}\vdots \notag
		\end{align}
		where we have used $[\mcQ_\mcS^w,\mcS]=0$ and $I\notin \mcQ_\mcS$.
		We also define the numbers of $X$, $Y$, and $Z$ operators in $P\in\mcPn$ as $w_X(P), w_Y(P)$, and $w_Z(P)$, respectively.
		By definition, $w_X(P)+w_Y(P)+w_Z(P)=w$ holds for $\forall P\in \mcQ_\mcS^w$.

		Let us prove $\mcQ_\mcS^w \subset \mcG$ by mathematical induction for $w$.
		Since $\mcQ_\mcS^{0} = \mcQ_\mcS^{1} = \varnothing$, it suffices to consider $w\geq2$.
		We first show that $\mcQ_\mcS^2 \subset \mcG$, i.e., $X_jX_k, Y_jY_k$, and $Z_jZ_k$ for $j<k$ are contained in $\mcG$.
		For example, let us consider $X_jX_k$.
		Because of $X_jX_{j+1}\in \mcGc$, we trivially have $X_jX_{j+1}\in \mcG$.
		Then, since $Y_{j+1}Y_{j+2},Z_{j+1}Z_{j+2}\in\mcG$, a nested commutator $[[X_jX_{j+1}, Y_{j+1}Y_{j+2}],Z_{j+1}Z_{j+2}]\propto X_{j}X_{j+2}$ is also contained in $\mcG$.
		Similarly, we have $[[X_jX_{j+2}, Y_{j+2}Y_{j+3}],Z_{j+2}Z_{j+3}]\propto X_{j}X_{j+3}\in\mcG$.
		Repeating this procedure, we can easily show that $X_jX_k \in \mcG$ for any $j<k$.
		In the same way, we can show $Y_jY_k, Z_jZ_k \in \mcG$ for any $j<k$, proving $\mcQ_\mcS^2 \subset \mcG$.

		Next, we prove $\mcQ_\mcS^w \subset \mcG$ for $3\leq w\leq n-1$ by mathematical induction (we prove the case of $w=n$ later).
		Assume $\mcQ_\mcS^w \subset \mcG$ for $2\leq w \leq n-2$.
		Then, we show that $\forall P\in \mcQ_\mcS^{w+1}$ is contained in $\mcG$ in two cases: (i) two or more of $w_X(P), w_Y(P), w_Z(P)$ are greater than one, and (ii) otherwise.
		\begin{enumerate}
			\item[(i)] 
			The weight-$(w+1)$ Pauli string is written as $P = \sigma^{\mu_1}_{\gamma_1} \cdots \sigma^{\mu_{w+1}}_{\gamma_{w+1}}$, where we have defined
			\begin{align}
				\sigma^\mu_j =
				\begin{cases}
					X_j & \mu=1 \\
					Y_j & \mu=2 \\
					Z_j & \mu=3
				\end{cases}
			\end{align}
			and the qubit indices $\gamma_1<\cdots<\gamma_{w+1}$.
			Given that two or more of $w_X(P), w_Y(P), w_Z(P)$ are greater than one, there necessarily exist $1\leq j<k \leq w+1$ such that $\mu_j \neq \mu_k$.
			Then, we define $Q=iPR$ using a weight-2 Pauli string $R=\sigma^{\mu_{k}}_{\gamma_{j}}\sigma^{\mu_{k}}_{\gamma_{k}}\in\mcQ_\mcS^2$.
			One can easily show that $Q$ is a weight-$w$ Pauli string and is contained in $\mcQ_\mcS^w$ because of $P\in \mcQ_\mcS^{w+1}$ and $R\in\mcQ_\mcS^2$.
			Therefore, we have $Q\in\mcG$ by assumption.
			We can construct $P$ by taking the commutator $P\propto [Q,R]$ (note that $\{Q,R\}=0$), proving $P \subset \mcG$ for the case (i).
			
			\item[(ii)]
			When $P$ contains only one of $X$, $Y$, and $Z$ operators, we can construct $P$ using another qubit $\gamma'$ that is different from $\gamma_1,\cdots,\gamma_{w+1}$.
			For example, let us consider a Pauli string that has only $X$ operators: $P = X_{\gamma_1} \cdots X_{\gamma_{w+1}}\in \mcQ_\mcS^{w+1}$ (then $w+1$ is even because of $[P,\prod_j Y_j]=0$).
			Given that any weight-$(w+1)$ Pauli string with two or more of $X, Y$, and $Z$ operators is contained in $\mcG$ according to the proof of the case (i), $S=X_{\gamma_1} \cdots X_{\gamma_{w-1}} Y_{\gamma_{w}} Y_{\gamma_{w+1}}\in \mcQ_\mcS^{w+1}$ is also contained in $\mcG$.
			Then, we can construct $P$ from $S, Z_{\gamma_{w}} Z_{\gamma'}, Z_{\gamma_{w+1}} Z_{\gamma'}\in\mcG$ as $P\propto [[S,Z_{\gamma_{w}} Z_{\gamma'}],Z_{\gamma_{w+1}} Z_{\gamma'}]$, showing that $P\in \mcG$.
			Similarly, Pauli strings that have only $Y$ or $Z$ operators are also contained in $\mcG$.
			Therefore, $P \subset \mcG$ holds for the case (ii).
		\end{enumerate}
		These results prove that $\mcQ_\mcS^w \subset \mcG$ for $2\leq w \leq n-1$ by mathematical induction.
		Finally, we prove $\mcQ_\mcS^{n} \subset \mcG$.
		Because $\forall P\in \mcQ_\mcS^n$ has at least two of $X$, $Y$, and $Z$ operators (note that $\mcQ_\mcS$ does not contain $\mcS$), we can construct $P$ from weight-$(n-1)$ and weight-2 Pauli strings $Q$ and $R$ in a similar way to the above discussion on the case (i).
		
		In summary, we show $\mcQ_\mcS^w \subset \mcG$ for any $w$, proving $\mcQ_\mcS = \bigsqcup_{w=0}^n \mcQ_\mcS^w \subseteq \mcG$.
		Therefore, together with $\mcG\subseteq \mcQ_\mcS$, we have proven this lemma.
	\end{proof}

	\section{Additional numerical results}
	
	This section provides additional numerical results to show the validity of the efficiency-expressivity trade-off and the wide applicability of SLPA.
	
	\subsection{Gradient measurement efficiency in disentangled ansatz} \label{secap: disentangled circuit}
	
	\begin{figure}[t]
		\centering
		\includegraphics[width=\linewidth]{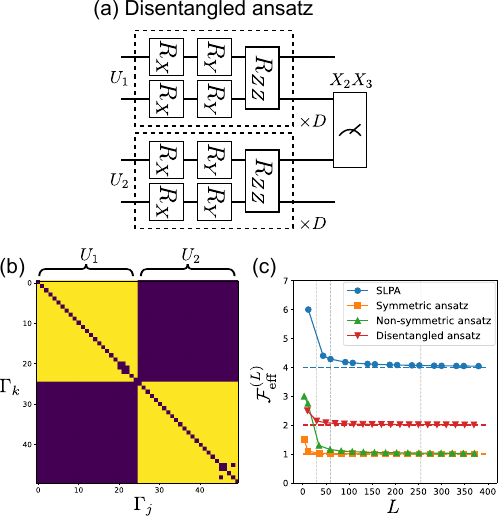}
		\caption{
			{\bf Gradient measurement efficiency in disentangled ansatz.}
			(a) The structure of the disentangled ansatz.
			The number of parameters is $L=10D$.
			(b) Commutators between two gradient operators $\grad_j(\bt)$ and $\grad_k(\bt)$ in the disentangled ansatz for $L=50$ parameters.
			The black and yellow regions represent $[\grad_j(\bt),\grad_k(\bt)]=0$ and $[\grad_j(\bt),\grad_k(\bt)]\neq 0$ for random $\bt$, respectively.
			We have rearranged the rows and columns such that the first (last) $L/2$ gradient operators correspond to those of $U_1$ ($U_2$).
			(c) Changes in gradient measurement efficiency when the number of parameters $L$ is varied.
			Their values are computed by minimizing the number of simultaneously measurable sets of $\grad_j(\bt)$'s for random $\bt$.
			The blue circles, orange squares, green triangles, and red inverted triangles are the results of SLPA and symmetric/non-symmetric/disentangled ansatzes, approaching four, one, and two in the limit of $L\to\infty$, respectively.
		}
		\label{fig: disentangled_circuit}
	\end{figure}

	Besides the quantum circuits investigated in the main text (SLPA and symmetric/non-symmetric ansatzes), we explore a disentangled ansatz to show the validity of the efficiency-expressivity trade-off.
	Let us consider the following disentangled quantum ansatz on $n=4$ qubits [see Fig.~\ref{fig: disentangled_circuit} (a)]:
	\begin{align}
		U_{\rm DE}(\bt) = U_1(\bt_1)\otimes U_2(\bt_2),
	\end{align}
	where $U_1(\bt_1)$ and $U_2(\bt_2)$ are parameterized quantum circuits acting on the first $n/2$ qubits and the remaining $n/2$ qubits, respectively, defined as
	\begin{align}
		&U_1(\bt_1) = \prod_{\ell=1}^D R_{Z_1Z_2}(\theta_{1,5}^{\ell})R_{Y_2}(\theta_{1,4}^{\ell})R_{Y_1}(\theta_{1,3}^{\ell}) \notag\\
		&\hspace{2cm}\times R_{X_2}(\theta_{1,2}^{\ell})R_{X_1}(\theta_{1,1}^{\ell}), \label{ap: DE_U1} \\
		&U_2(\bt_2) = \prod_{\ell=1}^D R_{Z_3Z_4}(\theta_{2,5}^{\ell})R_{Y_4}(\theta_{2,4}^{\ell})R_{Y_3}(\theta_{2,3}^{\ell}) \notag\\
		&\hspace{2cm}\times R_{X_4}(\theta_{2,2}^{\ell})R_{X_3}(\theta_{2,1}^{\ell}). \label{ap: DE_U2}
	\end{align}
	We also set the measurement observable as $O=X_2X_3$.
	
	The remarkable properties of this ansatz are low expressivity and the absence of stabilizer-type symmetry.
	As for the expressivity, since $U_{\rm DE}$ is disentangled to $U_1$ and $U_2$, the DLA of this ansatz is written as $\mfg = \mfg_1 \oplus \mfg_2$, where $\mfg_1$ and $\mfg_2$ are the DLA of $U_1$ and $U_2$.
	For the specific form of the ansatz given in Eqs.~\eqref{ap: DE_U1} and \eqref{ap: DE_U2}, we have $\mfg_1=\mfg_2=\mathfrak{su}(2^{n/2})$~\cite{Larocca2022-so}.
	Therefore, the expressivity of $U_{\rm DE}$ is $\ex=2\times\text{dim}[\mathfrak{su}(2^{n/2})]=30$, leading to the upper bound of gradient measurement efficiency $\ef\leq 6$ according to the trade-off $\ex\leq 4^n/\ef - \ef$.
	Meanwhile, this ansatz has no stabilizer-type symmetry. 
	That is, for any Pauli operator $P\in \mcPn \smallsetminus \{I\}$, there exists $\bt$ such that $[U_{\rm DE}(\bt),P]\neq 0$, which is derived from $\mfg = \mathfrak{su}(2^{n/2}) \oplus \mathfrak{su}(2^{n/2})$.
	In summary, unlike the SLPA and symmetric/non-symmetric ansatzes investigated in the main text, $U_{\rm DE}$ is an example of parameterized quantum circuits with low expressivity despite the absence of stabilizer-type symmetry.

	Figure~\ref{fig: disentangled_circuit} (b) depicts the commutation relation between the gradient operators $\grad_j$ and $\grad_k$ for $L=50$ parameters, where the black and yellow regions represent $[\grad_j,\grad_k]=0$ and $[\grad_j,\grad_k]\neq0$, respectively.
	Here, to improve readability, we have rearranged the rows and columns of Fig.~\ref{fig: disentangled_circuit} (b) such that the first (last) $L/2$ gradient operators correspond to those of $U_1$ ($U_2$).
	As shown in the figure, the gradient operators of $U_1$ commute with all gradient operators of $U_2$, and vice versa, implying that the gradient components of $U_1$ and $U_2$ can be measured simultaneously.
	This result is in agreement with Lemma~\ref{thm: not connected}, where the disentanglement of $U_1$ and $U_2$ leads to the separability of $\mfg_1$ and $\mfg_2$ (see also Definition~\ref{def: DLA-separability}).
	On the other hand, most pairs of the gradient operators within $U_1$ or $U_2$ are not commutative, indicating that we cannot measure two or more gradient components within $U_1$ or $U_2$ simultaneously.
	These observations imply that the gradient measurement efficiency, i.e., the mean number of simultaneously measurable gradient components, is almost two in the disentangled ansatz.
	Therefore, this ansatz is an example in which $\ef>1$ is due to a mechanism other than stabilizer-type symmetry.

	We investigate the gradient measurement efficiency in more detail in Fig.~\ref{fig: disentangled_circuit} (c), where $\ef^{(L)}$ is plotted with varying the number of parameters $L$. 
	As discussed above, the gradient measurement efficiency of the disentangled ansatz (the red inverted triangles) approaches $\ef=2$ in the $L\to\infty$ limit, which is less than the upper bound $\ef\leq6$ derived from the efficiency-expressivity trade-off.
	Therefore, this result supports the validity of our theoretical predictions even for circuits without stabilizer-type symmetry.

	\subsection{Trainability of SLPA} \label{secap: barren plateau}

	\begin{figure*}[t]
		\centering
		\includegraphics[width=\linewidth]{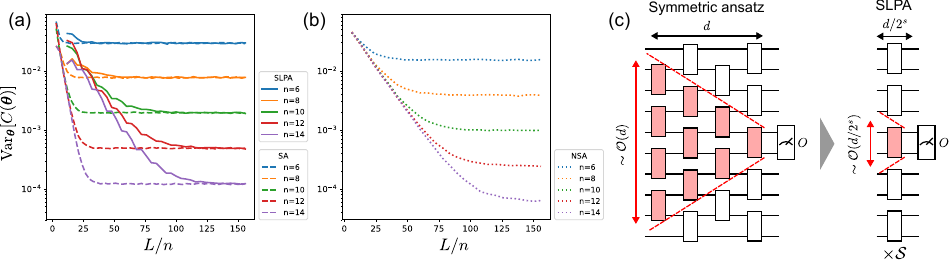}
		\caption{
			{\bf Barren plateaus in SLPA.}
			(a)--(b) The variances of the cost function are plotted for several numbers of qubits $n$ and numbers of parameters $L$.
			The variance is computed by sampling 10000 parameter points from the uniform distribution $\bt\in [-\pi/2,\pi/2]^L$.
			The solid, dashed, and dotted lines indicate the results of the SLPA, symmetric ansatz (SA), and non-symmetric ansatz (NSA), respectively.
			The initial state is $\ket{0}^{\otimes n}$, and the measurement observable is $O=X_1X_2$.
			(c) Adjoint action of the local symmetric ansatz on the local observable.
			In the SLPA, given that the total number of Pauli rotation gates is the same as the symmetric ansatz, the circuit depth is effectively reduced by a factor of $2^s$, leading to high trainability.
		}
		\label{fig: barren_plateau}
	\end{figure*}

	The barren plateau is a phenomenon in which the cost function landscape $C(\bt)$ becomes flat exponentially in the number of qubits $n$~\cite{Larocca2024-vh}.
	Formally, it is defined as
	\begin{align}
		\text{Var}_{\bt} [C(\bt)] \in \mO(b^n)
	\end{align}
	for some $0<b<1$.
	When this phenomenon occurs, optimizing the QNN requires the estimation of the cost function (or its gradient) with exponential accuracy, resulting in an exponentially high measurement cost.
	This is the most critical obstacle to achieving exponential quantum speedups in variational quantum algorithms.
	The barren plateau occurs due to various factors, such as high circuit expressivity~\cite{McClean2018-qf}, strong entanglement~\cite{Ortiz-Marrero2021-lx}, global measurement observables~\cite{Cerezo2021-tq}, and hardware noise~\cite{Wang2021-xp}.
	In particular, it is known that a parameterized quantum circuit consisting of local quantum gates and a local observable (e.g., hardware-efficient ansatz) exhibits the barren plateau of $\text{Var}_{\bt} [C(\bt)]\sim 1/\text{dim}(\mfg)$ if the circuit depth scales as $\Omega(n)$~\cite{McClean2018-qf, Ragone2024-hl, Fontana2024-ky}.
	Understanding the barren plateau is essential to realize practical QNNs with high trainability.

	Here, we investigate the barren plateau of the SLPA.
	A potential concern is that the (potentially global) multi-Pauli rotations within the SLPA could lead to a barren plateau even in shallow circuits.
	However, this concern can be ignored; the global operators do not directly induce a barren plateau in the SLPA. 
	Rather, our model is less prone to the barren plateau phenomenon than the symmetric ansatz when the circuit is shallow.

	First of all, numerical evidence is presented in Figs.~\ref{fig: barren_plateau} (a) and (b), which shows the variance of the cost function in the parameter space for the SLPA and symmetric and non-symmetric ansatzes used in the main text with varying the number of qubits $n$ and the number of parameters $L$.
	In the symmetric ansatz, $\text{Var}_{\bt}[C(\bt)]$ decreases as the circuit depth and finally converges to a finite value at a depth of $L/n \sim \mO(n)$.
	The converged value of $\text{Var}_{\bt}[C(\bt)]$ becomes exponentially small as the number of qubits $n$ increases, which indicates that a barren plateau occurs in the deep symmetric ansatz.
	In the SLPA, we observe a similar behavior, where $\text{Var}_{\bt}[C(\bt)]$ decreases as $L/n$ and finally converges to almost the same value as the symmetric ansatz.
	However, the convergence of $\text{Var}_{\bt}[C(\bt)]$ is slower than that of the symmetric ansatz.
	This result implies the higher trainability of our model in shallow circuits.
	Similarly, the deep non-symmetric ansatz also exhibits a barren plateau, where the variance converges to a smaller value than those of the SLPA and the symmetric ansatz due to the higher expressivity.

	We provide an intuitive mechanism for the slow convergence of the variance of the cost function in the SLPA compared to the symmetric ansatz.
	To this end, we follow the argument in Ref.~\cite{Cerezo2023-hz}.
	Let us consider the observable in the Heisenberg picture $\tilde{O}(\bt)= U^\dag(\bt)O U(\bt)$, where the cost function is written as the Hilbert–Schmidt inner product of the input state $\rho$ and the observable $\tilde{O}(\bt)$: $C(\bt)=\tr [ \rho \tilde{O}(\bt) ]$.
	When $\tilde{O}(\bt)$ lives in an exponentially large subspace, the inner product $\tr [ \rho \tilde{O}(\bt) ]$ will be exponentially small due to the curse of dimensionality, which is the fundamental cause of barren plateaus.
	To quantify this mechanism, let $\mcB$ be the operator subspace where $\tilde{O}(\bt)$ lives.
	That is, $\tilde{O}(\bt)$ can be expanded as $\tilde{O}(\bt) = \sum_{P \in \mcB} \alpha_P(\bt) P$ with Pauli operators $P$.
	Then, given that the cost function is the inner product of $\tilde{O}(\bt)$ and $\rho$ in the subspace $\mcB$, the variance of the cost function decreases inversely proportional to the dimension of $\mcB$: $\text{Var}_{\bt}[C(\bt)]\sim 1/\text{poly}(\text{dim}(\mcB))$.
	For example, an exponentially large $\text{dim}(\mcB)$ generally results in an exponentially small variance of the cost function.
	Therefore, $\text{dim}(\mcB)$ gives an insight into barren plateaus.

	We investigate $\text{dim}(\mcB)$ of the symmetric ansatz and the SLPA to show that our model has smaller $\text{dim}(\mcB)$ when the stabilizers commute with the observable $O$.
	Suppose that $\tilde{O}(\bt)$ in the symmetric ansatz is expanded as
	\begin{align}
		\tilde{O}_\text{SA}(\bt) = \sum_{P \in \mcB_\text{SA}} \alpha_P(\bt) P,
	\end{align}
	where $P$ is a Pauli operator and $\alpha_P(\bt)$ is a coefficient.
	In the SLPA, given that the stabilizers $\mcS=\{S_1,S_2,\cdots\}$ are closed under products (i.e., $S_jS_k \in \mcS$), we have
	\begin{align}
		\tilde{O}_\text{SLPA}(\bt) 
		&= \sum_{S\in \mcS} \sum_{P \in \mcB_\text{SA}} \beta_{P,S}(\bt) PS, \label{apeq: Ocbc}
	\end{align}
	where $\beta_{P,S}(\bt)$ is a coefficient.
	One can verify Eq.~\eqref{apeq: Ocbc} using $[S_j,S_k]=[S_j,L_k]=[S_j,O]=0$ and $e^{-i\theta Q} P e^{i\theta Q}= (\cos^2 \theta \pm \sin^2 \theta)P - i(\cos\theta\sin\theta \mp \cos\theta\sin\theta)PQ$ iteratively in $\tilde{O}(\bt)$ ($P,Q$ are Pauli operators, and $\pm$ corresponds to $[P,Q]=0$ and $\{P,Q\}=0$ respectively).
	Equation~\eqref{apeq: Ocbc} leads to the operator subspace of SLPA as $\mcB_\text{SLPA} = \mcB_\text{SA}\times \mcS$.

	This discussion gives the relationship between the dimensions of the operator subspaces in the symmetric ansatz and the SLPA: 
	\begin{align}
		\text{dim}(\mcB_\text{SLPA}) \leq 2^s \text{dim}(\mcB_\text{SA}).
	\end{align}
	Therefore, if $s=\mO(\text{polylog}(n))$, the SLPA never exponentially spoils the trainability of the symmetric ansatz. 
	For example, let us consider a local quantum circuit.
	When the observable $O$ is local~\cite{Cerezo2021-tq}, since $\tilde{O}(\bt)$ acts only on $\mO(d)$ qubits inside the backward light cone as shown in Fig.~\ref{fig: barren_plateau} (c) ($d$ is the circuit depth), we have
	\begin{align}
		\text{dim}(\mcB_\text{SA})\sim \mO(4^{d})  \label{apeq: dimBsc}
	\end{align}
	and thus
	\begin{align}
		\text{dim}(\mcB_\text{SLPA})\lesssim \mO(4^{d+s/2}). \label{apeq: dimBcbc}
	\end{align}
	If $d$ and $s$ are $\mO(\text{polylog}(n))$, we obtain $\text{dim}(\mcB_\text{SLPA})\in \mO(\text{poly}(n))$, indicating that the SLPA does not suffer from barren plateaus.

	We note that the above discussion does not take into account the differences in the number of rotation gates between $U_\text{SA}$ and $U_\text{SLPA}$.
	To fairly compare $\text{dim}(\mcB)$ of the symmetric ansatz and the SLPA, we should match the total number of rotation gates for these two models.
	In the SLPA, replacing $d$ with $d/2^s$ in Eq.~\eqref{apeq: dimBcbc} to match the total number of gates [see Fig.~\ref{fig: barren_plateau} (c)], we have
	\begin{align}
		\text{dim}(\mcB_\text{SLPA})\lesssim \mO(4^{d/2^s+s/2}). \label{apeq: dimBcbc2}
	\end{align}
	Comparing Eq.~\eqref{apeq: dimBcbc2} with Eq.~\eqref{apeq: dimBsc}, $\text{dim}(\mcB_\text{SLPA})$ can be smaller than $\text{dim}(\mcB_\text{SA})$ when $d/2^s + s/2\lesssim d$, which suggests that the SLPA can have a larger variance of the cost function than the symmetric ansatz with the same number of rotation gates.  
	We emphasize that since this result does not rely on the details of the model, except for the localities of $U_\text{SA}$ and $O$ and some commutation relations, we would observe similar behaviors in other SLPAs.
	Finally, in sufficiently deep circuits, $\text{Var}_{\bt}[C(\bt)]$ converges to the same value in both models because the variance of the cost function in the deep circuit limit is determined by the expressivity~\cite{Larocca2022-so, Ragone2024-hl, Fontana2024-ky}, $\text{Var}_{\bt}[C(\bt)]\sim 1/\ex$, and these two models have the same $\ex=4^n/4 - 4$.

	\subsection{Application to quantum phase recognition} \label{secap: QPR}

	Here, we numerically demonstrate that the SLPA exhibits high training efficiency and classification accuracy in a quantum phase recognition task without an explicit symmetric target function.

	Let us consider the one-dimensional spin-$1/2$ bond-alternating XXZ model on an $n$-qubit system.
	The Hamiltonian is given by
	\begin{align}
		H(J) 
		&= \sum_{j=1}^{n/2} \left(X_{2j-1}X_{2j} + Y_{2j-1}Y_{2j} + \delta Z_{2j-1}Z_{2j}\right) \notag \\
		&+ J\sum_{j=1}^{n/2-1} \left(X_{2j}X_{2j+1} + Y_{2j}Y_{2j+1} + \delta Z_{2j}Z_{2j+1}\right)
	\end{align}
	with $J,\delta \in \mathbb{R}$.
	We fix $n=8$ and $\delta=0.5$ in this work.
	For $\delta=0.5$, the ground state of this Hamiltonian, $\ket{g(J)}$, exhibits the trivial phase for $J\leq J_c$ and the symmetry-protected topological (SPT) phase for $J\geq J_c$, where $J_c\sim 1$~\cite{Elben2020-ip}.
	This Hamiltonian commutes with all $S_j\in \mcS$,
	\begin{align}
		[H(J),S_j]=0,
	\end{align}
	where $\mcS$ is a stabilizer group:
	\begin{align}
		&\mcS=\left\{ I, \prod_{j=1}^n X_j, \prod_{j=1}^n Y_j, \prod_{j=1}^n Z_j \right\}.
	\end{align}
	Hence, the ground state of $H(J)$ is an eigenstate of $S_j$: $S_j \ket{g(J)} \propto \ket{g(J)}$.

	Here, we assume that the ground state is disturbed by local noise as 
	\begin{align}
		\ket{\tilde{g}(J;\bm{\alpha},\bm{\beta})}=R(\bm{\alpha},\bm{\beta})\ket{g(J)},
	\end{align}
	where $R(\bm{\alpha},\bm{\beta})=\prod_{j=1}^n \exp(i \beta_j Y_j) \exp(i \alpha_j X_j)$.
	Let $\mc{D}(J)$ be the data distribution of $\ket{\tilde{g}(J;\bm{\alpha},\bm{\beta})}$ in which $J$ is fixed and $\alpha_j,\beta_j$ are sampled from the normal distribution with zero mean and $\pi/10$ standard deviation. 
	Also, let $\mc{D}$ be the mixture distribution of $\mc{D}(J)$ in which $J$ is sampled from the uniform distribution of $[0,2]$.
	The task here is classifying the noisy ground state into the trivial and SPT phases, i.e., classifying which distribution an unseen data was sampled from, $[\mc{D}(J)|J\leq J_c]$ or $[\mc{D}(J)|J\geq J_c]$.
	We use $N$ training data $\{ \ket{\phi_i},y_i \}_{i=1}^{N}$ to learn the quantum phases, where $\ket{\phi_i}$ is sampled from $\mc{D}$ and $y_i$ is the corresponding label defined as $y_i=1$ ($y_i=0$) for the trivial (SPT) phase.
	We also use $M$ test data sampled from the same distribution $\mc{D}$ to validate the accuracy.
	We set $N=200$ and $M=500$, respectively.

	The data distribution $\mc{D}(J)$ is invariant under the action of $\mcS$.
	To see this, we consider the action of $S_j\in \mcS$ on a data $\ket{\tilde{g}(J;\bm{\alpha},\bm{\beta})}$:
	\begin{align*}
		S_j \ket{\tilde{g}(J;\bm{\alpha},\bm{\beta})} 
		&= S_j R(\bm{\alpha},\bm{\beta}) \ket{g(J)} \\
		&= R(\pm\bm{\alpha},\pm\bm{\beta}) S_j \ket{g(J)} \\
		&\propto R(\pm\bm{\alpha},\pm\bm{\beta}) \ket{g(J)} \\
		&= \ket{\tilde{g}(J;\pm\bm{\alpha},\pm\bm{\beta})},
	\end{align*}
	where we have used $S_j R(\bm{\alpha},\bm{\beta}) = R(\pm\bm{\alpha},\pm\bm{\beta}) S_j$ in the second line (the signs $\pm$ depend on $S_j$) and $S_j \ket{g(J)} \propto \ket{g(J)}$ in the third line.
	Therefore, given that $\bm{\alpha}$ and $\bm{\beta}$ follows the normal distribution with zero mean, the probabilities of sampling $\ket{\tilde{g}(J;\bm{\alpha},\bm{\beta})}$ and $S_j \ket{\tilde{g}(J;\bm{\alpha},\bm{\beta})}$ from $\mc{D}(J)$ are identical.
	In general, the following invariance holds for any $\ket{\phi}$:
	\begin{align*}
		\text{Prob}[\ket{\phi} \sim \mc{D}(J)] 
		= \text{Prob}[S_j\ket{\phi} \sim \mc{D}(J)], \label{eq: qpr_inv}
	\end{align*}
	where $\text{Prob}[\ket{\phi} \sim \mc{D}(J)]$ is the probability density of sampling $\ket{\phi}$ from $\mc{D}(J)$.
	In light of this invariance, an equivariant model incorporates the symmetry into the circuit, assuming $p_{\bt}^i(\ket{\phi}) = p_{\bt}^i(S_j\ket{\phi})$. 
	Here, $p_{\bt}^i(\ket{\phi})$ represents the probability, predicted by a QNN with $U(\bt)$, that $\ket{\phi}$ belongs to the trivial ($i=1$) or the SPT ($i=2$) phase.
	In the following, we numerically demonstrate that this symmetry encoding would help improve trainability and generalization.
	Note that this is the invariance of the data distribution rather than the target function to learn.
	Therefore, this task has a different type of symmetry from the symmetric function learning investigated in the main text.

	To solve this task, we employ the same parameterized quantum circuits as those used in the main text: the SLPA and the symmetric and non-symmetric ansatzes (see Methods for details).
	We fix the number of training parameters as $L=768$ for all the models.
	The SLPA and the symmetric ansatz are symmetric under the stabilizer group $\mcS$, which would help learn the data class in $\mc{D}$.
	To estimate the gradient, we use the parameter-shift method for the symmetric and non-symmetric ansatzes and the linear combination of unitaries with an ancilla qubit for the SLPA,  where $1000$ measurement shots are used per circuit.

	\begin{figure*}[t]
		\centering
		\includegraphics[width=\linewidth]{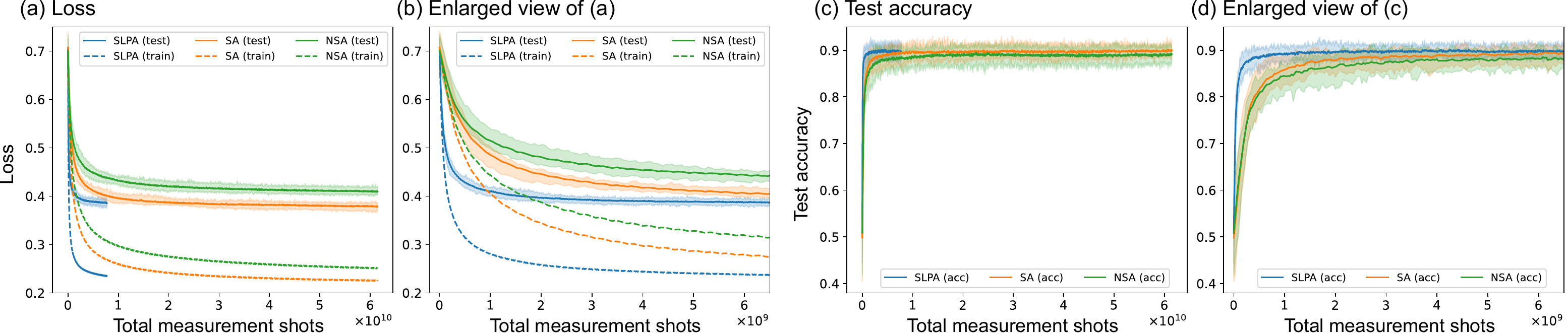}
		\caption{
			{\bf Results of quantum phase recognition.}
			Changes in (a) training and test losses and (c) test accuracy.
			(b) and (d) are the enlarged views of (a) and (c), respectively.
			The horizontal axis is the cumulative number of measurement shots.
			The blue, orange, and green lines represent the results for the SLPA, symmetric ansatz (SA), and non-symmetric ansatz (NSA), respectively.
			The shaded areas are the maximum and minimum of the test loss and accuracy for 20 sets of random initial parameters.
			The numbers of qubits and parameters are $n=8$ and $L=768$.
		}
		\label{fig: qpr}
	\end{figure*}

	Based on the outputs from the quantum circuits, the logistic regression classifies the quantum phases.
	To this end, we define the probabilities that $\ket{\phi}$ is in the trivial and SPT phases as
	\begin{align}
		&p^1_{\bt}(\ket{\phi}) = \frac{1}{1 + \exp\left(- h_{\bt}(\ket{\phi}) \right)}, \\
		&p^2_{\bt}(\ket{\phi}) = \frac{1}{1 + \exp\left( h_{\bt}(\ket{\phi}) \right)},
	\end{align}
	where $h_{\bt}(\ket{\phi})=\gamma\braket{\phi | U^\dag(\bt) O U(\bt) | \phi}$ is the expectation value of the observable $O=X_1X_2$ for a time-evolved input state $U(\bt) \ket{\phi}$.
	We set $\gamma=5$.
	Note that $p_{\bt}^i$ satisfies the positivity and the conservation of probabilities: $p_{\bt}^i\geq0$ and $p_{\bt}^1+p_{\bt}^2=1$.
	In the SLPA and symmetric ansatz, $[U(\bt),\mcS]=[O,\mcS]=0$ ensures the invariance of $p_{\bt}^i(\ket{\phi})=p_{\bt}^i(S_j\ket{\phi})$.
	To train the models, we use the following cross entropy as a loss function:
	\begin{align}
		L(\bt) 
		= -\frac{1}{N} &\sum_{k=1}^{N} \big[ y_k \log p_{\bt}^1(\ket{\tilde{g}(J_i;\mathcal{\alpha}_i,\mathcal{\beta}_i)}) \notag \\
		&+ (1-y_k) \log p_{\bt}^2(\ket{\tilde{g}(J_i;\mathcal{\alpha}_i,\mathcal{\beta}_i)}) \big].
	\end{align}
	We optimize this loss function using the Adam algorithm~\cite{Kingma2014-db}, where the hyper-parameter values are initial learning rate $=10^{-3}$, $\beta_1=0.9$, $\beta_2=0.999$, and $\eta=10^{-8}$.
	We also employ the stochastic gradient descent~\cite{Robbins1951-ql}, where only one training data is used to measure the gradient at each iteration.

	Figure~\ref{fig: qpr} shows the numerical result of the changes in training loss, test loss, and test accuracy during training.
	As expected, the SLPA exhibits faster convergence of training curves than the other two models in terms of the cumulative number of measurement shots.
	This is because four gradient components can be measured simultaneously in the SLPA ($\ef=4$), whereas only one gradient component can be measured at a time in the symmetric and non-symmetric ansatzes ($\ef=1$).
	Furthermore, the parameter-shift method used in the symmetric and non-symmetric ansatzes needs twice the number of circuits for estimating a gradient component than the linear combination of unitaries used in the SLPA.
	Consequently, the SLPA requires only one-eighth the number of measurement shots per epoch compared to other models, resulting in a significant reduction in training sample complexity.

	Besides high training efficiency, the SLPA (and the symmetric ansatz) demonstrates superior generalization performance compared to the non-symmetric ansatz, exhibiting lower test loss and higher test accuracy in Fig.~\ref{fig: qpr}. 
	This suggests that incorporating the symmetry of the data distribution into the circuit improves generalization performance by imposing reasonable constraints and reducing the model space without sacrificing accuracy. 
	These results support the high training efficiency and generalization capability of SLPA in this task, indicating its potential for broad applications to various problems with symmetry.

	
	%

\end{document}